\documentclass[11pt,reqno]{amsart}

\usepackage[utf8]{inputenc}

\usepackage{amsmath,amsthm,amsfonts,amssymb}
\usepackage[hidelinks]{hyperref}
\usepackage[numbers]{natbib}
\usepackage{fourier}
\usepackage{setspace}
\usepackage[T1]{fontenc}
\usepackage[english]{babel}
\usepackage{amsmath,amssymb,amsfonts,amsthm}
\usepackage{mathtools}
\mathtoolsset{showonlyrefs}
\usepackage[pdftex]{graphicx}
\usepackage{color}
\renewenvironment{align}{\begin{equation}\begin{aligned}}{\end{aligned}\end{equation}}
\renewenvironment{align*}{\begin{equation*}\begin{aligned}}{\end{aligned}\end{equation*}}

\numberwithin{equation}{section}
\usepackage{enumerate}

\calclayout

\let\oldtocsection=\tocsection

\let\oldtocsubsection=\tocsubsection

\let\oldtocsubsubsection=\tocsubsubsection

\renewcommand{\tocsection}[2]{\hspace{0em}\oldtocsection{#1}{#2}}
\renewcommand{\tocsubsection}[2]{\hspace{1em}\oldtocsubsection{#1}{#2}}
\renewcommand{\tocsubsubsection}[2]{\hspace{2em}\oldtocsubsubsection{#1}{#2}}

\DeclareMathOperator{\spann}{span}
\DeclareMathOperator{\ad}{ad}
\DeclareMathOperator{\diag}{diag}

\DeclareMathOperator{\supp}{supp}
\DeclareMathOperator{\Op}{Op}

\newcommand{\de}{\, \mathrm{d}}
\newcommand{\del}{\partial}
\newcommand{\R}{\mathbb R}
\newcommand{\C}{\mathbb C}
\newcommand{\N}{\mathbb N}
\newcommand{\BBS}{\mathbb S}

\newcommand{\resc}{s}
\newcommand{\rco}{\chi}
\newcommand{\pd}{uncoupleable}
\newcommand{\mpu}{competely uncoupleable}

\newcommand{\CA}{\mathcal{A}}
\newcommand{\CB}{\mathcal{B}}
\newcommand{\CC}{\mathcal{C}}
\newcommand{\CD}{\mathcal{D}}
\newcommand{\CE}{\mathcal{E}}

\newcommand{\CG}{\mathcal{G}}

\newcommand{\CI}{\mathcal{I}}
\newcommand{\CJ}{\mathcal{J}}
\newcommand{\CK}{\mathcal{K}}
\newcommand{\CL}{\mathcal{L}}
\newcommand{\CM}{\mathcal{M}}
\newcommand{\CN}{\mathcal{N}}
\newcommand{\CO}{\mathcal{O}}

\newcommand{\CR}{\mathcal{R}}
\newcommand{\CS}{\mathcal{S}}
\newcommand{\CT}{\mathcal{T}}
\newcommand{\CU}{\mathcal{U}}
\newcommand{\CV}{\mathcal{V}}

\newcommand{\CX}{\mathcal{X}}
\newcommand{\CY}{\mathcal{Y}}
\newcommand{\CZ}{\mathcal{Z}}
\newcommand{\COD}{{\CO\CD}}
\newcommand{\CNR}{{\CN\!\CR}}

\newcommand{\RB}{\mathrm{B}}
\newcommand{\RC}{\mathrm{C}}
\newcommand{\RD}{\mathrm{D}}

\newcommand{\RG}{\mathrm{G}}
\newcommand{\RH}{\mathrm{H}}

\newcommand{\RL}{\mathrm{L}}

\newtheorem{thm}{Theorem}[section]
\newtheorem{lem}[thm]{Lemma}
\newtheorem{prop}[thm]{Proposition}
\newtheorem{cor}[thm]{Corollary}

\theoremstyle{definition}
\newtheorem*{convention}{Convention}
\newtheorem{defi}[thm]{Definition}
\newtheorem{rem}[thm]{Remark}
\newtheorem{cond}{Condition}
\newtheorem{exa}[thm]{Example}

\newcommand{\bee}{\begin{equation}}
\newcommand{\ene}{\end{equation}}
\newcommand{\bees}{\begin{equation*}}
\newcommand{\enes}{\end{equation*}}
\newcommand{\bes}{\begin{split}}
\newcommand{\ens}{\end{split}}

\newcommand{\bet}{\begin{thm}}
\newcommand{\ent}{\end{thm}}
\newcommand{\bel}{\begin{lem}}
\newcommand{\enl}{\end{lem}}
\newcommand{\bec}{\begin{cor}}
\newcommand{\enc}{\end{cor}}
\newcommand{\becl}{\begin{cla}}
\newcommand{\encl}{\end{cla}}
\newcommand{\bep}{\begin{proof}}
\newcommand{\enp}{\end{proof}}
\newcommand{\ber}{\begin{rem}}
\newcommand{\enr}{\end{rem}}

\newcommand{\eps}{\varepsilon}

\newcommand{\Z}{\mathbb Z}

\newcommand{\btheta}{\boldsymbol \theta}
\newcommand{\bzeta}{\boldsymbol \zeta}

\newcommand{\bxi}{\boldsymbol \xi}
\newcommand{\bmu}{\boldsymbol \mu}
\newcommand{\bnu}{\boldsymbol \nu}
\newcommand{\btau}{\boldsymbol \tau}
\newcommand{\boeta}{\boldsymbol \eta}

\newcommand{\bone}{\boldsymbol 1}
\newcommand{\bzero}{\boldsymbol 0}
\newcommand {\ba}{\mathbf a}

\newcommand {\bg}{\mathbf g}

\newcommand {\bb}{\mathbf b}

\newcommand {\BA}{\mathbf A}
\newcommand {\BB}{\mathbf B}
\newcommand {\BD}{\mathbf D}
\newcommand {\BE}{\mathbf E}

\newcommand {\BS}{\mathbf S}
\newcommand {\BR}{\mathbf R}

\newcommand {\BT}{\mathbf T}
\newcommand {\BU}{\mathbf U}

\newcommand{\BX}{\mathbf X}
\newcommand{\BY}{\mathbf Y}
\newcommand {\bx}{\mathbf x}
\newcommand {\bbb}{\mathbf b}

\newcommand {\be}{\mathbf e}

\newcommand {\bk}{\mathbf k}
\newcommand {\bp}{\mathbf p}

\newcommand {\by}{\mathbf y}

\newcommand {\br}{\mathbf r}

\newcommand {\bu}{\mathbf u}
\newcommand {\bv}{\mathbf v}
\newcommand {\bn}{\mathbf n}
\newcommand{\BPsi}{\boldsymbol \Psi}

\newcommand {\FA}{\mathfrak A}
\newcommand {\FB}{\mathfrak B}

\newcommand {\FD}{\mathfrak D}
\newcommand {\FE}{\mathrm{EP}}

\newcommand {\FH}{\mathfrak H}

\newcommand {\FL}{\mathfrak L}

\newcommand {\FS}{\mathfrak S}
\newcommand {\FT}{\mathfrak T}
\newcommand {\FU}{\mathfrak U}
\newcommand {\FV}{\mathfrak V}
\newcommand {\FW}{\mathfrak W}

\renewcommand{\phi}{\varphi}
\renewcommand{\subseteq}{\subset}
\renewcommand{\supseteq}{\supset}

\newcommand{\bigo}[1]{O\left( #1 \right)}
\newcommand{\smallo}[1]{o\left( #1 \right)}

\DeclareMathOperator{\spec}{spec}
\DeclareMathOperator{\dom}{dom}
\DeclareMathOperator{\vol}{vol}

\DeclareMathOperator{\dist}{dist}

\DeclareMathOperator{\Id}{Id}
\renewcommand{\epsilon}{\varepsilon}
\newcommand{\act}{\triangleright}
\DeclareMathOperator{\id}{id}
\renewcommand{\mathsf}{\mathrm}

\renewcommand{\tilde}{\widetilde}
\newcommand{\flo}[1]{\left\lfloor#1\right\rfloor}
\newcommand{\ceil}[1]{\left\lceil#1\right\rceil}

\newcommand{\abs}[1]{\left\lvert #1 \right\rvert}
\newcommand{\set}[1]{\left\{ #1 \right\}}
\newcommand{\norm}[1]{\left\| #1 \right\|}

\newcommand{\snorm}[3]{\left\| #1\right\|_{#3}^{(#2)}}

\newcommand{\ang}[1]{\langle #1 \rangle}
\newcommand{\er}{r}

\begin{document}

\title[Gauge transform and applications]{The almost periodic Gauge Transform --- An Abstract Scheme with Applications
to Dirac Operators}

\author[J. Lagacé]{Jean Lagacé}
\author[S. Morozov]{Sergey Morozov}
\author[L. Parnovski]{Leonid Parnovski}
\author[B. Pfirsch]{Bernhard Pfirsch}
\author[R. Shterenberg]{Roman Shterenberg}
\address{Department of Mathematics, University College London, Gower Street, London, WC1E 6BT, UK}
\email{j.lagace@ucl.ac.uk}
\email{l.parnovski@ucl.ac.uk}
\email{bernhard.pfirsch.15@alumni.ucl.ac.uk}

\address{Mathematisches Institut der Universität München, Theresienstr. 39,
D-80333 München, Germany}
\email{morozov@math.lmu.de}

\address{University of Alabama at Birmingham, 1300 University Blvd, Birmingham, AL 35294. USA}
\email{shterenb@math.uab.edu}
\date{\today}
\begin{abstract}
  One of the main tools used to understand both qualitative and
  quantitative spectral behaviour of periodic and almost periodic Schr\"odinger operators is the
  method of gauge transform. In this paper, we extend this method to an abstract setting, thus
  allowing for greater flexibility in its applications that include, among others,   
  matrix-valued operators. In particular, we obtain
  asymptotic expansions for the density of states of certain almost periodic systems of elliptic
  operators, including systems of Dirac type. We also prove that a range of
  periodic systems including the two-dimensional Dirac operators satisfy the
  Bethe--Sommerfeld property, that the spectrum contains a semi-axis --- or
indeed two semi-axes in the case of operators that are not semi-bounded.
\end{abstract}
\maketitle

\tableofcontents

\section{Introduction}

  \subsection{A Gauge transform.}

During the last fifteen years, substantial progress has been made in the spectral
theory of periodic and almost periodic scalar operators. An important tool that was developed during this period and was used  
to obtain asymptotic spectral results was the method of gauge
transform (see, e.g., 
\cite{Sobolev2005,Sobolev2006,ParSob2010,ParSht2012,MorParSht2014,ParSht2016,Ivrii2018,ParSht2019}). 
This method, which heavily uses commutator estimates, was originally created for classical 
pseudo-differential operators (see e.g. \cite{Weinstein1977,Rozenbljum1978}) but
was then modified to the periodic case by Sobolev \cite{Sobolev2005,Sobolev2006}
and to the almost periodic setting by Parnovski and Shterenberg
\cite{ParSht2012}. The aim of this paper is to describe the method of gauge
transform on an abstract level and then apply this abstract scheme to a concrete
example --- elliptic systems of operators (including Dirac operators).

Here is the basic setting: suppose that we are given an operator 
\begin{equation}
  A = A_0 + B,
  \label{LP1}
\end{equation}
where $A_0$ is a diagonal operator in a given basis and $B$ is a perturbation,
which is assumed to be small in some sense. The standard example which the
reader may want to keep in mind is
\begin{equation}
  A_0 = \diag(a_1(-\Delta)^{\alpha/2},\dotsc,a_m(-\Delta)^{\alpha/2}),  
\end{equation}
where $\alpha>0$, $0 \ne a_j \in \R$, and $B$ is a pseudo-differential
perturbation of order smaller than $\alpha$
with periodic or almost periodic coefficients. For instance, a Dirac operator
with an almost periodic potential can be brought to such a form by a unitary 
transformation. In many applications we will
furthermore require $A$ to be self-adjoint, even though our general scheme may
not always require it. 

We want to find an operator $A'$
 that is unitarily
equivalent to $A$ and is simpler -- either diagonal or, failing this, has a form 
\begin{equation}\label{LP2}
  A' = U A U^{-1} = A_0' + B',
\end{equation}
where $A_0'$ is diagonal, $U$ is unitary, and $B'$ is a perturbation that is smaller than $B$.
The notion of `smallness' assumes that we have a small parameter, and $B'$ has
this small parameter entering in a higher power than $B$. The most common
example of application to PDEs assumes that the order of $B'$ is smaller than
the order of $B$ (so the role of the small parameter is played by the inverse of the
energy), but in some cases the small parameter can be chosen to be a coupling
constant, see \cite{ParSht2019}. The operators $A$ and $A'$ have the same
spectrum and the hope is that it is easier to describe the spectrum of $A'$, both quantitatively and
qualitatively. As an example of the spectral properties  we want to study, we list the following two types of problems:
\begin{enumerate}
  \item Obtaining asymptotic expansions for the so-called \emph{integrated density of
    states}
    $N(A;\lambda)$ as the spectral parameter $\lambda$ tends to $\pm \infty$;
  \item If $B$ has periodic coefficients, to prove that whenever $A$ is
    unbounded above (resp. below), its spectrum contains a semi-axis
    $[\lambda_0,\infty)$ (resp. $(-\infty,\lambda_0])$. Such an operator $A$ is
    said to satisfy the
      \emph{Bethe--Sommerfeld property}.
\end{enumerate}

If we seek the unitary operator $U$ in \eqref{LP2} in the form 
$U=\exp(i \Psi)$, then we have
\begin{equation}\label{LP3}
  A' = A_0 + B+i [A_0,\Psi] + i[B,\Psi]-\frac12 [[A_0,\Psi],\Psi]-\frac12 [[B,\Psi],\Psi]+R,
\end{equation}
where $R$ consists of further terms given by formally expanding the series for
the exponentials $\exp(i\Psi)$. Our hope is to solve the equation
\begin{equation}
\label{LP4}
B + i[A_0,\Psi]=0
\end{equation}
for $\Psi$, so that the second and third terms of \eqref{LP3} cancel each other.
Ideally, the
rest of the terms (starting from the fourth one) would indeed be smaller than $B$.
In most cases, however, these two wishes turn out to be infeasible. 

The main obstacle is that solutions $\Psi$ to equation \eqref{LP4} involve a  denominator
that could be small for some $B$ (for example, to have any hope of solving
\eqref{LP4}, the diagonal part of $B$ has to be absent). Therefore, we usually
have to modify our procedure and divide the perturbation $B$ into two parts --
good (or non-resonant) part $B^{\CNR}$ for which the equation 
\begin{equation} 
\label{LP5}
B^{\CNR} + i[A_0,\Psi]=0
\end{equation}
has a nice solution $\Psi^{\CNR}$ and bad (or resonant) part
$B^{\CR}=B-B^{\CNR}$
which we will be unable to destroy using our procedure. Thus, at the end we will
have
\begin{equation}
\label{LP6}
  A' = A_0' + B^{R}+B',
\end{equation}
where $B'$ is smaller (in order, say) than $B$. Of course, we also hope that the
resonant part $B^{R}$ is better in some sense than the initial perturbation
$B$; in many applications, the operator $B^{R}$ acts in subspaces of our Hilbert
space generated by `specially designated
and geometrically defined' areas of the phase space.

After we have reduced our operator to the improved form \eqref{LP6}, in
principle we can repeat the same procedure --- finitely, or even infinitely
many times. The latter process is much more difficult to realise, and we will
not give examples of it in this paper. However, in many settings we indeed have
to run this procedure several times (more than once) in order to achieve the
desired `smallness' of the remainder. In other words, we construct the
`improved' operator in the form 
\begin{equation}
\label{LP7}
  A_n = \exp(i \Psi_n)...\exp(i \Psi_2)\exp(i \Psi_1)A\exp(-i \Psi_1)\exp(-i \Psi_2)...\exp(-i \Psi_n).
\end{equation}
We call this method {\emph {the consecutive gauge transform}}. Sometimes, it is more convenient to look for the improved operator in the form 
\begin{equation}
\label{LP8}
  A^{(n)} = \exp(i (\Psi_n+\dotso+ \Psi_2+\Psi_1))A\exp(-i (\Psi_1+\Psi_2+\dotso +\Psi_n)),
\end{equation}
which we call {\emph {the parallel gauge transform}}. In both situations, the operators
$\Psi_j$ are solutions of equations similar in form to \eqref{LP5}.

Another important distinction between different variations of the gauge
transform is as follows. In order to prove that the order of the remainder  $B'$ is smaller than
the order of $B$, we have to estimate the orders of various commutators.
Sometimes, it is enough to have the basic estimate: the order of the commutator
is not greater than the sum of the orders of its entries. This estimate holds
without any restrictions, but for it to be effective we need to have some a
priori inequalities between the orders of the principal term $A_0$ and the
perturbation $B$; we call this approach the weak gauge transform. On the other
hand, quite often we can improve our estimate on commutators: for example,
in the classical scalar pseudo-differential  calculus,  the order
of the commutator can be estimated by the sum of the orders of the entries minus
one. If we have such an estimate, we can guarantee that the order of $B'$ is
indeed smaller than the order of $B$, assuming nothing other than that the order
of $A_0$ is larger than the order of $B$. This approach is called the strong
gauge transform. In this paper, we will define the weak and strong
gauge transforms rigorously and give a general abstract setting in which they can be applied. We
discuss the advantages and drawbacks of both types of gauge transforms and
finish with a couple of concrete applications.

The first application is to obtain asymptotic expansions for the density of
states of elliptic almost periodic operator systems. Under some technical
conditions described later, we may either obtain complete or limited expansions
as the spectral parameter goes to $\pm\infty$.
The other application is to prove that some elliptic periodic systems have the
Bethe--Sommerfeld property. This will be done under the same
conditions that allow us to obtain a complete asymptotic expansion for the density of
states. In either of these cases, some Dirac operators are examples of those to which we
can apply our results.

\subsection{Description of the results for elliptic systems and the Dirac operator}

While describing the precise class of operators $A$ for which we obtain spectral
asymptotics requires definitions that are made later, we can make these results
explicit for Dirac operators in dimension $2$ and $3$ perturbed by classical
pseudo-differential almost periodic operators right away. The
two-dimensional Dirac operator with mass $M$ acts in $\RL^2(\R^2;\C^2)$ and is given by
\begin{equation}
  \BA_{2,M} := -i(\sigma_1 \del_{x_1} + \sigma_2 \del_{x_2}) + \sigma_3 M,
\end{equation}
where $\sigma_1, \sigma_2,
\sigma_3$ are the Pauli matrices
\begin{align}
\sigma_1=\begin{pmatrix}
0 & 1\\
1 & 0
\end{pmatrix},\qquad 
\sigma_2=\begin{pmatrix}
0 & -i\\
i & 0
\end{pmatrix},\qquad \text{and} \qquad
\sigma_3=\begin{pmatrix}
1 & 0\\
0 & -1
\end{pmatrix}.
\end{align}

The three-dimensional Dirac operator with mass $M$ acts in $\RL^2(\R^3;\C^4)$ and is
given by
\begin{equation}
  \BA_{3,M} := -i \left( \gamma_1 \del_{x_1} + \gamma_2 \del_{x_2} + \gamma_3
  \del_{x_3} \right) + \Gamma M,
\end{equation}
where the matrices $\gamma_j$, $\Gamma$ are the Dirac matrices (see
\cite{Upmeier})\footnote{Many authors would write $\alpha_j$ for $\gamma_j$ and
  $\beta$ for $\Gamma$, see e.g. \cite{thaller}. We keep our convention in line
  with higher-dimensional generalisations and to avoid some notational conflicts
later on.}
\begin{equation}
\begin{aligned}
\gamma_j=\begin{pmatrix}
\mathbf 0  & \sigma_j\\
\sigma_j & \mathbf 0
\end{pmatrix}, \quad \text{and} \quad\ 
\Gamma=\begin{pmatrix}
\Id_2 & \mathbf 0\\
\mathbf 0 & -\Id_2
\end{pmatrix}.
\end{aligned}
\end{equation}
We obtain asymptotic expansions for the density of states of operators of the
type $\BA = \BA_{d,M} + \BB$
under the assumption that $\BB$ is a `generic' 
almost periodic pseudo-differential 
perturbation. The precise meaning of generic is given in Section
\ref{sec:besicovitch}. The density of states for elliptic differential operators $A$ that are not semi-bounded can be
defined by the formula
\begin{equation}
  N(\lambda;A)  := \lim_{L \to \infty} \frac{N(\lambda;A_D^{(L)})}{(2L)^d}.
\end{equation}
Here, $A_D^{(L)}$ is the restriction of $A$ to the cube $[-L,L]^d$ with
Dirichlet boundary condition, and $N(\lambda;A_D^{(L)})$ is the counting function
  for the discrete eigenvalues of $A_D^{(L)}$ in the interval
  $[0,\lambda)$ when $\lambda > 0$ and $(\lambda,0]$ when $\lambda < 0$. Later,
  we will give several equivalent definitions of $N(\lambda)$ which are more
  convenient to work with and allow pseudo-differential perturbations.

\begin{thm}
  \label{thm:dirac2dids}
  Let $\BA = \BA_{2,M} + \BB$, where $\BB$ is a generic symmetric pseudo-differential
  operator with almost periodic coefficients of order $\beta < 1$ acting
  in $\RL^2(\R^2;\C^2)$. Then, there is a complete asymptotic expansion for the
  density of states of $\BA$ in the sense that for every $K > -2$, 
  there is a
  finite set $L \subset (0,2+K)$ and constants $C_{j}^\pm$, $C_{j,\log}^\pm$,
  $j \in L \cup \set 0$ such that
    \begin{equation}
      N(\pm\lambda;\BA) = C_0^\pm \lambda^2 + \sum_{j \in L}\left( C_{j}^\pm
        \lambda^{2 - j} +
      C_{j,\log}^\pm \lambda^{j} \log \lambda \right)+ \bigo{\lambda^{-K}}
    \end{equation}
    as $\lambda \to \infty$.
\end{thm}
 Theorem \ref{thm:aexpconcrete} is a more general version of Theorem
 \ref{thm:dirac2dids}. It is
applicable to elliptic systems of
pseudodifferential operators whose principal symbol has only simple eigenvalues.

We obtain a restricted expansion for the three-dimensional case.
\begin{thm}
  \label{thm:dirac3dids}
  Let $\BA = \BA_{3,M} + \BB$, where $\BB$ is a generic operator of the form
  $$ \BB = B_1 \gamma_1 + B_2 \gamma_2 + B_3 \gamma_3 + B_\Gamma \Gamma + B_{\Id}
  \Id_4,$$
  where each $B_j$, $j \in\set{1,2,3,\Gamma,\Id}$ is a scalar symmetric
  pseudo-differential operator with almost periodic coefficients of order
  $\beta$, $0 \le \beta \le 1/2$. Then, writing $\gamma^* = \max \set{\beta - 1,
  2 \beta - 1}$  there is a
  finite set $L \subset (0,1 - \gamma^*)$ and constants $C_{j}^\pm$, $C_{j,\log}^\pm$,
  $j \in L \cup \set{0}$ such that
    \begin{equation}
      N(\pm\lambda;A) = C_0^\pm \lambda^3 + \sum_{j \in L}\left( C_{j}^\pm
        \lambda^{3 - j} +
      C_{j,\log}^\pm \lambda^{3 - j} \log \lambda \right)+ \bigo{\lambda^{2 +
      \gamma^*}}
    \end{equation}
    as $\lambda \to \infty$.
\end{thm}
This time, it is Theorem \ref{thm:aexpconcretecut} which is a more general
version of Theorem \ref{thm:dirac3dids}. It is applicable to elliptic
systems of pseudodifferential operators whose principal symbol has multiple
eigenvalues under some more restrictive conditions on the perturbation.

We also obtain that two-dimensional Dirac operators satisfy the
Bethe--Sommerfeld property.
\begin{thm}
  \label{thm:dirac2dbs}
  Let $\BA = \BA_{2,M} + \BB$, where $\BB$ is a symmetric pseudo-differential operator of
  order $\beta < 1$ with
  periodic coefficients. Then, there exists $\lambda_0 > 0$ such that the
  spectrum of $\BA$ contains intervals $(-\infty,-\lambda_0]$ and
  $[\lambda_0,\infty)$. 
\end{thm}
This theorem also has a more general version in Theorem \ref{thm:bs}. It is applicable to systems
whose principal symbol has only simple eigenvalues.

In Section \ref{sec:dirac}, we also describe generalisations of these results to
higher dimensional Dirac operators, and give some technical conditions under
which we can get complete asymptotic expansions or the Bethe--Sommerfeld
property for the three-dimensional Dirac operator.

\subsection{Description of the main results and plan of the paper}

In the first half of our paper, we discuss the gauge
transform in an abstract setting. The setting is developed  while keeping in
mind
particular applications to almost periodic operators. As such, the space on
which the operators act looks like an
abstract version of a Besicovitch space. In the second half, we will discuss
the specific applications of the results obtained in the first half to
elliptic systems of pseudo-differential almost periodic operators; in
particular,
in the last section
we will show that Dirac operators are a specific example of them. 
An interesting part of the application of
our methods to systems is that we need to intertwine and alternate the use of
the weak and strong gauge transforms, whereas in the past only one type was used
at a time. In order to help the reader familiar with previous literature on 
 the method of gauge transform, we have kept the notation as close as possible to the one used in
\cite{MorParSht2014,ParSob2010}. 

\subsubsection*{Plan and results of Part I} 
In Section \ref{sec:operators} we define an
algebra of operators $\BS^\infty$ acting on a non separable Hilbert space which
should be thought of as an abstract version of a Besicovitch space. For some set
$\Xi$ this algebra will be concretely realised on
$\ell^2(\Xi)$ through a group action on its basis elements. 
This algebra is filtered as an algebra of pseudo-differential operators on
$\ell^2(\Xi)$, and it has similar properties to those of classical pseudo-differential
operators in the PDE sense. Their natural domains are
 generalisations of Sobolev spaces. This section contains many technical but
very useful lemmas describing boundedness properties, adjoints,
compositions and commutators of operators in $\BS^\infty$. One of the main
differences with classical pseudo-differential operators is illustrated in
Proposition \ref{prop:normorder}, which plays the role of the
Calderon--Vaillancourt theorem in our setting. It essentially says that we can
directly correlate symbol norms of operators with the norms of
individual summands in Paley-Wiener type decompositions.

In Section \ref{sec:perturbation}, we turn our attention to some natural 
subspaces of $\BS^\infty$ -- operators that are either elliptic or diagonal. Just
as in the classical setting, our definition of elliptic operators
allows us to characterise natural domains of self-adjointness for operators in
the algebra $\BS^\infty$. The three main results of
this section illustrate the three most important properties of elliptic
operators. In Proposition \ref{prop:globell}, they are shown to admit a
parametrix, and are therefore invertible up to a controllable error. Lemma
\ref{lem:relative boundedness} is used repeatedly throughout the paper and shows
that lower order perturbations of elliptic operators are relatively bounded,
with explicit bounds. Finally, in Proposition \ref{prop:closedness and
self-adjointness}, we show that elliptic operators are closed and self-adjoint
if symmetric. 

In Section \ref{sec:ids}, we consider the situation where operators in
$\BS^\infty$ are affiliated to a $\mathrm{I}_\infty$ or $\mathrm{II}_\infty$
factor. This is common in the study of almost-periodic operators and
their generalisation. We define a general notion of density of states
measures (DSM) in $\BS^\infty$ as traces in the affiliated $\mathrm{I}_\infty$
or $\mathrm{II}_\infty$ factor. We give a variational
description of the DSM of an interval $J$ even in situations where the operator is not bounded
below. This is used to show the principal results of this section: small
perturbations of elliptic self-adjoint operators do not change their density of
states much. The definition of 'smallness' of the perturbation is made clear in
that section. In Lemma \ref{lem:monotonicity}, we control to what extent perturbations of
smaller order can affect the DSM, whereas in Lemma \ref{lem:spectralperturb} it
is perturbations that are spectrally supported away from the interval $J$ that are shown to
have a small effect.

In Section \ref{sec:gt}, we describe the abstract gauge transform scheme,
which is split into two cases: the weak and strong gauge transforms. In both
cases, we describe the resonant regions geometrically as subsets of the index set
$\Xi$. The consecutive scheme for the weak gauge transform is described in
Lemmas \ref{lem:psi} and \ref{lem:onestepweakgauge} and Corollary
\ref{cor:onestepweakgauge}, whereas the parallel scheme is described in Proposition
\ref{prop:symbolestweakgauge}. In both cases, only trivial estimates on the
commutator are used. In Lemma \ref{lem:sgt}, we describe conditions under which
a stronger scheme can be used. Since conditions for the strong transform to be
applicable are varying in nature, we do not attempt at completely classifying
them.

Finally, in Section \ref{sec:mapo}, we describe the case where the
symbols are functions into $\mathrm{Mat}^{m\times m}(\C)$ rather than $\C$.
We describe how this can be reduced to the abstract scalar case and introduce a new
class of operator systems: uncoupled operators. Our goal is to show that under
some specific conditions, elliptic systems are unitarily equivalent to uncoupled
operators up to a remainder which we can control. In that
light, the main results of this section are Theorems
\ref{thm:systemonestep} and \ref{thm:gtsystem} which
give explicit conditions under which one can use the weak gauge transform to
conjugate elliptic symmetric operators into almost uncoupled ones. The remainders are small (in the sense of Section \ref{sec:ids}) perturbations.

\subsubsection*{Plan and results of Part II}

In the second part, we apply the results of Part I to concrete systems of
elliptic pseudo-differential operators with periodic and almost periodic
perturbations. More specifically, we study operator systems of the form $\BA =
\BA_0
+ \BB$,
defined on a dense domain in $\RL^2(\R^d;\C^m)$ where $\BA_0$ is defined as in
\eqref{LP2},
and $\BB$ is a pseudo-differential perturbation of order $\beta < \alpha$. In
Section \ref{sec:besicovitch}, we give a description of these operators in term
of Besicovitch space, and we make the relevant definitions concerning periodic
operators. In Sections \ref{sec:prelim} and \ref{sec:asyexp}, we obtain
asymptotic expansions for the IDS. In Sections \ref{sec:bs} and \ref{sec:cg} we
prove that some elliptic systems of operators have the Bethe--Sommerfeld
property
using some combinatorial geometric arguments. Finally in Section \ref{sec:dirac}
we expose how Dirac operators may fit in our setting.

Since the precise description of the results requires some notations and
language defined in Part I, we postpone their description to the beginning of
Part II.

\subsection*{Acknowledgements}
The research of JL and LP was supported by EPSRC grant EP/P024793/1. The
research of JL was also
partially supported by NSERC's postdoctoral fellowship. The research of SM was
supported by RSF grant 18-11-00032.  The research of RS was supported by NSF grant
DMS-1814664.

\section*{\textsc{Part I : An abstract gauge transform scheme}}

\section{Generalised almost-periodic operators} \label{sec:operators}

In this section, we define an algebra of generalised almost-periodic operators.
We start by defining the space on which those
operators are defined. We also define generalised Sobolev spaces which are their natural domains.
We then describe the algebraic properties of the generalised almost-periodic
operators, and obtain version of the Calderon--Vaillancourt theorem in our
context in Proposition \ref{prop:normorder}.

\subsection{Generalised Sobolev spaces}

Let $\Xi$ be an infinite, possibly uncountable set equipped with a weight function 
$\ang \cdot : \Xi \to [1,\infty)$. We will often call $\Xi$ {\emph {the index set}}. 
For $\gamma \in \R$ we define the spaces
\begin{equation}
 \begin{split} \RH^{\gamma}(\Xi) := \bigg\{\begin{matrix}x: \Xi \to \C\\ \displaystyle  x:\xi\mapsto x_\xi\end{matrix}:\quad \sum_{\xi \in \Xi}\ang{\xi}^{2\gamma}\abs{x_\xi}^2 < \infty\bigg\}
 \end{split}
\end{equation}
and 
\begin{equation}
{\RG}^\infty(\Xi):=\bigcap\limits_{\gamma\in\mathbb{R}}{\RH}^\gamma(\Xi).
\end{equation}
In particular, every $x\in \RH^{\gamma}(\Xi)$ vanishes at all but
countably many $\xi\in\Xi$. Every $\RH^{\gamma}(\Xi)$ is a Hilbert space
with inner product
\begin{equation}
  \label{eq:inner product}
  \left( x,y \right)_{\RH^\gamma(\Xi)} := \sum_{\xi \in
  \Xi}\ang{\xi}^{2\gamma}x_\xi\,\overline{y_\xi}.
\end{equation}
It is easy to see that $\RH^0(\Xi) = \ell^2(\Xi)$ with the standard
orthonormal basis indexed bijectively from $\Xi$ as
\begin{equation}
 \CE := \set{\be_{\xi}: \xi \in \Xi}, \qquad \be_{\xi} :\eta\in\Xi\mapsto \begin{cases}
                                                                     1 &
                                                                     \text{if }\eta
                                                                     =\xi,\\ 0 &
                                                                     \text{if }\eta\neq\xi
                                                                    \end{cases},
\end{equation}
and that $\RH^{\gamma_1}(\Xi) \subset \RH^{\gamma_2}(\Xi)$ for all $\gamma_1 >\gamma_2\in\R$. When there is no risk of confusion, we will write
$\RH^\gamma:=\RH^\gamma(\Xi)$. 

\subsection{An algebra of operators}

Let $G$ be a group that acts from the left on $\Xi$, so that the action is free, i.e.~only the identity of $G$ has fixed points. We denote by $g\act \xi$ the action of $g\in G$ on $\xi\in\Xi$.
Starting from the weight function $\langle\cdot\rangle$ on $\Xi$ we define one on $G$ by 
\begin{equation}
  \ang g := 1 + \sup_{\xi \in \Xi} \left|\ang{g\act\xi} - \ang{\xi}\right|.
  \label{eq:modulusg}
\end{equation}
We assume that $G$ has a \emph{bounded range of action}, which means that $\ang g$ is finite for all $g\in G$.

It will be useful for future convenience to observe the following properties of the
weight function:
\begin{lem} For all $f,g \in G$, $\xi \in \Xi$ and $t \in\mathbb R$ the following relations hold:
\begin{enumerate}
 \item 
\begin{equation}
 \ang g = \ang{g^{-1}};\label{eq:anginv}
\end{equation}
 \item Peetre-type inequalities:
\begin{equation}
  \ang g^{-1} \ang \xi \le \ang{g \act \xi} \le \ang g \ang \xi
\label{eq:peetre}
\end{equation}
and
\begin{equation}
  \ang{fg}^t \le \min\left\{\ang f^t \ang g^{\abs t}, \ \ang f^{\abs t} \ang g^{t}\right\}.
  \label{eq:peetre2}
\end{equation}
\end{enumerate}
\end{lem}

\begin{proof}
For all $g \in G$, $\xi \in \Xi$ the definition \eqref{eq:modulusg} implies \eqref{eq:anginv} and the estimates
\begin{align}\label{bounds on g act xi}
 \max\big\{1, 1 +\ang{\xi} -\ang{g}\big\} \leq \ang{g\act \xi} \leq \ang\xi +\ang g -1.
\end{align}
Note the relations
\begin{eqnarray}
 a +1 -b =\big((b -1)(a -b) +a\big)/b \geq a/b,\quad&\text{for all }& a\geq b\geq 1.\label{upper1}\\
 a +b -1 \leq a +b -1 +(a -1)(b -1) =ab,\quad&\text{for all }& a, b\geq 1,\label{lower1}
\end{eqnarray}
The first estimate in \eqref{eq:peetre} follows from \eqref{bounds on g act xi} and \eqref{upper1}, the second from \eqref{bounds on g act xi} and \eqref{lower1}.
Now by \eqref{eq:modulusg} and \eqref{upper1} for all $f,g \in G$ we obtain
\begin{equation}\label{eq:angprod}
\begin{aligned}
 \ang{fg}&\leq 1 +\sup_{\xi \in \Xi} \left|\ang{fg\act\xi} - \ang{g\act \xi}\right| +\sup_{\xi \in \Xi}\left|\ang{g\act\xi} - \ang{\xi}\right| \\
 &=\ang f +\ang g -1 \leq \ang f\ang g,
 \end{aligned}
\end{equation}
which implies \eqref{eq:peetre2} for $t >0$. Now \eqref{eq:angprod} and \eqref{eq:anginv} imply
\begin{equation}
 \ang g =\ang{f^{-1}fg} \leq\ang f\ang{fg}\text{ and }\ang f =\ang{fgg^{-1}} \leq \ang{fg}\ang{g},
\end{equation}
which delivers \eqref{eq:peetre2} for $t < 0$. The case $t = 0$ is trivial.
\end{proof}

\begin{defi}
We call a function $b : G \times \Xi \to \C$, $(g,\xi) \mapsto b_g(\xi)$ an almost periodic
\emph{symbol} if there exists a countable set $\Theta \subset G$, closed under inversion and containing the
identity $\id_G$, such that for all $g \in G \setminus
\Theta$, $b_g(\xi) \equiv 0$. Whenever there is no risk of confusion, we will write $\id := \id_G$. 
We call $\Theta$ a \emph{frequency set} 
for $b$ and the functions $\left\{b_\theta(\,\cdot\,)\right\}_{\theta\in\Theta}$
the \emph{Fourier coefficients} of $b$.
For every symbol $b$ and every $\gamma\in\mathbb{R}$, $l\geq 0$,
we define the family of norms
\begin{equation}\label{eq:defsnorm}
  \snorm{b}{\gamma}{l} := \sum_{\theta \in \Theta} \ang \theta^l
    \sup_{\xi \in \Xi} \left( \ang{\xi}^{-\gamma}
    \left|b_\theta(\xi)\right|\right).
\end{equation}
The class of \emph{symbols of order $\gamma$} is defined as
\begin{equation} \label{eq:sgamma}
  \BS^{\gamma}:=\BS^{\gamma}(G,\Xi):=\set{ b:G\times\Xi\to\C :
  \snorm{b}{\gamma}{l}<\infty \text{ for all } l\geq 0}.
\end{equation}
\end{defi}

The space of symbols is naturally a linear space. It is clear that if $\Theta$
is a frequency set for a symbol, then any $\Gamma \supset \Theta$ is also one.
It is obvious from the definition that  $\snorm{\cdot}{\gamma}{l}$
is a decreasing function of $\gamma$ and an increasing function of $l$, thus
\begin{equation}
\BS^{\gamma_1}\subseteq\BS^{\gamma_2},\quad \text{ for all }\gamma_1\leq \gamma_2.
\end{equation}
We introduce
\begin{equation}
  \BS^\infty:= \bigcup_{\gamma \in \R} \BS^\gamma \quad \text{ and }
    \quad \BS^{-\infty}:= \bigcap_{\gamma \in \R} \BS^\gamma.
  \label{eq:algebradef}
\end{equation}
\begin{lem}\label{lem:complete}
For every $\gamma\in\mathbb{R}$, the space $\BS^\gamma$ equipped with the family of norms 
 $\big\lbrace\snorm{\cdot}{\gamma}{l}\big\rbrace_{l\geq 0}$ is a Fréchet space.
\end{lem}
\begin{proof}
  Consider a sequence 
  \begin{equation}
(b_n)_{n\geq 1}\subset \BS^\gamma
\end{equation}
  that is Cauchy with respect to $\snorm{\cdot}{\gamma}{l}$ for every $l \geq
  0$, and denote by $\Theta(n)$ a frequency set for each $b_n$.
  Then, for all $\theta \in G$, we observe that 
$b_\theta(\xi):=\lim\limits_{n\to\infty}(b_n)_\theta(\xi)$ exists and vanishes
outside the countable set $\Theta = \bigcup_n \Theta(n)$. It is a simple computation to see that 
$b\in\BS^\gamma$ with $\snorm{b_n-b}{\gamma}{l}\to 0$, as 
$n\to\infty$, for all $l\geq 0$. Hence, the claim follows.
\end{proof}

\begin{defi}
Let $b:G\times\Xi\to\C$ be a symbol with frequency set $\Theta\subseteq G$ and 
\begin{equation}
\big(b_{\theta}(\xi)\big)_{\theta\in\Theta}\in \ell^2(\Theta), \text{ for all } \xi\in\Xi.
\end{equation}
Then the \emph{almost periodic linear operator associated to $b$} is 
\begin{equation}
B:= \Op(b):\spann(\CE)\to \ell^2(\Xi)
\end{equation}
defined by 
\begin{equation}\label{eq:defopB}
  B \be_\xi := \sum_{\theta \in \Theta} b_\theta(\xi) \be_{\theta\act\xi}, \quad\text{for all }
  \xi\in\Xi.
\end{equation}
\end{defi}
\begin{rem}
If $b\in\BS^\infty$, then, in view of \eqref{eq:defsnorm} and \eqref{eq:sgamma},
$\big(b_{\theta}(\xi)\big)_{\theta\in\Theta}\in \ell^1(\Theta)\subseteq
\ell^2(\Theta)$ holds for all $\xi\in\Xi$. This means that we can associate an
almost periodic operator to every symbol in $\BS^\infty$. 
On the other hand, since the group action of $G$ on $\Xi$ is free, $b$ can be recovered from $B$ via the identity
\begin{equation}\label{eq:symbrecovery}
  b_g(\xi)=(\be_{g\act\xi},B\be_\xi)_{\ell^2(\Xi)}, \quad\text{for all $g\in G$, $\xi\in\Xi$.}
\end{equation}
Thus, there is a one-to-one correspondence between almost periodic symbols and almost
periodic operators. This correspondance is in contrast to the case of 
classical pseudo-differential operators
where this correspondence is only modulo smoothing operators. Hence, we allow
ourselves to overload the notation and write $B = \Op(b)  \in \BS^\gamma$ if $b
\in \BS^\gamma$, $\gamma \in \R \cup \set{\pm \infty}$, and let $\snorm B \gamma
l := \snorm b \gamma l$ for all $l \ge 0$, $\gamma \in \R$. Note that this
correspondence gets lost if one does not require the group action of $G$ on $\Xi$ to be free.
Our construction can be generalised to such non-free group actions, but for
simplicity of the exposition we do not do it in this paper.
\end{rem} 
We call $B$ \emph{quasi-periodic} if $b$ admits a finite frequency set.
A simple example of a quasi-periodic operator of class $\BS^\gamma$,
$\gamma\in\R$, is $\Op(h)$ with
$$h_g(\xi):= \begin{cases}\tilde h(\xi) & \text{if } g = \id, \\ 0
&\text{otherwise.}\end{cases}$$ Here, $\tilde h$ is a function on $\Xi$ satisfying $\big|\tilde h(\xi)\big| \leq\ang\xi^\gamma$ for all $\xi\in\Xi$.

\begin{rem}
Our terminology is justified by the following example. Suppose that $G$ is a
locally compact abelian (LCA) group and $G_B$ is its Bohr compactification, see
\cite[\S 1]{Shubin1978}. Index by $\Xi$ the set of characters
$\widetilde{\CE}:=\lbrace \tilde{\be}_\xi:\xi\in\Xi\rbrace$ of $G$ or,
equivalently, $G_B$. On $\mathrm{CAP}(G)$, the continuous almost periodic
functions on $G$, we can define an inner product $(f,g) = \CM(f\overline g)$,
where $\CM(f)$ is the mean of $f$ with respect to the normalised Haar measure on
$G_B$. The Besicovitch space $\RB^2(G)$ is defined as the closure of
$\mathrm{CAP}(G)$ with respect to the norm induced by this inner product. 
By \cite[Proposition 1.5]{Shubin1978},
 the map
\begin{equation}
\mathcal{E}\to\widetilde{\mathcal{E}}\, \quad \be_\xi\mapsto\tilde{\be}_\xi,
\end{equation}
extends to an isometric isomorphism $\ell^2(\Xi)\to \mathsf{B}^2(G)$. In particular, for
$G = (\R^d,+)$, one has $\widetilde{\CE}=\lbrace \bx\mapsto \exp(i \bx\cdot
\bxi),\ \bxi\in\mathbb{R}^d\rbrace$ and the operators in
$\BS^\infty(\mathbb{R}^d,\mathbb{R}^d)$ correspond to almost periodic
pseudo-differential operators in $\mathsf{B}^2(\mathbb{R}^d)$ or
$\mathsf{L}^2(\mathbb{R}^d)$, see \cite[\S 3--4]{Shubin1978} and \cite[Equation
(8.8)]{ParSht2012}.  
The present work 
  can be applied to more general settings, for example,
  when the underlying group $G$ is non-abelian. Note that the Bohr
  compactification construction is inadequate in that situation, e.g., for $G = SL(2;\R)$
  we have $G_B = \set \id$, see \cite[p. 4]{Shubin1978}.
\end{rem}

From Lemma \ref{lem:complete} we obtain the following corollary.
\begin{cor} \label{cor:complete}
  Let $(B_n)_{n\geq 1} \subset \BS^\gamma$ be such that
\begin{equation}
  \sum_{n\geq 1}\snorm{B_n}{\gamma}{l}<\infty
\end{equation}
for all $l \geq 0$. Then the sum 
\begin{equation}
B:=\sum_{n\geq 1}B_n
\end{equation}
converges in $\BS^\gamma$ with
\begin{equation}
  \snorm{B}{\gamma}{l}\leq\sum\limits_{n\geq 1}\snorm{B_n}{\gamma}{l}.
\end{equation}
\end{cor}

Up until now, operators from $\BS^\infty$ were only defined on $\spann(\CE)$.
We now show that they can be extended in a natural way.
\begin{lem} \label{lem:domain}
For every $\beta,\gamma\in\mathbb{R}$ the operator $B \in \BS^\gamma$ can be uniquely extended to a bounded  linear operator $B : \mathsf H^{\beta} \to \mathsf H^{\beta-\gamma}$. Moreover, we have the bound
  \begin{equation} 
  \norm{B}_{\mathsf H^{\beta} \to \mathsf H^{\beta-\gamma}} 
    \le \snorm{B}{\gamma}{\abs{\beta-\gamma}}.
\end{equation}
\end{lem}
\begin{proof}
  Let $x,y \in \spann(\CE)$, i.e.~$x_\xi=y_\xi=0$ for all but finitely many $\xi$.
  Then, the Cauchy--Schwarz and Peetre inequalities \eqref{eq:peetre} imply
\begin{equation}
  \begin{aligned}
    \left|(x,By)_{\mathsf H^{\beta-\gamma}} \right|& = \abs{\sum_{\theta\in\Theta} \sum_{\xi\in\Xi}
            \ang{\theta\act\xi}^{2(\beta - \gamma)} b_\theta(\xi) \overline{x_{\theta\act\xi}}\,y_\xi}\\
    &\le \sum_{\theta \in \Theta} \ang \theta^{\abs{\beta-\gamma}} 
     \sup_{\zeta \in \Xi} \left(\ang{\zeta}^{-\gamma}\abs{b_\theta(\zeta)}\right) \times \\
    & \qquad \times \left(\sum_{\xi\in \Xi}\ang{\theta \act \xi}^{2(\beta-\gamma)}
         \abs{x_{\theta\act\xi}}^2\right)^{1/2}
    \left(\sum_{\xi \in \Xi} \ang \xi^{2\beta} \abs{ y_\xi}^2\right)^{1/2}\\
    &\leq \snorm{B}{\gamma}{\abs{\beta-\gamma}} \|x\|_{\mathsf H^{\beta-\gamma}
            } \|y\|_{\mathsf H^{\beta}}.
\end{aligned}
\end{equation} 
The claim follows by density of $\spann\CE$ in $\mathsf H^{\alpha}$ for all $\alpha\in\mathbb R$.
\end{proof}

We obtain the following immediate corollary.

\begin{cor} \label{cor:domain}
  Every $B \in \BS^0$ extends to a bounded operator on $\ell^2 =\mathsf H^0$
  such that
 \begin{equation}
   \|B\|_{\ell^2\to\ell^2} \le \snorm{B}{0}{0}.
 \end{equation}
\end{cor}
\begin{defi}
For $b\in\BS^\infty$, we define
\begin{equation}\label{eq:adjoint symbol}
b^\dagger_\theta(\xi):=\begin{cases}\overline{b_{\theta^{-1}}(\theta\act\xi)}&
  \text{if } \theta\in \Theta,\\
  0& \text{if }\theta\in G\setminus\Theta
\end{cases}
\end{equation}
for all $\xi\in\Xi$, where $\Theta$ is a frequency set for $b$.
\end{defi}
\begin{lem} \label{lem:adjoint}
If $b\in\BS^\gamma$, then $b^\dagger\in\BS^\gamma$.
Moreover, for all $x,y\in\mathsf{H}^\gamma$, one
has 
\begin{equation}\label{eq:B+onH}
  (x,By)_{\ell^2(\Xi)}=(B^\dagger x,y)_{\ell^2(\Xi)}, \quad  \text{i.e.~$B^\dagger\subseteq B^*$.}
\end{equation}
In particular, $B$ is symmetric on $\mathsf{H}^\gamma$ if and only if $B=B^\dagger$.
\end{lem}
\begin{proof}
Every frequency set $\Theta$ for $b\in\BS^\gamma$ is also one for $b^\dagger$. Moreover, since $\Theta=\Theta^{-1}$ holds by convention, \eqref{eq:anginv} and \eqref{eq:peetre} imply that for all $l\geq 0$,
\begin{equation*}
\begin{aligned}
\snorm{b^\dagger}{\gamma}{l}&=\sum\limits_{\theta\in\Theta_B} \ang{\theta}^l\sup\limits_{\xi\in\Xi}\big[\ang{\xi}^{-\gamma}|b_{\theta^{-1}}(\theta\act\xi)|\big]\\
&=\sum\limits_{\theta\in\Theta_B}\ang{\theta}^l\sup\limits_{\xi\in\Xi}\big[\ang{\theta^{-1}\act\xi}^{-\gamma}|b_{\theta^{-1}}(\xi)|\big]\\
&\leq \sum\limits_{\theta\in\Theta_B}\ang{\theta}^{l+|\gamma|}\sup\limits_{\xi\in\Xi}\big[\ang{\xi}^{-\gamma}|b_{\theta^{-1}}(\xi)|\big]\\
&=\sum\limits_{\theta\in\Theta_B}\ang{\theta}^{l+|\gamma|}\sup\limits_{\xi\in\Xi}\big[\ang{\xi}^{-\gamma}|b_{\theta}(\xi)|\big]\\
&=\snorm{b}{\gamma}{l+|\gamma|},
\end{aligned}
\end{equation*}
thus $b^\dagger\in\BS^\gamma$ holds. 
Moreover, \eqref{eq:symbrecovery} and \eqref{eq:adjoint symbol} yield
\begin{equation}\label{eq:B+onspanE}
(\be_\eta,B\be_\xi)=(B^\dagger\be_\eta,\be_\xi), \qquad \text{for all $\eta,\xi\in\Xi$.}
\end{equation}
 In view of Lemma \ref{lem:domain} and the density of $\spann\mathcal{E}$ in $\mathsf{H}^\gamma$, \eqref{eq:B+onspanE} extends to \eqref{eq:B+onH}. This finishes the proof of the lemma.
\end{proof}
\begin{defi}
  Let $a, b \in \BS^\infty$ be
  symbols with frequency sets $\Theta_a$ and $\Theta_b$. The \emph{composed
  symbol} $a \circ b$ with frequency set
\begin{align}\label{eq:frequencysetcomposition}
\Theta_{a\circ b}:=\Theta_a \Theta_b := \set{\theta_a \theta_b : 
                    \theta_a \in \Theta_a, \theta_b \in \Theta_b}
\end{align}
is defined as 
\begin{equation}\label{eq:symbolcomposition}
 (a \circ b)_\theta(\xi) := \sum_{\theta_a \theta_b = \theta} 
           a_{\theta_a}(\theta_b \act\xi) b_{\theta_b}(\xi) \quad\text{ for all
           $\theta\in\Theta_{a\circ b}$, $\xi\in\Xi$}.
\end{equation} 
\end{defi}

\begin{lem}\label{lem:product}
For $\alpha, \beta\in \R$ let $A = \Op(a) \in \BS^{\alpha}$ and $B = \Op(b) \in \BS^{\beta}$.
Then $AB\in\BS^{\alpha+\beta}$ and $AB =\Op(a\circ b)$. Moreover, for all $l \ge 0$ we have the bound
\begin{equation}\label{eq:product}
  \snorm{AB}{\alpha+\beta}{l} \leq 
          \snorm{A}{\alpha}{l}\snorm{B}{\beta}{l + \abs{\alpha}}.
\end{equation}
\end{lem}
\begin{proof}
The frequency set $\Theta_{a\circ b}$ is, clearly, a countable set.
For any $l \ge 0$, we have
\begin{equation}
  \begin{aligned}
    \snorm{a\circ b}{\alpha + \beta}{l} &= \sum_{\theta \in \Theta_{a\circ b}}
           \sum_{\theta_a\theta_b = \theta} \ang{\theta}^l 
            \sup_{\xi \in \Xi} \left(\ang{\xi}^{-\alpha - \beta} 
             \abs{a_{\theta_a}(\theta_b\act\xi)}\abs{b_{\theta_b}(\xi)}\right) \\
  &\le \sum_{\theta_b \in \Theta_b} \ang{\theta_b}^{l + \abs{\alpha}}
             \sup_{\xi \in \Xi}\left(\ang{\xi}^{-\beta} 
              \abs{b_{\theta_b}(\xi)}\right)\times \\
  &\qquad \times   \sum_{\theta_a \in \Theta_a}\ang{\theta_a}^{l} 
                    \sup_{\zeta \in \Xi}\left(\ang{\theta_b \act \zeta}^{-\alpha}
                     \abs{a_{\theta_a}(\theta_b \act \zeta)}\right) \\
  &\le \snorm{a}{\alpha}{l} \snorm{b}{\beta}{l + \abs{\alpha}}.
\end{aligned}
\end{equation}
Thus $a\circ b\in \BS^{\alpha + \beta}$ and \eqref{eq:defopB} implies $AB =\Op(a\circ b)$.
\end{proof}

It is natural to consider operators from $\BS^\infty$ on the common domain $\mathsf{H}^\infty$. Then Lemmata \ref{lem:domain}, \ref{lem:adjoint}, and \ref{lem:product} yield the following corollary.
\begin{cor}
$\BS^\infty=\bigcup\limits_{\gamma\in\mathbb{R}}\BS^\gamma$ is a $*$-algebra of operators on $\mathsf{H}^\infty$, filtered by $\mathbb{R}$, with involution $\dagger$. The subalgebra of regularising operators $\BS^{-\infty}$ forms a two-sided ideal of $\BS^\infty$.
\end{cor}

We also consider the adjoint actions $\ad(A,B): = \mathrm i(AB -BA)$ with the frequency set
$\Theta_{\ad(a,b)} = \Theta_{a \circ b } \cup \Theta_{b \circ a}$. The Fourier coefficients
of $\ad(A,B)$ are
\begin{equation}\label{eq:ad symbol}
 \ad(a,b)_\theta(\xi) = \mathrm i\, \left(\sum_{\theta_a \theta_b = \theta} 
              a_{\theta_a}(\theta_b\act\xi) b_{\theta_b}(\xi) - 
 \sum_{\theta_b \theta_a = \theta} 
        b_{\theta_b}(\theta_a\act\xi)a_{\theta_a}(\xi)\right),
\end{equation}
for all $\theta \in \Theta_{\ad(a,b)}$.
If $G$ is commutative, \eqref{eq:ad symbol} simplifies to
\begin{equation}\label{eq:adcom}
 \ad(a,b)_\theta(\xi) = \mathrm i\, \sum_{\theta_a \theta_b = \theta}\big(
   a_{\theta_a}(\theta_b\act\xi) b_{\theta_b}(\xi) - 
 b_{\theta_b}(\theta_a\act\xi)a_{\theta_a}(\xi)\big). 
\end{equation}

For $k=1,2,3,\dots$ and $A,B,B_1,\dots B_k\in\BS^\infty$, we define recursively \label{page:defad}
\begin{equation}\label{eq:defadk}
\begin{aligned}
\ad(A;B_1,\dotsc,B_k) &:= \ad\left(\ad(A;B_1,\dotsc,B_{k-1}),B_k\right),\\
\ad^0(A,B)&:= A,\\
\ad^k(A;B)&:= \ad(\ad^{k -1}(A;B);B).
\end{aligned}
\end{equation}

The following lemma is a direct consequence of Lemma \ref{lem:product}.
\begin{lem}\label{lem:weakcommest}
  Let $k \in \N$ and assume that $A_j \in \BS^{\gamma_j}$ for $0 \le j \le k$. Put
  \begin{equation}
    \gamma = \sum_{j=0}^k \gamma_j, \qquad \hat\gamma =\sum_{j=0}^k |\gamma_j|.
  \end{equation}
 Then
 $\ad(A_0;A_1,\dotsc,A_k) \in \BS^{\gamma}$. Furthermore, if for all $0 \le j
 \le k$ we have $A_j = A_j^\dagger$, then
 $\ad(A_0;A_1,\dotsc,A_k)=\ad(A_0;A_1,\dotsc,A_k)^\dagger$. Moreover, for all $l
 \ge 0$ we have
  \begin{equation}
  \begin{aligned}\label{eq:normestcomm}
  \snorm{\ad(A_0;A_1,\dotsc,A_k)}{\gamma}{l}\leq 2^{k} 
  \prod_{j=0}^k \snorm{A_j}{\gamma_{j}}{l + \hat\gamma - |\gamma_j|}.
\end{aligned}
\end{equation}
 In particular, for any $A\in\BS^{\alpha}$, $B\in\BS^{0}$ and $k\in\mathbb N$ we obtain the estimate
\begin{align}\label{eq:normestcommzeroorder}
  \snorm{\ad^k(A;B)}{\alpha}{l}\leq 2^k \snorm{A}{\alpha}{l}
            \left(\snorm{B}{0}{l + \abs{\alpha}}\right)^k.
\end{align}
\end{lem}
For some $\Xi$ and $G$ it may be possible to improve this lemma and show that $\ad(A,B)\in\BS^\gamma$ holds with $\gamma<\alpha+\beta$ for all $A\in\BS^{\alpha}$, $B\in\BS^{\beta}$.
This will be discussed in Section \ref{sec:sgt}.

The following proposition provides bounds on norms of operators restricted to
`annuli' in $\Xi$.
\begin{prop} \label{prop:normorder}
  For $1 \le m \le M \le \infty$, let $\Upsilon \subseteq \set{\xi \in \Xi : m \le \ang \xi \le M}$ and denote by 
  $P_{\Upsilon}$ the orthogonal projection in $\ell^2(\Xi)$ onto the closure of $\spann\set{\be_\xi : \xi \in \Upsilon}$.
  Then, for any $A \in \BS^\gamma$ with $\gamma \ge 0$, the norm inequality
  \begin{equation}
  \label{eq:upslarge}
    \norm{A P_\Upsilon}_{\ell^2\to\ell^2} \le   M^\gamma\snorm{A}{\gamma}{0}
  \end{equation}
  holds. For any $A \in \BS^\gamma$ with $\gamma \le 0$, we get the inequality
  \begin{equation}
  \norm{A P_\Upsilon}_{\ell^2\to\ell^2} \le m^\gamma\snorm{A}{\gamma}{0}.
  \end{equation}
\end{prop}
\begin{proof}
  Observe that $P_\Upsilon$ is a quasi-periodic operator with a frequency set $\Theta=\{\id\}$ and the symbol $(p_\Upsilon)_{\id} =\bone_{\Upsilon}$ (the indicator function of $\Upsilon$). Thus, for all $\gamma\in\mathbb{R}$ and $l\geq0$, 
  \begin{equation}
  \label{eq:upscases}
  \begin{aligned}
   \snorm{P_\Upsilon}{-\gamma}{l} &= \sup_{\xi \in \Upsilon} \ang \xi^\gamma \leq \begin{cases}
      \hfill m^\gamma &\text{if } \gamma \le 0, \\
       \hfill M^\gamma &\text{if } \gamma \ge 0.
      \end{cases}
   \end{aligned}
  \end{equation}
If $M < \infty$ or $\gamma \le 0$, then Corollary \ref{cor:domain} and Lemma \ref{lem:product} imply the bound
\begin{equation}
 \norm{A P_\Upsilon}_{\ell^2\to\ell^2} \le \snorm{A}{\gamma}{0} \snorm{P_\Upsilon}{-\gamma}{\abs{\gamma}},
\end{equation}
and the statement of the lemma follows from \eqref{eq:upscases}. On the other hand, the inequality \eqref{eq:upslarge} is trivial for $M = \infty$ and $\gamma > 0$.
\end{proof}
\section{Elliptic and diagonal operators}\label{sec:perturbation}

In this section, we introduce particular classes of operators from $\BS^\infty$
and study their properties. Some of these classes do depend on the specific
choice of orthonormal basis $\CE$ for $\ell^2(\Xi)$. However, the class of
operators on which our main theorems depend, that of elliptic operators, is
invariant under change of basis.
\begin{defi}\label{def:dapo}
  The subalgebra $\BD\BS^\infty \subset \BS^\infty$ of \emph{diagonal operators} is defined as
\begin{equation}
  \BD\BS^\infty := \set{A = \Op(a) \in \BS^\infty:\set{\id} \text{ is a frequency set
  for } a}.
  \label{eq:dapo}
\end{equation}
For symbols of operators from $\BD\BS^\infty$ we can suppress the subscript $\id$, i.e. we let $a(\xi) :=a_{\id}(\xi)$ for all $A =\Op(a)\in\BD\BS^\infty$, $\xi\in\Xi$.
For $\alpha \in \R\cup \set{-\infty}$ we define $\BD\BS^\alpha := \BD\BS^\infty \cap \BS^\alpha$. Introduce the map $\CD
: \BS^\infty \to \BD\BS^\infty$, $A\mapsto A^\CD$, that projects $A=\Op(a)$ onto its diagonal
part $A^\CD:= \Op(a^\CD)$ where
\begin{equation}
a^\CD(\xi) := a_{\id}(\xi),
\end{equation}
i.e.
\begin{equation}
 A^\CD\be_{\xi} =\langle\be_{\xi}, A\be_{\xi}\rangle \be_{\xi}
\end{equation}
holds for all $\xi \in\Xi$.
We also define the \emph{off-diagonal part} as $A^{\CO\CD} := \Op(a^{\CO\CD})$ with $a^{\CO\CD}:=a-a^\CD$.
\end{defi}

Note that for any $A\in\BS^\alpha$ with $\alpha\in\mathbb R$ and all $l\geq 0$,
\begin{equation}
  \label{eq:partition}
\snorm{A^{\CD}}{\alpha}{l}+\snorm{A^{\COD}}{\alpha}{l}=\snorm{A}{\alpha}{l}
\end{equation}
and
\begin{equation}\label{eq:diagonal l irrelevant}
 \snorm{A^\CD}{\alpha}{l} =\snorm{A^\CD}{\alpha}{0}.
\end{equation}

\begin{defi}\label{def:DES}
The set $\BD\BE\BS^\alpha$ of \emph{diagonal  elliptic operators of order
$\alpha\in\R$} is defined as the set of operators $A = \Op(a) \in \BD\BS^\alpha$ for which
there exist \emph{ellipticity parameters} $\kappa> 0$ and $\er\ge 1$ such that
\begin{equation}
  \label{eq:ellconst}
|a(\xi)|\geq \kappa\ang{\xi}^\alpha \qquad \text{ for all  $\xi\in\Xi$
 such that  $\ang{\xi}\geq 
\er$}.
\end{equation}
Let the set of ellipticity parameters $(\kappa, \er)$ of $A$ be denoted by
$\FE(A)$. Note that $(\kappa, \er)\in \FE(A)$ implies $(\tilde \kappa,\widetilde
\er)\in \FE(A)$ for all $0 <\tilde \kappa\leq \kappa$, and $\widetilde \er \geq \!
\er$.
\end{defi}

\begin{defi}\label{def:ES}
The set $\BS\BE\BS^\alpha$ of \emph{strongly elliptic operators of order $\alpha\in\mathbb{R}$} consists of operators $A\in\BS^\alpha$ such that $A^\CD \in \BD\BE\BS^\alpha$ and $A^{\CO\CD} \in
\BS^\gamma$ for some $\gamma < \alpha$.
For $(\kappa,\er)\in \FE(A^\CD)$ we define $P_\er$ as the diagonal operator with symbol
$\bone_{\{\xi: \ang{\xi}\leq
\er\}}$. We also define $P_\er^c$ as $\Id - P_\er$, and
\begin{equation}\label{eq:tilde A}
  \tilde A_{\kappa,\er} := A^\CD P_\er^c + \kappa \er^\alpha P_\er.
\end{equation}
\end{defi}
\begin{defi}\label{def:WES}
The set $\BE\BS^\alpha$ of \emph{elliptic operators of order $\alpha\in\R$}
consists of operators $A\in\BS^\alpha$ for which there exists a unitary
$U\in\BS^0$ with $U A U^\dagger\in\BS\BE\BS^\alpha$.
\end{defi}
As we did with diagonal operators, we set
\begin{equation}
\BT^\infty:=\bigcup\limits_{\gamma\in\mathbb{R}}\BT^\gamma, \ \BT^{-\infty}:=\bigcap\limits_{\gamma\in\mathbb{R}}\BT^\gamma, \ \text{for} \ \BT\in \lbrace \BD\BE\BS, \BS\BE\BS, \BE\BS \rbrace.
\end{equation}
 Clearly, both $\BS\BE\BS^\alpha$ and $\BE\BS^\alpha$ are closed
under addition of operators in $\BS^\beta$, $\beta < \alpha$. 
\begin{prop}
  \label{prop:globell}
Let $A\in\BS^\infty$ and $\alpha >0$ such that $A^{\CD}\in\BD\BE\BS^\alpha$. For any $(\kappa, \er)\in\FE(A^\CD)$ the
  operator $\tilde A_{\kappa,\er}$ is invertible with $\tilde
  A_{\kappa,\er}^{-1}\in\BD\BS^{-\alpha}$ and for all $l \ge 0$ we have 
\begin{equation}
  \label{eq:globinvsnorm}
  \snorm{\tilde A_{\kappa,\er}^{-1}}{\gamma}{l} =\snorm{\tilde
  A_{\kappa,\er}^{-1}}{\gamma}{0} \leq \kappa^{-1}
  \begin{cases}
    \er^{-\alpha}
  &\text{ for } \gamma \geq 0,\\ \er^{-\alpha -\gamma} &\text{for }-\alpha \leq
\gamma < 0.\end{cases}
\end{equation}
Moreover, the following estimates hold for all $\gamma \in\R$, $l\geq 0$:
    \begin{equation}\label{eq:AA_R -Id}
      \snorm{A \tilde A_{\kappa,\er}^{-1} - \Id}{\gamma - \alpha}{l} \le 
      \er^{\alpha -\gamma} + \frac 1 \kappa \left(
         \er^{\alpha - \gamma} \snorm{A^\CD}{\alpha}{0} +
        \snorm{A^{\CO\CD}}{\gamma}{l}
      \right)
    \end{equation}
    and
    \begin{equation}\label{eq:A_RA -Id}
      \snorm{\tilde A_{\kappa,\er}^{-1} A - \Id}{\gamma - \alpha}{l} \le 
       \er^{\alpha -\gamma} + \frac 1 \kappa \left(
         \er^{\alpha - \gamma} \snorm{A^\CD}{\alpha}{0} +
        \snorm{A^{\CO\CD}}{\gamma}{l+\alpha}
      \right).
    \end{equation}
  \end{prop}
    
\begin{proof}
  We have that $\tilde A_{\kappa,\er} - A^{\CD} = (\kappa \er^\alpha - A^{\CD}) P_\er$, and
since $P_\er \in \BS^{-\infty}$,
which is an ideal of $\BS^\infty$, we observe that $\tilde A_{\kappa,\er} \equiv
A^{\CD}
\mod{\BS^{-\infty}}$. By \eqref{eq:ellconst} and \eqref{eq:tilde A},
$\tilde A_{\kappa,\er}\in\BD\BS^\infty$ and its symbol satisfies
\begin{equation}
  \abs{\tilde a_{\kappa, \er}(\xi)} = \abs{a_{\id}(\xi)}\bone_{\ang \xi > \er} +
  \kappa \er^\alpha\bone_{\ang \xi \le \er} \ge \kappa \ang
  \xi^\alpha\bone_{\ang \xi > \er} + \kappa \er^\alpha\bone_{\ang \xi \le \er}
\end{equation}
for all $\xi \in \Xi$. Hence $\tilde A_{\kappa,\er}^{-1} =\Op\big(\tilde
a_{\kappa,\er}^{-1}\big)\in\BD\BS^{-\alpha}$ and \eqref{eq:globinvsnorm} holds.

The estimates \eqref{eq:AA_R -Id} and \eqref{eq:A_RA -Id} follow by applying \eqref{eq:product} term-wise to the right hand sides of the identities
\begin{equation*}\begin{split}
 A \tilde A_{\kappa,\er}^{-1} - \Id &= -P_\er
 +\kappa^{-1}\er^{-\alpha}A^{\CD}P_\er +A^{\CO\CD}\tilde A_{\kappa,\er}^{-1},\\
 \tilde A_{\kappa,\er}^{-1}A  - \Id &= -P_\er
 +\kappa^{-1}\er^{-\alpha}A^{\CD}P_\er +\tilde A_{\kappa,\er}^{-1}A^{\CO\CD}
\end{split}\end{equation*}
and taking \eqref{eq:diagonal l irrelevant} into account.
\end{proof}

\begin{lem}\label{lem:relative boundedness}
Let $\beta\in\mathbb R$, $\alpha>\max(\beta,0)$, $0 <\gamma < \alpha$, and assume that $A \in \BS\BE\BS^\alpha$ with $A^{\CO\CD}\in \BS^\gamma$ and $B \in \BS^\beta$.
Then for $\beta\le 0$ the operator $B$ is bounded, and, in the case of $\alpha
>\beta >0$, for every $x\in\mathsf H^\alpha$ and $$(\kappa,
\er)\in\FE(A^\CD)\cap \Big\{\er
\geq\big(\|A^{\CO\CD}\|_0^{(\gamma)}/\kappa\big)^{1/(\alpha
-\gamma)}\Big\},$$ we have
\begin{equation}\label{eq:infinitesimal estimate}
 \|Bx\| \le \frac{\er^{\beta -\alpha}\|B\|_0^{(\beta)}}{\kappa -\er^{\gamma
 -\alpha}\|A^{\CO\CD}\|_0^{(\gamma)}}\Big(\|Ax\| +\kappa \er^{\alpha}\big(1 +\kappa^{-1}\|A^\CD\|_0^{(\alpha)}\big)\|x\|\Big). 
\end{equation}
In particular, $B$ is infinitesimally $A$-bounded in $\ell^2(\Xi)$.
\end{lem}

\begin{proof}
The only non-trivial case is $\alpha >\beta >0$. 
For every $x\in\mathsf H^\alpha$ we have
\begin{equation*}
 \|Bx\| \le \big\|B\tilde A_{\kappa,\er}^{-1}\big\|\norm{A^\CD x}
 +\big\|B(\tilde A_{\kappa,\er}^{-1}A^\CD -\Id)\big\|\norm{x},
\end{equation*}
with $\tilde A_{\kappa,\er}$ defined as in \eqref{eq:tilde A}. 
Corollary \ref{cor:domain} and displays \eqref{eq:product},
\eqref{eq:globinvsnorm} and \eqref{eq:A_RA -Id} imply the estimates
\begin{align*}
 \big\|B(\tilde A_{\kappa,\er}^{-1}A^\CD -\Id)\big\| &\le
 \|B\|_0^{(\beta)}\|\tilde A_{\kappa,\er}^{-1}A^\CD
 -\Id\|_{|\beta|}^{(-\beta)}\\
 &\le \er^\beta\|B\|_0^{(\beta)}\big(1 +\kappa^{-1}\|A^\CD\|_0^{(\alpha)}\big)
\end{align*}
and
\begin{align*}
 \big\|B\tilde A_{\kappa,\er}^{-1}\big\| &\le \|B\|_0^{(\beta)}\big\|\tilde 
 A_{\kappa,\er}^{-1}\big\|_0^{(-\beta)}\\ &\le \kappa^{-1}\|B\|_0^{(\beta)}\er^{\beta -\alpha},
\end{align*}
and we obtain
\begin{equation}\label{eq:prelim infinitesimal estimate}
 \|Bx\| \le \kappa^{-1}\er^{\beta -\alpha}\|B\|_0^{(\beta)}\Big(\|A^\CD x\|
 +\kappa \er^\alpha\big(1 +\kappa^{-1}\|A^\CD\|_0^{(\alpha)}\big)\|x\|\Big),
\end{equation}
which is \eqref{eq:infinitesimal estimate} with $A^\CD\in \BD\BE\BS^\alpha$
replacing $A$.
Applying \eqref{eq:prelim infinitesimal estimate} with $B=A^{\CO\CD}$, we arrive at
\begin{equation}
  \begin{aligned}
 \|A^{\CO\CD}x\| &\le \kappa^{-1}\er^{\gamma
 -\alpha}\|A^{\CO\CD}\|_0^{(\gamma)}\|A^\CD x\| \color{black} \\ &\qquad +\er^\gamma\|A^{\CO\CD}\|_0^{(\gamma)}\big(1 +\kappa^{-1}\|A^\CD\|_0^{(\alpha)}\big)\|x\|.
 \end{aligned}
\end{equation}
Hence we have
\begin{equation*}
\begin{split}
 \|Ax\| &\ge \|A^\CD x\| -\|A^{\CO\CD}x\|\\ &\ge \big(1 -\kappa^{-1}\er^{\gamma
 -\alpha}\|A^{\CO\CD}\|_0^{(\gamma)}\big)\|A^\CD x\| -\er^\gamma\|A^{\CO\CD}\|_0^{(\gamma)}\big(1 +\kappa^{-1}\|A^\CD\|_0^{(\alpha)}\big)\|x\|,
\end{split}
\end{equation*}
which implies
\begin{equation}\label{eq:A^CD control}
\begin{aligned}
 \|A^\CD x\| \le \big(1 -&\kappa^{-1}\er^{\gamma
 -\alpha}\|A^{\CO\CD}\|_0^{(\gamma)}\big)^{-1} \Big(\|A x\| +\er^\gamma\|A^{\CO\CD}\|_0^{(\gamma)}\big(1 +\kappa^{-1}\|A^\CD\|_0^{(\alpha)}\big)\|x\|\Big).
 \end{aligned}
\end{equation}
Substituting \eqref{eq:A^CD control} into \eqref{eq:prelim infinitesimal estimate} we obtain \eqref{eq:infinitesimal estimate}.
\end{proof}

We conclude the section with the following proposition. 
\begin{prop}\label{prop:closedness and self-adjointness}
For $\alpha \in\mathbb R$, every operator from $\BE\BS^\alpha$ is closed on
$\mathsf H^{\max\{\alpha, 0\}}$ in the Hilbert space $\ell^2(\Xi)$. Every symmetric
operator from $\BE\BS^\alpha$ defined on $\mathsf H^{\max\{\alpha, 0\}}$ is
self-adjoint.
\end{prop}
\begin{proof}
For $\alpha\leq 0$ we have $\BE\BS^\alpha\subset \BS^0$, and the statements follow from Corollary \ref{cor:domain}.
Now assume $A\in \BS\BE\BS^\alpha$ with $\alpha >0$. By \eqref{eq:inner
product}, Definition \ref{def:ES} and Lemma \ref{lem:domain}, for any $(\kappa,
\er)\in\FE(A^\CD)$ we have the estimates
\begin{equation}
 \kappa^2\|x\|_{\mathsf H^\alpha}^2 \le \|\tilde A_{\kappa,\er}x\|^2 \le
 \big(\|\tilde A_{\kappa,\er}\|^{(\alpha)}_{\alpha}\big)^2\|x\|_{\mathsf H^\alpha}^2
\end{equation}
for all $x \in \RH^\alpha$. Hence the graph norm of $\tilde A_{\kappa,\er}$ is equivalent to the norm of
$\mathsf H^\alpha$, and $\tilde A_{\kappa,\er}$ is closed on $\mathsf H^\alpha$.
If $\tilde A_{\kappa,\er}$ is symmetric, then for every $x\in\dom(\tilde
A_{\kappa,\er}^*)$ there exists $C_x\geq 0$ such that for all $y \in \RH^\alpha$
\begin{equation}
  \big|(x, \tilde A_{\kappa,\er}y)\big|\leq C_x\|y\|_{\ell^2(\Xi)}.
\end{equation}
In particular, with $(y_{n})_\xi:=\overline{\tilde a_{\kappa, \er, \id}(\xi)}\bone_{\ang
\xi \le n}\ x_\xi$ for $n\geq \er$, $\xi\in \Xi$, we obtain by \eqref{eq:ellconst} and \eqref{eq:tilde A} that
\begin{equation}
\begin{aligned}
 \sum_{\substack{\xi \in \Xi \\ \ang \xi \le n}}\ang{\xi}^{2\alpha}|x_\xi|^2
 &\leq \kappa^{-1}(x, \tilde A_{\kappa,\er}y_{n})\leq
 \kappa^{-1}C_x\|y_n\|_{\ell^2(\Xi)}\\
&\leq \kappa^{-1}C_x\snorm{\tilde{A}_{\kappa,\er}}{\alpha}{0}\Big(\sum_{\substack{\xi \in \Xi \\ \ang \xi \le n}}\ang{\xi}^{2\alpha}|x_\xi|^2\Big)^{1/2}.
\end{aligned}
\end{equation}
Passing to the limit $n\to \infty$, it follows from \eqref{eq:inner product}
that $\|x\|_{\mathsf H^\alpha} \leq
\kappa^{-1}C_x\snorm{\tilde{A}_{\kappa,\er}}{\alpha}{0}$, i.e. $$\dom(\tilde
A_{\kappa,\er}^*)\subset \mathsf H^\alpha =\dom(\tilde A_{\kappa,\er}),$$ hence
$\tilde A_{\kappa,\er}$ is self-adjoint.
By Lemma \ref{lem:relative boundedness} $A -\tilde A_{\kappa,\er}$ is
infinitesimally $\tilde A_{\kappa,\er}$-bounded, so that $A$ is also
self-adjoint (see,
e.g., Theorems 3.4.2 and 4.1.9 in \cite{BirmanSolomyak}). 

For
$A\in\BE\BS^\alpha$, by Definition \ref{def:WES} there exist a unitary
$U\in\BS^0$ and $H\in\BS\BE\BS^\alpha$ such that $A=UH U^\dagger$. Moreover, it
follows from Lemma \ref{lem:domain} that
$U\mathsf{H}^\alpha=U^\dagger\mathsf{H}^\alpha=\mathsf{H}^\alpha$. Now, let
$(x_n)_{n\in\mathbb{N}}\subset\mathsf{H}^\alpha$ with $x_n\to x$ and $U^\dagger
HUx_n\to z$ in $\ell^2$, as $n\to\infty$. Since $U$ is bounded, $Ux_n\to Ux$ and
$HUx_n\to Uz$, thus the closedness of $H$ implies that $Ux\in \mathsf{H}^\alpha$
and $HUx=Uz$, i.e.~$x\in\mathsf{H}^\alpha$ and $U H U^\dagger x=z$. Hence, $A$ is
closed on $\mathsf{H}^\alpha$ and self-adjoint if symmetric.
\end{proof}

\section{The Density of States Measure and von Neumann Algebras} \label{sec:ids}

In this section, following \cite{Shubin1979a}, we consider a representation of
$\BS^\infty$ into another operator algebra, affiliated with an infinite factor
(accounting for the almost periodicity), and define the density of states
measure (DSM) for self-adjoint operators in $\BE\BS^\infty$ with respect to this
representation. For a suitable representation, this DSM will coincide with the
classically defined DSM on elliptic differential operators with almost periodic
coefficients. We follow the construction and terminology of \cite[\S
1]{Shubin1979a}, generalising Shubin's symbol classes to the ones defined in Section~\ref{sec:operators}. 

\subsection{Representations of the operator algebra}

Let $\FH$ be a Hilbert space and $\FA$
a factor of either type $\text{I}_\infty$ or $\text{II}_\infty$ in
$\mathcal{B}(\FH)$, the algebra of bounded linear operators in $\mathfrak{H}$.
Let $\widetilde{\mathsf{H}}^\infty$ be a dense subspace of $\mathfrak{H}$ and
$$\widetilde{\BS}^\infty =\bigcup_{\gamma\in\R}\widetilde\BS^\gamma$$ a
$*$-algebra of unbounded linear operators in $\mathfrak{H}$ defined on
$\widetilde{\mathsf{H}}^\infty$, filtered by $\mathbb{R}$. We assume that
$\widetilde{\BS}^\infty\widetilde{\mathsf{H}}^\infty\subseteq\widetilde{\mathsf{H}}^\infty$, and that $\widetilde{\BS}^\infty$ is invariant under the involution
$\tilde
A\mapsto\tilde{A}^\dagger:=\big.\tilde{A}^*\big|_{\widetilde{\mathsf{H}}^\infty}$,
where $\tilde{A}^*$ is the adjoint to $\tilde
A\in\widetilde{\BS}^\infty$. We also suppose that $\widetilde{\BS}^\infty$ is
affiliated with
the factor $\FA$, denoted $\widetilde{\BS}^\infty \eta \FA$. Finally, we consider a representation $\rho: \BS^\infty \to \widetilde{\BS}^\infty$ having the following properties:
\begin{enumerate}[(i)]
\item \label{item:rhohomo} $\rho$ is a homomorphism of filtered $*$-algebras with $\rho(\BS^\gamma)\subseteq \widetilde\BS^\gamma$, for all $\gamma\in\mathbb R$.
\item \label{item:rhoS0bounded} For every $A\in\BS^0$, $\rho(A)$ extends to a bounded linear operator on $\mathfrak{H}$ with
\begin{equation}
  \norm{\rho(A)}_{\FH \to \FH} = \norm{A}_{\mathsf{\ell}^2(\Xi)\to\mathsf{\ell}^2(\Xi)}.
  \label{eq:normpreserving}  
\end{equation}
\item For all $A\in\BS^\infty$, $\rho(A)$ is closable in $\mathfrak{H}$ with the closure $A^\sharp:=\overline{\rho(A)}$. For every $\alpha>0$ there exists a dense subspace $\widetilde{\mathsf{H}}^\alpha\supseteq\widetilde{\mathsf{H}}^\infty$ such that
\begin{enumerate}
 \item $\widetilde{\mathsf{H}}^\alpha\subseteq\widetilde{\mathsf{H}}^\gamma$ if $0 <\gamma\leq \alpha$,
 \item \label{item:domsharpDES}$A\in\BD\BE\BS^\alpha$ implies $\dom(A^\sharp)=\widetilde{\mathsf{H}}^\alpha$,
 \item \label{item:domsharpinv} for all $B\in\BS^0$, $B^\sharp \widetilde{\mathsf{H}}^\alpha\subseteq\widetilde{\mathsf{H}}^\alpha$. 
\end{enumerate}
\item \label{item:sharpAselfadjoint} If $A\in\BD\BE\BS^\alpha$, $\alpha>0$, is self-adjoint on $\mathsf{H}^\alpha$, then $A^\sharp$ is self-adjoint.
\end{enumerate}
\begin{rem}
When $\FA$ is a $\mathrm{I}_\infty$ factor some of the statements in this section become rather trivial. However, we include this case for applications in Section \ref{sec:bs}.
\end{rem}
\begin{rem}
  \label{rem:Shubin}
In \cite{Shubin1979a}, Shubin considers $G = \R^d$ acting
  on itself by translation, with almost periodic  operators acting both in Besicovitch space $\mathsf{B}^2(\mathbb{R}^d)\cong \ell^2(\R^d)$ and in $\mathsf{L}^2(\R^d)$ through the Fourier
  integral representation of pseudo-differential operators. The appropriate Hilbert space is
  then
  \begin{equation}
    \FH = \mathsf{B}^2(\R^d) \otimes \mathsf{L}^2(\R^d),
    \label{eq:hilbert}
  \end{equation}
  and the $\text{II}_\infty$ factor $\FA$ is generated by the two families of
  operators
  \begin{equation}
    \set{\be_{\bxi} \otimes \be_{\bxi}: \bxi \in \R^d} \text{
    and }  \set{I \otimes T_{\bxi}: \bxi \in \R^d},
    \label{eq:iiinffact}
  \end{equation}
  where $\be_{\bxi}$ is multiplication by the character $\be_{\bxi}(\bx)=\mathrm
  e^{\mathrm i\bxi\cdot \bx}$ and 
  $T_{\bxi}$ is the
  translation operator $T_{\bxi} f(\bx) = f(\bx - \bxi) $. The representation
  $\rho$ is given on $A=\Op(a)\in \BS^\infty$ by the linear operator $\rho(A)=a(\bx+\by;D_\by)$ acting on 
\begin{equation}
\widetilde{\mathsf{H}}^\infty:=\mathsf B^2(\R^d) \otimes \hat{\mathsf{H}}^\infty(\R^d).
\end{equation}  
    Here, $\bx$ is the variable of functions in $\mathsf B^2(\R^d)$, $\by$ is the
    variable of functions in $\RL^2(\R^d)$, $D_y=-\mathrm i\nabla_\by$, and $\hat{\mathsf{H}}^\infty(\R^d):= \big\{f\in \mathsf C^\infty(\mathbb{R}^d): \partial^\alpha f\in\mathsf{L}^2(\mathbb{R}^d)\text{ for all }\alpha\in\mathbb N_0^d\big\}$. 
\end{rem}
Properties \eqref{item:rhohomo} and \eqref{item:rhoS0bounded} of the representation $\rho$ imply the following lemma.
\begin{lem}\label{lem:sharpS0}
If $A\in\BS^0$, then $A^\sharp$ is defined on $\mathfrak{H}$ and satisfies
$\big(A^\sharp\big)^*=\big(A^\dagger\big)^\sharp = \big(A^*\big)^\sharp$  and
$\|A^\sharp\|_{\FH\to\FH}=\|A\|_{\ell^2(\Xi)\to\ell^2(\Xi)}$. In particular, the
map $\BS^0\to\FB(\FH),\ A\mapsto A^\sharp$ is an injective homomorphism of
$*$-algebras. If $U\in\BS^0$ is unitary, then so is $U^\sharp$.
\end{lem}
We will now carry over Lemma \ref{lem:relative boundedness} to images under $\sharp$. This provides us with some information on the domains of operators from $(\BS^\infty)^\sharp$. 
\begin{lem}\label{lem:relboundsrho} Let $\beta\in\mathbb R$,  $B\in\BS^\beta$ and $A\in\BS\BE\BS^\alpha$ for some $\alpha>0$. Then
\begin{enumerate}
 \item\label{approx_domain} $\bigcup_{\zeta >\max\{\beta, 0\}}\widetilde{\mathsf H}^\zeta\subseteq\dom(B^\sharp)$,
 \item\label{tildeHalpha} $\dom(A^\sharp)=\widetilde{\mathsf{H}}^\alpha$.
 \item\label{relbound} Suppose $\beta < \alpha$ and $0 <\gamma < \alpha$ with
   $A^{\CO\CD}\in \BS^\gamma$. Then for $\beta\le 0$ the operator $B^\sharp$ is
   bounded, and, otherwise, for every $\phi\in\widetilde{\mathsf{H}}^\alpha$ and
   $$(\kappa, \er)\in\FE(A^\CD)\cap \Big\{\er
   \geq\big(\|A^{\CO\CD}\|_0^{(\gamma)}/\kappa\big)^{1/(\alpha
 -\gamma)}\Big\},$$
   we have
\begin{equation}\label{eq:relboundssharp}
 \|B^\sharp\phi\|_{\mathfrak{H}} \le \frac{\er^{\beta
 -\alpha}\|B\|_0^{(\beta)}}{\kappa -\er^{\gamma
 -\alpha}\|A^{\CO\CD}\|_0^{(\gamma)}}\Big(\|A^\sharp\phi\|_{\mathfrak{H}}
 +\kappa \er^{\alpha}\big(1 +\kappa^{-1}\|A^\CD\|_0^{(\alpha)}\big)\|\phi\|_{\mathfrak{H}}\Big). 
\end{equation}
In particular, $B^\sharp$ is infinitesimally $A^\sharp$-bounded in $\mathfrak{H}$.
\end{enumerate}
\end{lem}
\begin{proof}
For $\beta\le 0$ the statements \emph{\eqref{approx_domain}} and \eqref{relbound} follow from \eqref{eq:normpreserving}. Let now $\beta>0$ and assume that $0 <\gamma < \alpha$ with $A^{\CO\CD}\in \BS^\gamma$.
Following the proof of Lemma \ref{lem:relative boundedness} and applying properties \eqref{item:rhohomo} and \eqref{item:rhoS0bounded} of the representation $\rho$ where necessary, we derive \eqref{eq:relboundssharp} for $\phi\in\widetilde{\mathsf H}^\infty$. Consequently, the graph norm of $\rho(B)$ is dominated 
by the graph norm of $\rho(A)$, thus $\dom(B^\sharp)\supseteq \dom(A^\sharp)$.
Applying \eqref{eq:relboundssharp} for $A^{\CO\CD}$ instead of $B$ and $A^{\CD}$ instead of $A$, we conclude that the graph norms of $\rho(A)$ and $\rho(A^\CD)$ are equivalent, thus \eqref{item:domsharpDES} implies $\dom(A^\sharp)=\dom\big(\big(A^\CD\big)^\sharp\big)=\widetilde{\mathsf{H}}^\alpha$, which is \eqref{tildeHalpha}. Now \eqref{approx_domain} follows by varying $A\in \BS\BE\BS^\alpha$ with $\alpha >\beta$.
Finally, we can extend \eqref{eq:relboundssharp} from $\widetilde{\mathsf{H}}^\infty$ to $\widetilde{\mathsf{H}}^\alpha$ by density with respect to the graph norm of $A^\sharp$.
\end{proof}
Properties \eqref{item:domsharpDES} and \eqref{item:sharpAselfadjoint} of the
map $\rho$ can also be extended to operators from the classes $\BE\BS^\alpha$, $\alpha>0$.
\begin{lem}\label{lem:sharpWES}
Let $\alpha>0$ and $A\in\BE\BS^\alpha$. Then $\dom(A^\sharp)=\widetilde{\mathsf{H}}^\alpha$ and for all unitary $U\in\BS^0$
\begin{equation}\label{eq:Usharp*AsharpUsharp}
  U^\sharp A^\sharp \big(U^\sharp\big)^* = \big(U A U^\dagger\big)^\sharp \qquad \text{holds on} \ \widetilde{\mathsf{H}}^\alpha.
\end{equation}
Moreover, if $A$ is self-adjoint on $\mathsf{H}^\alpha$, then $A^\sharp$ is self-adjoint.
\end{lem}
\begin{proof}
Assume first that $A\in\BS\BE\BS^\alpha$, so that $A^\CD\in\BD\BE\BS^\alpha$ and $A^{\CO\CD}\in\BS^\gamma$ for some $0<\gamma<\alpha$. According to Lemma \ref{lem:relboundsrho}$(\ref{approx_domain}, \ref{tildeHalpha})$ we have that $\dom(A^\sharp)=\widetilde{\mathsf{H}}^\alpha\subseteq\dom\big(\big(A^{\CO\CD}\big)^\sharp\big)$.
Moreover, if $A$ is self-adjoint, then $\big(A^\CD\big)^\sharp$ is self-adjoint on $\widetilde{\mathsf{H}}^\alpha$ and $\rho(A^{\CO\CD})$ is symmetric on $\widetilde{\mathsf H}^\infty$, as follows from properties \eqref{item:sharpAselfadjoint} and \eqref{item:rhohomo} of $\rho$, respectively. Since $\dom\big(\big(A^{\CO\CD}\big)^\sharp\big)\supseteq\widetilde{\mathsf{H}}^\alpha$ is the closure of $\widetilde{\mathsf H}^\infty$ with respect to the graph norm of $\big(A^{\CO\CD}\big)^\sharp$, the operator $\big(A^{\CO\CD}\big)^\sharp$ is also symmetric on $\widetilde{\mathsf{H}}^\alpha$.
Moreover, by Lemma \ref{lem:relboundsrho}\eqref{relbound} it is infinitesimally $A^\sharp$-bounded. 
Thus, \cite[Theorem~4.1.9]{BirmanSolomyak} implies that $A^\sharp=\big(A^\CD +A^{\CO\CD}\big)^\sharp$ is self-adjoint on $\widetilde{\mathsf{H}}^\alpha$.

Let now $A\in\BE\BS^\alpha$. By definition, there exist $H\in\BS\BE\BS^\alpha$
and $V\in\BS^0$ unitary such that $A=V^\dagger HV$ on $\mathsf{H}^\infty$. Since $\rho$ is a $*$-homomorphism, 
\begin{equation}\label{eq:rhoAprod}
\rho(A)=\rho(V)^\dagger \rho(H)\rho(V)
\end{equation}
holds on $\widetilde{\mathsf{H}}^\infty$.
By Lemma \ref{lem:sharpS0} the operator $V^\sharp$ is unitary, and property \eqref{item:domsharpinv} implies that 
\begin{equation}\label{eq:VsharptildeHalpha}
V^\sharp \widetilde{\mathsf{H}}^\alpha= \big(V^\sharp\big)^*\widetilde{\mathsf{H}}^\alpha=\widetilde{\mathsf{H}}^\alpha.
\end{equation}
We have already proved in Lemma \ref{lem:relboundsrho}\eqref{tildeHalpha} that
$\dom\big(H^\sharp\big)=\widetilde{\mathsf H}^\alpha$, thus the argument at the
end of the proof of Proposition \ref{prop:closedness and self-adjointness}
implies that $\big(V^\sharp)^\ast H^\sharp V^\sharp$ is closed on
$\widetilde{\mathsf{H}}^\alpha$. As by \eqref{eq:rhoAprod} and Lemma
\ref{lem:sharpS0} this operator is an extension of $\rho(A)$, it follows that
$\dom(A^\sharp)\subseteq\widetilde{\mathsf{H}}^\alpha$. Similarly, we have on
$\widetilde{\mathsf{H}}^\infty$ that
\begin{equation}
\rho(H)=\rho(V)\rho(A)\rho(V)^\dagger
\end{equation}
and $V^\sharp A^\sharp \big(V^\sharp\big)^*$ is a closed operator on
$V^\sharp\dom(A^\sharp)\subseteq \widetilde{\mathsf{H}}^\infty$. Thus,
$$\widetilde{\mathsf{H}}^\alpha=\dom(H^\sharp)\subseteq
V^\sharp\dom(A^\sharp),$$
and \eqref{eq:VsharptildeHalpha} yields
$\widetilde{\mathsf{H}}^\alpha\subseteq\dom(A^\sharp)$. Hence
$\dom(A^\sharp)=\widetilde{\mathsf{H}}^\alpha$. Finally, let $U\in\BS^0$ be
unitary. Then 
\begin{equation}
\rho(U)\rho(A)\rho(U)^\dagger=\rho(U A U^\dagger)\subseteq (U A
U^\dagger)^\sharp,
\end{equation} 
and $U^\sharp A^\sharp\big( U^\sharp\big)^*$ is a closed extension of
$\rho(U) \rho(A)\rho(U)^\dagger$ on the domain
$\widetilde{\mathsf{H}}^\alpha=\dom\big((U A U^\dagger)^\sharp\big)$,
i.e.~\eqref{eq:Usharp*AsharpUsharp} holds. If $A$ is self-adjoint on
$\mathsf{H}^\alpha$, then so is $H$, thus $H^\sharp$ by the first part of the
proof. Hence, the self-adjointness of $A^\sharp$ follows from
\eqref{eq:Usharp*AsharpUsharp} with $U=V$ and $H$ instead of $A$.
\end{proof}

\subsection{The density of states measure}
Since $\mathfrak{A}$ is a factor of type $\mathrm{I}_\infty$ or $\mathrm{II}_\infty$, there exists, by definition, a semi-finite faithful normal trace $\mathfrak{T}$ on $\mathfrak{A}$, see \cite[I.6 and I.8.4]{Dixmier}. 
Moreover, due to \cite[I.6.4, Corollary]{Dixmier}, this trace is unique up to multiplication by a positive number. Following the notation of \cite{Naimark}, we write $\FL\eta\FA$ to denote that 
$\FL\subset \FH$ is a closed linear subspace adjoint to $\FA$, i.e.~$P_{\FL}\in\FA$, where $P_\FL$ is the projection onto $\FL$. If $\FL\eta \FA$, the relative dimension of $\FL$ is defined by 
\begin{align*}
\FD(\FL):=\FT(P_{\FL})\in [0,\infty].
\end{align*}
If $\FA$ is a $\mathrm I_\infty$-factor, the range of the relative dimension is $c\mathbb N_0\cup\set{\infty}$, for some
$c>0$. It is $[0,\infty]$ if $\FA$ is a
$\mathrm{II}_\infty$-factor. 

\begin{defi}
Let $A\in\BS^0\cup\BE\BS^\infty$ be symmetric and $J\subseteq \mathbb R$ be a
Borel measurable set. Denote by $E_J(A^\sharp)$ the spectral projection of
$A^\sharp$ for $J$. We define the \emph{density of states measure} (DSM) of $A$
on $J$, relative to the representation $\rho$, by 
\begin{equation}
 N(J; A) := \FT\big(E_J(A^\sharp)\big)=\FD\big(E_J(A^\sharp)\FH\big).
  \label{eq:dfnids}
\end{equation}
\end{defi}

\begin{rem} \label{rem:iiinfinityfactor}
  Usually, the dependence on the representation $\rho$ and the factor $\FA$ is unambiguous
  and is thus not reflected in the notation.
  
\end{rem}
The following corollary generalises \cite[Lemma 4.4]{ParSht2012}. It follows
directly from Lemma \ref{lem:sharpWES} (or Lemma \ref{lem:sharpS0} for
$A\in\BS^0$) and the invariance of $\FT$ under unitary transformations in
$\mathfrak{A}$. We remark at this point that, since
$\widetilde{\BS}^\infty\eta\FA$, one has $U^\sharp\in\mathfrak{A}$ for every
unitary $U\in\BS^0$, see Lemma \ref{lem:sharpS0} and \cite[\S 35.1]{Naimark}.
\begin{cor} \label{cor:unitary}
Let $U\in\BS^0$ be unitary and let $A\in\BS^0\cup\BE\BS^\infty$ be symmetric.
Then one has $N(J;A) = N(J;U A U^\dagger)$ for any Borel measurable set $J\subseteq \mathbb{R}$.
\end{cor}

In the remainder of this section, we investigate the behaviour of the DSM for
elliptic operators of positive order under perturbations. In
\cite{ParSht2012,MorParSht2014,ParSht2016} such an analysis was conducted for
operators that are bounded from below and the particular case $J = (-\infty,
\lambda)$, $\lambda\in\mathbb{R}$.

Before continuing, let us introduce the following notation. For any interval $J =[s,t]\subset \R$, $s <t$ and $\eps \in\R$,
 we define 
 $$
 J_{\eps}:= \begin{cases}
   \hfill\varnothing\hfill & \text{for } \eps < \dfrac{s-t}{2}, \\
   [s-\eps, t+\eps] & \text{otherwise}.
 \end{cases}
   $$
 The following lemma gives
us a variational characterisation of the DSM (cf. \cite[Lemma 4.1]{ParSht2012}).

\begin{lem}
  \label{lem:variational}
  Let $A\in \BS^0\cup\BE\BS^\infty$ be symmetric. Then, for any interval
  $J=[q-r,q+r]$ with $q\in\mathbb{R}$ and $r>0$, we
  have 
\begin{equation} \label{eq:variational}
  \begin{aligned}
   N(J;A) &= \sup\big\{\FD(\FL) : \FL\subset\dom(A^\sharp),\ \FL \eta \FA,
     \\&\qquad \qquad\text{\emph{and} } 
   \|(A^\sharp-q)\phi\|_{\mathfrak{H}}\leq r\|\phi\|_\mathfrak{H} \ \forall \phi
 \in \FL\big\}.
 \end{aligned}
 \end{equation}
The analogous statement holds for the open interval $J=(q-r,q+r)$ with strict inequality in \eqref{eq:variational}.
\end{lem}
\begin{rem}
  Usually, variational characterisations such as \eqref{eq:variational} are
  given in terms of quadratic forms rather than norms. The reason why we cannot
  do so is because we do not assume the operator $A$ to be semi-bounded, $J$ a
  semi-infinite interval. One can interpret Lemma \ref{lem:variational} in terms
  of quadratic forms as usual for the nonnegative operator $(A - q)^2$.
\end{rem}
\begin{proof}
 Choosing $\FL := E_J(A^\sharp) \FH$, we observe that $N(J;A)$ is at most the right hand side of \eqref{eq:variational}.
Suppose that there exists a subspace $\FL$ that satisfies the assumptions on the
righthand side of \eqref{eq:variational} and $\FD(\FL) > \FD(E_J(A^\sharp)\FH)$. 
Then \cite[\S 37.1, Lemma]{Naimark} implies that $\FL$ contains an element
$\phi$ orthogonal to $E_J(A^\sharp)\FH$, implying that $\|(A^\sharp-q)\phi\|_{\FH}^2>r^2\|\phi\|_{\FH}^2$, which is a contradiction.
\end{proof}

The following lemma generalises \cite[Corollary 4.3]{ParSht2012}
to operators that are not necessarily bounded below and unbounded perturbations.

\begin{lem} \label{lem:monotonicity}
Let $A\in\BS\BE\BS^\alpha$, $\alpha>0$, and $B\in\BS^\beta$, $\beta<\alpha$,
symmetric operators. Let $J:=[q-r,q+r]\subset \R$ be the interval of length $2r
> 0$ centred at $q\in\mathbb{R}$. Then there exists
a constant $C\geq 0$ depending only on $A$ and $\beta$ such that, for
\begin{equation}\label{eq:epsa}
\epsilon:=\epsilon_{J,A,B}:=\begin{cases}
  \hfill\norm{B}\hfill & \text{if }\beta \leq 0,\\
\dfrac{\snorm{B}{\beta}{0}}{2 +\snorm{B}{\beta}{0}}\Big(r + \abs q+
C\big(1+\snorm{B}{\beta}{0}\big)^{\frac{\alpha}{\alpha-\beta}}\Big)& \text{if } \beta>0,\end{cases}
\end{equation}
the inequality
\begin{equation}\label{eq:bounds}
N(J_{-\epsilon}; A)\leq N(J; A+B)\leq N(J_{\epsilon}; A)
\end{equation}
holds.
\end{lem}
\begin{proof} 
In view of Lemma \ref{lem:relboundsrho}(\ref{approx_domain},\ref{relbound}) and property \eqref{item:rhohomo} of the representation $\rho$, one has that $(A+B)^\sharp =A^\sharp+B^\sharp$ on $\dom(A^\sharp)\subseteq\dom(B^\sharp)$. Fix $\phi\in \FL := E_{J}\big((A+B)^\sharp)\FH\subseteq \dom\big((A+B)^\sharp\big)=\dom(A^\sharp)$, so that
\begin{equation}\label{A+B}
  \norm{(A^\sharp + B^\sharp - q)\phi}_\FH \le r \norm \phi_\FH.
\end{equation}
We will show that
\begin{equation}\label{eq:A-Meps}
  \norm{(A^\sharp - q) \phi}_\FH \le (r + \eps) \norm \phi_\FH
\end{equation}
holds, which in view of \eqref{eq:variational} implies the second inequality in \eqref{eq:bounds}. Since 
\begin{equation}\label{eq:A-Mxtriag}
  \begin{aligned}
  \norm{(A^\sharp-q)\phi}_\FH &\le \norm{(A^\sharp+B^\sharp - q) \phi}_\FH +
  \norm{B^\sharp \phi}_\FH\\
  &\le r \norm \phi_\FH + \norm{B^\sharp \phi}_\FH,
\end{aligned}
\end{equation}
it is sufficient to estimate $\norm{B^\sharp \phi}_\FH$. For $\beta \le 0$, Lemma
\ref{lem:sharpS0} and Corollary \ref{cor:domain} imply
\begin{equation}
  \norm{B^\sharp \phi}_\FH \le
  \norm{B^\sharp}\norm{\phi}_\FH=\norm{B}\norm{\phi}_\FH = \eps \norm \phi_\FH,
\end{equation}
and \eqref{eq:A-Meps} follows from \eqref{A+B} and \eqref{eq:A-Mxtriag}.

From now on, we consider $\beta > 0$.
By assumption we can choose $\gamma\in (\beta,\alpha)$ such that $A^{\CO\CD} \in
\BS^\gamma$. Let $(\kappa,\er)\in\FE(A^\CD)$ with
\begin{equation}
  \er\geq\max\Bigg\lbrace\bigg(\frac{4\big(1+\snorm{B}{\beta}{0}\big)}{\kappa}\bigg)^{1/(\alpha-\beta)};\quad\bigg(\frac{2\snorm{A^{\CO\CD}}{\gamma}{0}}{\kappa}\bigg)^{1/(\alpha-\gamma)}\Bigg\rbrace.
\end{equation}
As $\phi\in\mathsf \dom(A^\sharp)$, Lemma~\ref{lem:relboundsrho}\eqref{relbound} yields
\begin{equation}\label{eq:relativbdspecificR}
\begin{aligned}
\norm{B^\sharp \phi}_\FH&\leq 2\snorm{B}{\beta}{0}\Big[\er^{\beta-\alpha}\kappa^{-1}
\norm{A^\sharp \phi}_\FH +
\er^\beta\big(1+\kappa^{-1}\snorm{A^\CD}{\alpha}{0}\big)\norm{\phi}_\FH\Big]\\
&\leq \frac{\snorm{B}{\beta}{0}}{2}\Big[ \frac{\norm{(A^\sharp-q)\phi}_\FH
+|q|\|\phi\|_\FH}{1+\snorm{B}{\beta}{0}}+
C\big(1+\snorm{B}{\beta}{0}\big)^{\frac{\beta}{\alpha-\beta}}\norm{\phi}_\FH\Big],
\end{aligned}
\end{equation}
where $C$ is a constant only depending on $A$ and $\beta$. Combining \eqref{eq:A-Mxtriag} and \eqref{eq:relativbdspecificR}, we get
\begin{equation*}
\begin{aligned}
&\frac{2
+\snorm{B}{\beta}{0}}{2\big(1+\snorm{B}{\beta}{0}\big)}\|(A^\sharp-q)\phi\|_\FH
\leq
\bigg[r+\frac{\snorm{B}{\beta}{0}}{2\big(1+\snorm{B}{\beta}{0}\big)}\Big(|q|+C\big(1+\snorm{B}{\beta}{0}\big)^{\frac{\alpha}{\alpha-\beta}}\Big)\bigg]\norm{\phi}_\FH.
\end{aligned}
\end{equation*}
Hence, we arrive at \eqref{eq:A-Meps} with $\epsilon$ as in \eqref{eq:epsa}.

For the first inequality in \eqref{eq:bounds} the only non-trivial case is $\epsilon\leq r$. For all $\phi\in E_{J_{-\eps}}(A^\sharp)\FH\subset\dom(A^\sharp)$ we have
\begin{equation}\label{eq:A-Mr-eps}
\norm{(A^\sharp-q)\phi}_\FH\leq (r-\epsilon)\norm{\phi}_\FH.
\end{equation}
This implies
\begin{equation}
\norm{(A^\sharp+B^\sharp-q)\phi}_\FH\leq
(r-\epsilon)\norm{\phi}_\FH+\norm{B^\sharp \phi}_\FH,
\end{equation}
where in view of \eqref{eq:relativbdspecificR} and \eqref{eq:epsa}
\begin{equation}
\norm{B^\sharp \phi}_\FH \leq \frac{\snorm{B}{\beta}{0}}{2
+2\snorm{B}{\beta}{0}}\Big(r + \abs q+
C\big(1+\snorm{B}{\beta}{0}\big)^{\frac{\alpha}{\alpha-\beta}}\Big)\norm{\phi}_\FH
\leq\epsilon \norm{\phi}_\FH.
\end{equation}
Thus, the first inequality in \eqref{eq:bounds} follows and the lemma is proved. 
\end{proof}

The next lemma deals with perturbations that are `spectrally far' from a given interval. It is a generalisation of \cite[Lemma 11.1]{MorParSht2014} for operators that are not necessarily bounded below.

\begin{lem} \label{lem:spectralperturb}
For $\alpha>0$, $\beta < \alpha$ let $H_0 \in \BD\BE\BS^\alpha$, $B \in \BS^\beta$,
and $A\in\BS^0$ be symmetric operators and set $H:=H_0+B\in\BS\BE\BS^\alpha$. Suppose that there exists a family of orthogonal projections 
$\set{P_l}_{l=0}^L$ with $P_l\in \BS^{-\alpha}$, $0\leq l\leq L-1$, and $P_L\in \BS^0$ that all commute with $H_0$ and satisfy
\begin{equation}
  \sum_{l=0}^L P_l = I, \quad \text{and} \quad A = A P_0, \quad B_{n,l} := P_n B P_l = 0, \quad \text{for} \ \abs{n-l}>1.
  \label{eq:projections}
\end{equation}
Moreover, let 
$J = (q-r,q+r)$ be an interval such that 
\begin{equation}
D_l:=\operatorname{dist}\big(J,\sigma\big((P_l H P_l)^\sharp\big)\big)>0,\quad \text{for all}\ \  0\leq l<L.
\label{eq:distspectrum}
\end{equation}
Finally, assume that
\begin{equation}\label{eq:tecassL}
3^Lr\geq d_L:=\min\limits_{1\leq l<L} D_l
\end{equation}
and
\begin{equation}
  \max_{0 \le l < L} (\norm{B_{l,l-1}} + \norm{B_{l,l+1}})/D_l \le 1/4,
  \label{eq:projectionnorm}
\end{equation}
where we use the convention $B_{0,-1}:=0$.

Then for
\begin{equation}
  \eps := 3^{2 - \frac{L}{2}} \Big(\frac{r}{d_L}\Big)^{1/2} \|A\|
  \label{eq:epschoice}
\end{equation}
we have that
\begin{equation}
  N(J_{-\eps};H) \le N(J;H + A) \le N(J_{\eps};H).
  \label{eq:approximationcontrol}
\end{equation}
\end{lem}
\begin{proof}
We only prove the first inequality; the second inequality follows analogously. 
  It suffices to show that for any $\phi \in E_{J_{-\eps}}(H^\sharp) \FH\subseteq\dom(H^\sharp)=\widetilde{\mathsf{H}}^\alpha$, one has
  $$\norm{(H^\sharp + A^\sharp - q) \phi}_\FH \leq r \norm \phi_\FH.$$ For any $K \in \N$, we split the interval $J_{-\epsilon}$ into
  $2K+1$ subintervals of equal width: for $- K  \le k \le K - 1$, set
  \begin{equation}
    I_k := \left(q + (2k - 1)\frac{(r - \eps)}{2K+1}, q + (2k+1)\frac{(r -
      \eps)}{2K+1}\right]
    \label{eq:intIk}
  \end{equation}
  and 
  \begin{equation}
    I_K:= \left(q + (2K -1)\frac{(r-\eps)}{2K+1}, q + r - \eps\right).
    \label{eq:intIK}
  \end{equation}
  For $\phi \in E_{J_{-\eps}}(H^\sharp)\FH$ and $- K \le k \le K$ define $\phi^{k}:=
  E_{I_k} (H^\sharp)\phi\in\widetilde{\mathsf H}^\alpha$ and
  \begin{equation}
    \eta^{k} := H^\sharp \phi^{k} - \left(q + 2k \frac{(r - \eps)}{2K+1}\right) \phi^k,
    \label{eq:etak}
  \end{equation}
  so that
  \begin{equation}
  \norm{\eta^{ k}}_\FH \le \frac{r}{(2K + 1)} \norm{\phi^{ k}}_\FH
    \label{eq:normetak}
  \end{equation}
holds.
We also introduce 
$$\phi_l^k:=P_l^\sharp\phi_k \quad \text{ and } \quad
\eta_l^k:=P_l^\sharp\eta_k, \quad \text{
for } -K\leq k\leq K \, \text{ and } \, 0\leq l \leq L.$$ For $0\leq l<L$, we clearly have
$P_l^\sharp H^\sharp=(P_lH)^\sharp$ on $\widetilde{\mathsf H}^\infty$ and, since
$P_lH\in\BS^0$, this identity extends to $\widetilde{\mathsf{H}}^\alpha$.
Moreover, $P_l$ commutes with $H_0$ so that \eqref{eq:projections} implies that
on $\widetilde \RH^\alpha$
\begin{align*}
P_l^\sharp H^\sharp &= (P_lH)^\sharp=(P_lHP_l)^\sharp+ B_{l,l-1}^\sharp+B_{l,l+1}^\sharp\\
&=(P_lHP_l)^\sharp P_l^\sharp+ B_{l,l-1}^\sharp P_{l-1}^\sharp+B_{l,l+1}^\sharp
P_{l+1}^\sharp,
\end{align*}
where we use the convention $P_{-1}:=0$. Thus, applying $P_l^\sharp$ to \eqref{eq:etak}, we arrive at
  \begin{equation}
    \eta_l^{ k} = B_{l,l-1}^\sharp \phi_{l-1}^{ k} + \left( (P_lHP_l)^\sharp -
      \left(q + 2
        k\frac{(r-\eps)}{2K+1}
    \right) \right)\phi_l^{ k} + B_{l,l+1}^\sharp \phi_{l+1}^{k},
    \label{eq:etalk}
  \end{equation}
  for $0 \le l <L$, and Lemma \ref{lem:sharpS0} together with
  \eqref{eq:distspectrum} and \eqref{eq:normetak} gives for $0 \le l < L$,
  \begin{equation}
    \begin{aligned}
      \norm{\phi_l^{ k}}_\FH &\le D_l^{-1}\Big( \norm{\eta_l^{ k}}_\FH + \norm{B_{l,l-1}}
      \norm{\phi_{l-1}^{ k}}_\FH + \norm{B_{l,l+1}} \norm{\phi_{l+1}^{ k}}_\FH \Big)\\
      &\le \frac{r}{(2K+1) d_L} \norm{\phi^{k}}_\FH +
      \frac{\norm{\phi_{l-1}^{k}}_\FH +
      \norm{\phi_{l+1}^{k}}_\FH}{4}. 
  \end{aligned}
    \label{eq:normphilk}
  \end{equation}
Recursively for $0 \le l < L$ we deduce that
  \begin{equation}
    \norm{\phi^{k}_l}_\FH \le \frac{2r}{(2K+1) d_L}\norm{\phi^{k}}_\FH +
  \frac{1}{3}\norm{\phi_{l+1}^k}_\FH.
\end{equation}
Hence, employing the trivial bound $\norm{\phi_L^k}_\FH\le \norm{\phi^k}_\FH$, we get that
\begin{equation}
  \norm{\phi_0^k}_\FH \le \left( \frac{3r}{(2K+1) d_L} +
  3^{-L}\right)\norm{\phi^k}_\FH.
  \label{eq:phi0phik}
\end{equation}
In view of Lemma \ref{lem:sharpS0}, it follows that for all $- K \le k \le K$,
  \begin{equation}
    \begin{aligned}
    \norm{A^\sharp \phi^{k}}_\FH &=  \norm{(AP_0)^\sharp \phi^{k}}_\FH\\&=
    \norm{A^\sharp \phi_0^k}_\FH \\&\le  \left( \frac{3r}{(2K+1) d_L} +
    3^{-L}\right) \norm{A}\norm{\phi^k}_\FH,
  \end{aligned}
  \end{equation}
  whence the Cauchy--Schwarz inequality and the Pythagorean theorem yield
  \begin{equation}
    \begin{aligned}
    \norm{A^\sharp\phi}_\FH &\le \sum_{-K \le k \le K} \norm{A^\sharp
    \phi^k}_\FH \\&\le  
    \left( \frac{3r}{(2K+1) d_L} + 3^{-L}\right)\sqrt{2K+1}\norm{A}
    \norm{\phi}_\FH.
  \end{aligned}
  \end{equation}
  We choose
  \begin{equation}
    K = \flo{\frac{3^{L+1}r}{2d_L} - \frac 1 2} + 1,
    \label{eq:deltachoice}
  \end{equation}
  so that $\frac{3^{L+1}r}{d_L}\leq 2K+1\leq \frac{3^{L+2}r}{d_L}$. Then, by \eqref{eq:tecassL},
  we have
  \begin{equation}
    \left(\frac{3r}{(2K + 1)d_L} + 3^{-L}\right)\sqrt{2K+1} \le 3^{2 -
    \frac{L}{2}} \Big(\frac{r}{d_L}\Big)^{1/2}.
    \label{eq:boundsecterm}
  \end{equation}
Consequently, we arrive at
  \begin{equation}
    \begin{aligned}
      \norm{(H^\sharp + A^\sharp - q) \phi}_\FH &\le \norm{(H^\sharp - q )
      \phi}_\FH+ \norm{A^\sharp\phi}_\FH  \\
      &\le \Big( (r - \eps) + 3^{2 - \frac{L}{2}} \Big(\frac{r}{d_L}\Big)^{1/2}\norm{A}\Big)
      \norm{\phi}_\FH\\
      &= r \norm \phi_\FH,
    \end{aligned}
    \label{eq:therightthing}
  \end{equation}
  where we used that $\phi\in E_{J_{-\eps}}(H^\sharp)$ and the value of $\eps$
  given in \eqref{eq:epschoice}.
\end{proof}

\section{Gauge Transform}\label{sec:gt}

Let $\alpha\in \R$ and $A = \Op(a) \in \BS\BE\BS^\alpha$ be symmetric, thus extends to a self-adjoint linear operator
on $\mathsf H^{\alpha}$ by Proposition \ref{prop:closedness and self-adjointness}.
\begin{defi} For every symmetric $\Psi\in\BS^0$, the unitary transformation of $A$ into 
$$[A] := [A]_\Psi := \exp(- \mathrm i \Psi) A \exp(\mathrm i \Psi)$$ is called a
\emph{gauge transform}. 
\end{defi}
\noindent We remark here that, due to Lemma \ref{lem:product} and Corollary \ref{cor:domain}, the series
\begin{equation}
  \label{eq:exppsi}
\exp(\mathrm i\Psi)=\sum\limits_{k=0}^\infty \frac{\big(\mathrm i\Psi\big)^k}{k!}
\end{equation} 
converges both in $\BS^0$ and in the operator norm. In particular, Lemma \ref{lem:complete} implies that $\exp(\mathrm i\Psi)\in\BS^0$, whence $\exp(\mathrm i\Psi)$ is unitary and $[A]_\Psi\in\BE\BS^\alpha$ is symmetric.
The following lemma provides an expansion of $[A]_\Psi$  into a series of multiple commutators of $A$ with $\Psi$, see \eqref{eq:defadk} for the definition of $\ad^k$.
\begin{lem}\label{lem:A'series}
We have 
\begin{equation}\label{eq:H'formalseries}
  [A]_\Psi = \sum_{k = 0}^\infty \frac{1}{k!}\ad^k(A;\Psi),
\end{equation}
where the series converges absolutely in $\BS^\alpha$.
\end{lem}
\begin{proof}
Lemma \ref{lem:product} yields the bounds
\begin{equation}
\snorm{\Psi^jA\Psi^m}{\alpha}{l}\leq \big(\snorm{\Psi}{0}{l}\big)^j\snorm{A}{\alpha}{l}\big(\snorm{\Psi}{0}{l+|\alpha|}\big)^m, \quad \text{for all}\ l\geq 0.
\end{equation}
Thus, the double series 
\begin{equation}
  [A]_\Psi=\sum_{j= 0}^\infty \frac{(-\mathrm i \Psi)^j}{j!} A \sum_{m =0}^\infty \frac{(\mathrm i \Psi)^m}{m!}
  =  \sum_{j,m= 0}^\infty \frac{(-\mathrm i \Psi)^j}{j!} A \frac{(\mathrm i \Psi)^m}{m!}
\end{equation}
converges absolutely in $\BS^\alpha$. Recursively, we obtain
\begin{equation}
  \ad^k(A;\Psi) = k! \sum_{j + m = k} \frac{(-\mathrm i \Psi)^j}{j!} A \frac{(\mathrm i \Psi)^m}{m!}, \quad\text{ for all $k \ge 0$}.
\end{equation}
\end{proof}
In the remainder of this section, we look at gauge transforms that result in an
operator $[A]_\Psi$ that is closer to a diagonal operator (i.e.~an operator in $\BD\BE\BS^\alpha$) than $A$. More precisely, we construct $\Psi$ in such a way that the gauge transform removes as much of the off-diagonal part $A^{\CO\CD}$ from $A$ as possible. Let $\beta<\alpha$ such that $A^\COD\in\BS^\beta$. Then we aim at
\begin{equation}\label{eq:H'goal}
  [A]_\Psi=A^\CD + A^\CR + R,
\end{equation}
where $A^\CR\in\BS^\beta$ is an off-diagonal
\emph{resonant} part (which our transformation cannot eliminate) and $R\in\BS^\gamma$ for some $\gamma <\beta$. The exact form of the operators $A^\CR$ and $R$ depends on the choice of $\Psi$.

As a first step towards \eqref{eq:H'goal}, let us rewrite the series \eqref{eq:H'formalseries} as
\begin{equation}\label{eq:firstorder}
  [A]_\Psi = A^\CD + A^{\CO\CD} + \ad(A^\CD;\Psi) + R,
\end{equation}
with 
\begin{equation}\label{eq:defRconsecgauge}
  R:=\ad(A^{\CO\CD};\Psi)+\sum\limits_{k= 2}^\infty\frac{1}{k!}\ad^k(A;\Psi).
\end{equation}
Suppose that $\Psi \in \BS^\zeta$ with $\zeta \in \R$ and let $\kappa \in \R$. In order to achieve that $R\in\BS^\gamma$ for some $\gamma<\beta$, we can use the following estimates:
\begin{enumerate}[(1)]
\item If $\zeta <0$, then by Lemma \ref{lem:weakcommest} we get
  $\ad(X;\Psi)\in\BS^{\varkappa +\zeta}$ for all $X \in \BS^\varkappa$. We call a gauge transform that only uses these trivial bounds
on the commutator norms \emph{weak}. 
\item Sometimes the structure of the commutators allows us to prove
  $\ad(X;\Psi)\in\BS^{\varkappa +\zeta -\varepsilon}$ for some $\varepsilon >0$
  and appropriate $X \in \BS^\varkappa$.
A gauge transform exploiting this improvement shall be called \emph{strong}.
\end{enumerate} 

As we will see, the main issue with the strong gauge transform is that some
conditions under which it can be used may not be formally invariant under the
use of the gauge transform, which is in general an iterative scheme.
Furthermore, due to combinatorial issues it may be harder to verify that those
conditions are still satisfied as the number of steps increase. However, as we
will see, in many situations it is sufficient to make one step of the strong
gauge transform, and proceed from there with the weak one.

\subsection{The commutator equation}
We recall that after the gauge transform we would like to arrive at the operator $[A]_\Psi$ as in \eqref{eq:H'goal}, in the best possible case with $A^{\CR}=0$. Comparing \eqref{eq:H'goal} with \eqref{eq:firstorder} we obtain that $A^{\CR}=0$ is equivalent to the commutator equation
\begin{equation}\label{eq:comm}
  \ad(A^\CD;\Psi) + A^{\COD} = 0
\end{equation}
for $\Psi = \Op(\psi)$.

Let $A=\Op(a)$ and $\Theta$ be a frequency set for $a$. By \eqref{eq:ad symbol}, equation \eqref{eq:comm} is solved if
$\Theta$ is a frequency set for $\psi$ and
\begin{equation}
a^\CD(\theta \act \xi) \psi_\theta(\xi) -  \psi_\theta(\xi)a^\CD(\xi) = \mathrm i
a^\COD_\theta(\xi)\quad
\end{equation}
holds for all $\theta\in\Theta':=\Theta\setminus \set{\id}$ and $\xi\in\Xi$. This leads to 
\begin{equation}\label{eq:eq}
 \psi_\theta(\xi) = \frac{\mathrm ia^\COD_\theta(\xi)}{a^\CD(\theta\act\xi) - a^\CD(\xi)},
\end{equation}
for $\theta\in\Theta'$ and $\xi\in\Xi$. However, the problem of small
denominators $a^\CD(\theta\act\xi) - a^\CD(\xi)$ for some pairs $(\theta,\xi)$ generally prevents such choice of $\psi$. This motivates the following definition.
\begin{defi}\label{def:resonant}
For $\delta \in \R,\ \resc>0$, and $\theta \in G$, we call a set
$\Lambda_\theta^{\delta,\resc}\subseteq \Xi$ a $\delta$\emph{-resonant region}
generated by $\theta$ for $A^\CD$ if
it satisfies
\begin{equation}
  \Lambda_\theta^{\delta,\resc} \supseteq \set{\xi \in \Xi : \abs{a^\CD(\theta\act\xi) -
  a^\CD(\xi)} \leq \resc \ang \xi^\delta}.
  \label{eq:resonantregion}
\end{equation}
A corresponding \emph{resonance cut-off} is a function $\rco :=\rco^{\delta, s}
: G\times\Xi \to \R$, mapping $(\theta,\xi)\mapsto \rco_\theta^{\delta, s}(\xi)$, such that for all $\theta\in G$, we have
\begin{equation}\label{eq:dennonres}
  \begin{aligned}
0\leq \rco &\leq 1,\\
\rco_{\theta}^{\delta,s}(\xi)&=0,\text{ for all }\xi\in\Lambda_\theta^{\delta,s},\\
\overline{\rco_{\theta^{-1}}(\theta\act\xi)}&=\rco_\theta(\xi),\text{ for all }\xi\in\Xi.
\end{aligned}
\end{equation}
For a fixed resonance cut-off, we define the \emph{resonant} part
$B^{\CR}:=\Op(b^{\CR})$ and the \emph{non-resonant} part
$B^{\CNR}:=\Op(b^{\CNR})$ of any operator $B=\Op(b)\in\BS^\infty$ via their symbols
\begin{equation}
  \label{eq:resnonres}
  \begin{aligned}
    b^\CR &:= b^{\CO\CD}(1 - \rco^{\delta,s}), \\
    b^{\CNR} &:= b^{\CO\CD}\rco^{\delta,s}.
  \end{aligned}
\end{equation}
\end{defi}
\begin{rem}\label{rem:res}\ 
\begin{enumerate}[(i)]
\item For any $\delta\in\mathbb{R}$ and $s>0$, the only $\delta$-resonant region
  generated by $\id$ is $\Lambda_{\id}^{\delta,s}=\Xi$. Hence, every resonance cut-off $\rco$ satisfies $\rco_{\id}\equiv 0$.
\item \label{item:remrescutoff} If $\Lambda_\theta^{\delta,s}$ satisfies
\begin{equation}
\Lambda_{\theta^{-1}}^{\delta,s}=\theta\act\Lambda_{\theta}^{\delta,s},\quad  \text{ for all } \ \theta\in G,
\end{equation}
then the resonance cut-off $\rco$ can be chosen as
\begin{equation}\label{eq:rescutoffstandard}
\rco_\theta(\xi):=\bone_{\Xi\setminus\Lambda_{\theta}^{\delta,s}}(\xi), \quad \text{
for all } (\theta,\xi)\in G\times\Xi.
\end{equation}
\item \label{item:remres} If $B^{\CO\CD}\in\BS^\gamma$, $\gamma\in\R$, then 
\begin{equation}
B^{\CO\CD}=B^{\CNR}+B^{\CR},
\end{equation}
and 
\begin{equation}\label{eq:estpartbyfull}
  \snorm{B^{\CNR}}{\gamma}{l}\leq\snorm{B^{\CO\CD}}{\gamma}{l},\qquad
  \snorm{B^{\CR}}{\gamma}{l}\leq\snorm{B^{\CO\CD}}{\gamma}{l}
\end{equation}
hold for all $l\geq 0$. If $B$ is symmetric, then so are $B^{\CD}$, $B^{\CNR}$, and $B^{\CR}$.
\end{enumerate}
\end{rem}

With the help of Definition \ref{def:resonant}, the problem of small
denominators in \eqref{eq:eq} can be circumvented. Let $\delta\in\mathbb{R}$, $s>0$, and fix a resonance cut-off $\rco$ corresponding to $\delta$-resonant regions $\Lambda_\theta^{\delta,\resc}, \ \theta\in G$, for $A^\CD$. Using \eqref{eq:resnonres}, we define
\begin{equation}\label{eq:defpsiabstract}
  \psi_\theta^{\delta,s}(\xi):=\begin{cases}
    \dfrac{ia^\CNR_\theta(\xi)}{a^\CD(\theta\act\xi)-a^\CD(\xi)}&
    \text{if } \theta\in\Theta',\\[2ex]
    0 & \text{otherwise}.
\end{cases}
\end{equation}
Recall that $A^{\COD}\in\BS^\beta$ so that, in view of Remark
\ref{rem:res}\eqref{item:remres}, $A^{\CNR}\in \BS^\gamma$ for some
$\gamma\leq\beta$. 
\begin{lem}\label{lem:psi}
Let $\gamma\leq \beta$ with $A^{\CNR}\in \BS^\gamma$. Then \eqref{eq:defpsiabstract} defines a symbol $\psi^{\delta,s} \in \BS^{\gamma-\delta}$. The operator $\Psi:=\Op(\psi^{\delta,s})$ is symmetric with 
\begin{equation}\label{eq:estimatepsinorm}
  \snorm{\Psi}{\gamma - \delta}{l} \le \frac 1 \resc \snorm{A^\CNR}{\gamma}{l}, \ \text{for all}\ l\geq 0.
\end{equation}
It satisfies
\begin{equation}\label{eq:modifiedcomm}
\ad(A^\CD;\Psi)+A^{\CNR}=0.
\end{equation}
\end{lem}
\begin{proof}
 The bounds \eqref{eq:estimatepsinorm} follow directly from
 \eqref{eq:defpsiabstract} and \eqref{eq:resonantregion}--\eqref{eq:resnonres}.
 The equation \eqref{eq:modifiedcomm} follows as in
 \eqref{eq:comm}--\eqref{eq:eq} with $\CNR$ replacing $\CO\CD$.
\end{proof}

In view of \eqref{eq:modifiedcomm}, \eqref{eq:firstorder} takes the form
\begin{equation}\label{eq:A'}
[A]_\Psi=A^\CD+A^{\CR}+R,
\end{equation}
with $R$ defined in \eqref{eq:defRconsecgauge}.

\subsection{Weak gauge transform}\label{subsec:weakgauge} 

Let $\gamma\leq\beta$ such that $A^{\CNR}\in\BS^\gamma$, $A^{\COD}\in\BS^\beta$. We choose $\delta
>\gamma$, so that $\gamma -\delta< 0$ in Lemma \ref{lem:psi}. Note that $\delta$
determines the size of the resonant regions and thus the efficiency of the gauge
transform. 
\begin{lem}\label{lem:onestepweakgauge}
Let $\Psi=\Op(\psi^{\delta, s})$ be the operator defined in \eqref{eq:defpsiabstract} and $R$ be as in \eqref{eq:defRconsecgauge}. Then $\ad(A^\COD;\Psi)$, $R\in \BS^{\beta+\gamma - \delta}$ are symmetric and, for all $l\geq 0$,
\begin{equation}\label{eq:comaodpsi}
\snorm{\ad(A^\COD;\Psi)}{\beta+\gamma-\delta}{l}\le \frac{2}{\resc} \snorm{A^\COD}{\beta}{l + \abs{\gamma - \delta}}\snorm{A^\CNR}{\gamma}{l + \abs{\beta}}, 
\end{equation}
and
\begin{equation}\label{eq:Rnormboundsweak}
  \snorm{R}{\beta+\gamma-\delta}l \le  \frac{3}{\resc} \snorm{A^\COD}{\beta}{l + \abs{\gamma - \delta}}\snorm{A^\CNR}{\gamma}{l + \abs\beta+\abs{\gamma - \delta}}
  \exp\left(\frac 2 \resc \snorm{A^\CNR}{\gamma}{l + \abs\beta+\abs{\gamma - \delta}}\right).
\end{equation} 
\end{lem}
\begin{proof}
The estimates \eqref{eq:normestcomm} and \eqref{eq:normestcommzeroorder} together with $\delta>\gamma$ and Lemma \ref{lem:psi}
imply that, for all $k\in\mathbb{N}$,
\begin{equation}
  \begin{aligned}
    \snorm{\ad^k(A^\COD;\Psi)}{\beta+\gamma - \delta}{l} &\le 2
    \snorm{\ad^{k-1}(A^\COD;\Psi)}{\beta}{l + \abs{\gamma - \delta}} \snorm{\Psi}{\gamma - \delta}{l + \abs \beta} \\
    &\le \frac{2^{k}}{\resc} \left(\snorm{\Psi}{0}{l + \abs\beta+\abs
    {\gamma-\delta}}\right)^{k-1}  \snorm{A^\COD}{\beta}{l + \abs{\gamma - \delta}}\snorm{A^\CNR}{\gamma}{l + \abs{\beta}}.
\end{aligned}
  \label{eq:estmultcomm}
\end{equation}
Thus, \eqref{eq:comaodpsi} follows by choosing $k = 1$. Moreover, \eqref{eq:estimatepsinorm} implies that for all $k\in\mathbb{N}$,
\begin{equation}\label{eq:adkODweak}
  \begin{aligned}
    \snorm{\ad^k(A^\COD;\Psi)}{\beta+\gamma - \delta}{l} 
    &\le  \left(\frac{2}{\resc}\snorm{A^{\CNR}}{\gamma}{l + \abs\beta+\abs
    {\gamma-\delta}}\right)^{k}  \snorm{A^\COD}{\beta}{l + \abs{\gamma - \delta}}.
\end{aligned}
\end{equation}
Similarly, we get from \eqref{eq:modifiedcomm} that for all $k\geq 2$,
\begin{equation}\label{eq:adkDweak}
  \begin{aligned}
    \snorm{\ad^k(A^\CD;\Psi)}{\beta+\gamma - \delta}{l} 
    &=\snorm{\ad^{k-1}(A^\CNR;\Psi)}{\beta+\gamma - \delta}{l}\\
    &\le  \left(\frac{2}{\resc}\snorm{A^{\CNR}}{\gamma}{l + \abs\beta+\abs
    {\gamma-\delta}}\right)^{k-1}  \snorm{A^\CNR}{\beta}{l + \abs{\gamma - \delta}}.
\end{aligned}
\end{equation}
Hence, the bounds \eqref{eq:Rnormboundsweak} follow from \eqref{eq:defRconsecgauge}, \eqref{eq:adkODweak}, \eqref{eq:adkDweak}, and \eqref{eq:estpartbyfull}.
\end{proof}
Lemmata \ref{lem:psi} and \ref{lem:onestepweakgauge} have the following
immediate corollary, which follows by choosing $\gamma=\beta$ and applying \eqref{eq:estpartbyfull}.
\begin{cor}\label{cor:onestepweakgauge}
Let $\Psi=\Op(\psi^{\delta, s})$ be the operator defined in \eqref{eq:defpsiabstract} and $R$ be as in \eqref{eq:defRconsecgauge}. Assume that $\delta>\beta$. Then $\Psi \in \BS^{\beta-\delta}$ is symmetric with 
\begin{equation}\label{eq:estimatepsinormcor}
  \snorm{\Psi}{\beta - \delta}{l} \le \frac 1 \resc \snorm{A^\COD}{\beta}{l}, \ \text{for all}\ l\geq 0.
\end{equation}
Moreover, $\ad(A^\COD;\Psi)$, $R\in \BS^{2\beta- \delta}$ are symmetric and for all $l\geq 0$,
\begin{equation}
\snorm{\ad(A^\COD;\Psi)}{2\beta-\delta}{l}\le \frac{2}{\resc} \big(\snorm{A^\COD}{\beta}{l + \abs\beta + \abs{\beta - \delta}}\big)^2,
\end{equation}
and
\begin{equation}\label{eq:Rnormboundsweakcor}
  \snorm{R}{2\beta-\delta}l \le  \frac{3}{\resc} \big(\snorm{A^\COD}{\beta}{l + \abs\beta + \abs{\beta - \delta}}\big)^2 \exp\left(\frac 2 \resc \snorm{A^\COD}{\beta}{l + \abs\beta+\abs{\beta - \delta}}\right).
\end{equation} 
\end{cor}
As a consequence of Lemma \ref{lem:onestepweakgauge} we have arrived at \eqref{eq:H'goal} with $R\in\BS^{\beta+\gamma-\delta}$ and $\beta+\gamma-\delta<\beta$ as desired. One may now iterate the gauge
transform to further reduce the order of the error term, starting from $[A]_\Psi$ in the next step. We call such an iterative scheme
\emph{consecutive gauge transform}. A few remarks on this iterative scheme are in order.
\begin{rem}
\begin{enumerate}[(i)]
\item At each step of the consecutive gauge transform, the resonant regions can be chosen differently. 
\item \label{item:consecsameres} Let us consider a  consecutive gauge transform consisting of $k$ steps, starting with the operator $A_0:=A$, and transforming into the operator 
\begin{equation}
A_j:=\exp(-\mathrm i\Psi_j)A_{j-1}\exp(\mathrm i\Psi_j)=\big[\dots\big[[A]_{\Psi_1}\big]_{\Psi_2}\dots\big]_{\Psi_j}
\end{equation}
at step $j=1,2,\dots, k$. Moreover, suppose for simplicity that $\delta>\beta$ and that $\Lambda_{\theta}^{\delta,s}$, $\theta\in G$, are $\delta$-resonant regions for all $A_j^{\CD}$, $j=0,1,\dots k-1$, simultaneously (so that the resonant cut-off $\rco =\rco^{\delta, s}$ can be chosen at all steps as $\rco_\theta(\xi)=\bone_{\Xi\setminus\Lambda_{\theta}^{\delta,s}}(\xi), \ (\theta,\xi)\in G\times\Xi$). Then \eqref{eq:A'} and Corollary \ref{cor:onestepweakgauge} imply that 
\begin{equation}
A_1:=[A]_{\Psi_1}=A_0^{\CD}+A_0^{\CR}+R_1
\end{equation}
with $R_1\in\BS^{2\beta-\delta}$. Repeating the procedure, we obtain
\begin{equation}
A_j=A_{j-1}^{\CD}+A_{j-1}^{\CR}+R_j,\quad j=1,2,\dots,k,
\end{equation}
with $R_j\in\BS^{\beta+j(\beta-\delta)}$.
\item A disadvantage of the consecutive gauge transform lies in the fact that
  already the operator $A_1=[A]_{\Psi_1}$ may have a frequency set as large as 
\begin{equation}
  Z(\Theta):= \bigcup_{n \in \N} \Theta^n,
\label{eq:ZTheta}
\end{equation}
where
\begin{equation}
  \Theta^n := \underbrace{\Theta \cdot \dotso \cdot \Theta}_{\text{product
  taken } n \text{ times}}.
\end{equation}
This set $Z(\Theta)$ is usually infinite, even when $\Theta$ is finite. Thus, the same holds for $A_1^{\CR}$, which
arises after the second step of gauge transform. This might be inconvenient
since one generally likes to keep the structure of the resonant operators as
simple as possible. However, one can resolve this issue by excluding the terms
belonging to $\BS^{\beta+j(\beta-\delta)}$ from $A_j^{\CNR}$ at the $j$th
step (by moving them to the remainder). Then $\Theta^{j+1}$ will be the frequency set for $A_j^{\CR}$. 
\end{enumerate}
\end{rem}
In the next sub-section, we describe a different iterative gauge transform scheme that we call the \textit{parallel gauge transform}. This is often more convenient to work with than the consecutive gauge transform.
\subsection{Parallel weak gauge transform}
Here, we perform several steps of the gauge transform at the same time, i.e.
\begin{equation}
A^{(\tilde k)}=[A]_{\Psi^{(\tilde{k})}},
\end{equation}
where 
\begin{equation}
\Psi^{(\tilde k)}=\sum\limits_{j=1}^{\tilde{k}}\Psi_j
\end{equation}
for some $\tilde{k}\in\mathbb{N}$. Fix again $\delta\in\mathbb{R}$, $s>0$ and a resonant cut-off $\rco^{\delta,s}$ satisfying \eqref{eq:dennonres} corresponding to $\delta$-resonant regions $\Lambda_\theta^{\delta,s},\ \theta\in G$, for $A^{\CD}$, see Definition \ref{def:resonant}. Following \cite[Section 9]{ParSht2012}, the operators $\Psi_l$, $B_l$, and $T_l$ are recursively defined by
\begin{equation}\label{eq:defb1}
B_1:=A^\COD,
\end{equation}
\begin{equation}\label{eq:defbltl}
  \begin{aligned}
B_l:&=\sum\limits_{j=1}^{l-1}\frac{1}{j!}\sum\limits_{k_1+k_2+\dots
+k_j=l-1}\ad(A^\COD;\Psi_{k_1},\Psi_{k_2},\dots,\Psi_{k_j}),\ l\geq 2,\\
T_l:&=\sum\limits_{j=2}^l\frac{1}{j!}\sum\limits_{k_1+k_2+\dots +k_j=l}\ad(A^\CD;\Psi_{k_1},\Psi_{k_2},\dots,\Psi_{k_j}),\ l\geq 2,
\end{aligned}
\end{equation}
and the relations
\begin{equation}\label{eq:defpsijparalell}
  \begin{aligned}
\ad(A^\CD,\Psi_1)+B_1^{\CNR}&=0,\\
\ad(A^\CD,\Psi_l)+B_l^{\CNR}+T_l^{\CNR}&=0, \ l\geq 2.
\end{aligned}
\end{equation}
More precisely, let $\Theta$ be a frequency set for $A$ and for all $l\geq 1$, let $b_l$ and $t_l$ be the symbols of $B_l$ and $T_l$, respectively. Analogously to \eqref{eq:defpsiabstract}, we solve \eqref{eq:defpsijparalell} by choosing $\Psi_l:=\Op(\psi_l)$ with
\begin{equation}
(\psi_1)_\theta(\xi):=\begin{cases}\dfrac{i(b_1^\CNR)_\theta(\xi)}{a^\CD(\theta\act\xi)-a^\CD(\xi)}&
    \text{if } \theta\in\Theta',\\[2ex]
    0 & \text{otherwise}
    \end{cases}
\end{equation}
and
\begin{equation}
(\psi_l)_\theta(\xi):=\begin{cases}\dfrac{\mathrm i(b_l^\CNR)_\theta(\xi)+\mathrm i(t_l^{\CNR})_\theta(\xi)}{a^\CD(\theta\act\xi)-a^\CD(\xi)}&
    \text{if } \theta\in(\Theta^l)',\\[2ex]
     0 & \text{otherwise}
    \end{cases}
\end{equation}
for $l\geq 2$. Note that for all $l\geq 1$, $\Theta^l$ is a frequency set for $B_l$, $T_l$, and $\Psi_l$.
Finally, put 
\begin{equation}
\begin{aligned}\label{eq:defpsiytildek}
Y_{\tilde{k}}:&=\sum\limits_{l=1}^{\tilde{k}}B_l+\sum\limits_{l=2}^{\tilde{k}}T_l,
\end{aligned}
\end{equation}
and
\begin{equation}\label{eq:defR}
R_{\tilde{k}+1}:=B_{\tilde{k}+1}+R_{\tilde{k}+1}^{(1)}+R_{\tilde{k}+1}^{(2)},
\end{equation}
with
\begin{equation}
\begin{aligned}
R_{\tilde{k}+1}^{(1)}:&=\sum\limits_{j\geq \tilde{k}+1}\frac{1}{j!}\ad^j(A;\Psi^{(\tilde k)})
,\\
R_{\tilde{k}+1}^{(2)}:&=\sum\limits_{j=1}^{\tilde{k}}\frac{1}{j!}\sum\limits_{k_1+k_2+\dots + k_j\geq \tilde{k}+1}\ad(A;\Psi_{k_1},\Psi_{k_2},\dots,\Psi_{k_j}).
\end{aligned}
\end{equation}
Then we arrive at 
\begin{equation}\label{eq:afterparallelgauge}
A^{(\tilde k)}=A^\CD+Y_{\tilde{k}}^{\CD}+Y_{\tilde{k}}^{\CR}+R_{\tilde{k}+1},
\end{equation}
see Lemma \ref{lem:A'series}, where $Y_{\tilde{k}}^{\CR}$ is an operator
with frequency set $\Theta_{\tilde{k}}$. The following Proposition provides norm estimates for the operators after the parallel gauge transform. In particular, it shows that if 
$\delta>\beta$, then we can assure that the error term $R_{\tilde{k}}$ belongs
to classes of arbitrarily small order by choosing $\tilde{k}$ sufficiently large.
\begin{prop}\label{prop:symbolestweakgauge}
Let $A^{\CO\CD}\in\mathbf S^\beta$ with $\delta>\beta$. Then we have for all $l\geq 0$,
\begin{equation}
\begin{aligned}\label{eq:boundspsilbltlparallelweakgauge}
\snorm{\Psi_k}{k(\beta-\delta)}{l}&\ll \big(\snorm{A^\COD}{\beta}{l+n_k}\big)^k,\ k\geq 1\\
\snorm{B_k}{k(\beta-\delta)+\delta}{l}+\snorm{T_k}{k(\beta-\delta)+\delta}{l}&
\ll \big(\snorm{A^\COD}{\beta}{l+n_k}\big)^k,\ k\geq 2,
\end{aligned}
\end{equation}
where $n_k$ is an increasing function of $k$, depending on $k$, $\beta$ and $\delta$, and the
implied constants depend only on $k$, $\beta$, $\delta$, and $s$ in
\eqref{eq:resonantregion}. Moreover, the operators $\Psi^{(\tilde k)}\in\BS^{\beta-\delta}$, $Y_{\tilde{k}}\in\BS^{\beta}$,
$R_{\tilde{k}+1}\in \BS^{\tilde{k}(\beta-\delta)+\beta}$ are symmetric and satisfy the bounds
\begin{equation}\label{eq:estpsiafterparallelgauge}
\begin{aligned}
\big\|\Psi^{(\tilde k)}\big\|^{(\beta-\delta)}_{l}+\snorm{Y_{\tilde{k}}}{\beta}{l}&\ll\big(1+\snorm{A^\COD}{\beta}{l+n_{\tilde{k}}}\big)^{\tilde{k}},\\
\snorm{R_{\tilde{k}+1}}{\tilde{k}(\beta-\delta)+\beta}{l}&\leq C_{A,\tilde{k},\beta,\delta,l},
\end{aligned}
\end{equation}
for all $l\geq 0$, where the implied constants only depend on $\tilde{k}, \beta, \delta$, and $s$; and $C_{A,\tilde{k},\beta,\delta,l}$ is a bounded function
of the symbol norms $\Big\lbrace \snorm{A^\COD}{\beta}{l}\Big\rbrace_{l\geq 0}$, $\tilde{k}$, $\beta$, $\delta$, and $l$.
\end{prop}
\begin{proof}
The bounds \eqref{eq:boundspsilbltlparallelweakgauge} are easily deduced from Corollary
\ref{cor:onestepweakgauge} by induction in $k$, estimating all involved commutators
using \eqref{eq:normestcomm}. The estimates on the symbol norms of $\Psi^{(\tilde k)}$ and 
$Y_{\tilde{k}}$ follow readily. 

Let us prove the estimates on the norms of $R_{\tilde{k}+1}$. Starting with $R_{\tilde{k}+1}^{(1)}$ we note that, for $m\geq \tilde{k}+1$ and $\Psi:= \Psi^{(\tilde k)}$
\begin{equation}
\begin{aligned}
&\snorm{\ad^m(A;\Psi)}{\tilde{k}(\beta-\delta)+\beta}{l}\\ &\leq\snorm{\ad^m(A^\CD;\Psi)}{\tilde{k}(\beta-\delta)
+\beta}{l}+\snorm{\ad^m(A^\COD;\Psi)}{\tilde{k}(\beta-\delta)+\beta}{l}\\
&=\snorm{\ad^{m-1}(Y_{\tilde{k}}^{\CNR};\Psi)}{\tilde{k}(\beta-\delta)+\beta}{l}+\snorm{\ad^m(A^\COD;\Psi)}{\tilde{k}(\beta-\delta)+\beta}{l}\\
&\leq 2^{m-\tilde{k}-1}\snorm{\ad^{\tilde{k}}(Y_{\tilde{k}}^{\CNR};\Psi)}{\tilde{k}(\beta-\delta)+\beta}{l}\big(\snorm{\Psi}{0}{l+|\tilde{k}(\beta-\delta)+\beta|}\big)^{m-\tilde{k}-1} \\
& \ \ \ +2^{m-\tilde{k}}\snorm{\ad^{\tilde{k}}(A^\COD;\Psi)}{\tilde{k}(\beta-\delta)+\beta}{l}\big(\snorm{\Psi}{0}{l+|\tilde{k}(\beta-\delta)+\beta|}\big)^{m-\tilde{k}},
\end{aligned}
\end{equation}
where we apply \eqref{eq:normestcommzeroorder} in the second inequality. Dividing by $m!$ and summing over $m\geq \tilde{k}+1$ we obtain a convergent sum, for which we use the estimates on the norms of $Y_{\tilde{k}}$ and $\Psi^{(\tilde k)}$. Estimating $\snorm{R_{\tilde{k}+1}^{(2)}}{\tilde{k}(\beta-\delta)+\beta}{l}$ is somewhat easier since there are no convergence issues. This finishes the proof of the proposition.
\end{proof}

\subsection{Strong gauge transform} \label{sec:sgt}
The aim of any (iterative) gauge transform scheme is to force the error term $R$ after the gauge transform, see e.g. \eqref{eq:H'goal} or \eqref{eq:afterparallelgauge}, into a class of (relatively) small order. For instance, in  \eqref{eq:firstorder} and \eqref{eq:defRconsecgauge}, we aim at $R \in \BS^\gamma$ for some $\gamma < \beta$. 
If $\Psi$ belongs to a class of negative order, as is the case if
one can choose $\delta > \beta$ in Definition \ref{def:resonant},  for example, 
it can be trivially
satisfied for some $\eta<\beta$, as was seen to be the case with the weak gauge
transform. In some cases, however, one can not guarantee more
than $\Psi \in \BS^0$, whether it be by choosing $\delta = \beta$, or by
introducing additional cut-offs in the definition of $\Psi$.
This is notably the case for Schrödinger-type operators,
whenever the perturbation is not in $\BS^\beta$ for $\beta < \alpha - 1$. In
such a case, one can no longer rely on the trivial product
estimates for $\ad(A^{\COD},\Psi)$ to get that $\ad(A^{\COD},\Psi)\in\BS^\eta$
for some $\eta<\beta$. In the next lemma, we give a sufficient condition that,
nevertheless, yields the required improvement through commuting with $\Psi$.

\begin{lem}
\label{lem:sgt}
 Suppose that $A \in \BE\BS^\alpha$ and that $\Psi \in \BS^0$ is defined as in \eqref{eq:defpsiabstract}. 
If both $\ad(A^{\CO\CD};\Psi) \in \BS^{\gamma}$ and
 $\ad(A^{\CNR};\Psi) \in \BS^{\gamma}$, then, with 
 $$[A]_\Psi = \exp(-i\Psi) A \exp(i\Psi),$$
 we have that
 $$R := [A]_\Psi - A^{\CD} - A^{\CR} \in \BS^{\gamma}.$$
\end{lem}
\begin{rem}
 If, in addition to the assumptions of this lemma, we have $A^{\CO\CD} \not \in
 \BS^{\gamma}$, this would mean that commuting with $\Psi$ has improved order and we can
 therefore call this gauge transform strong. Note that we do not require
 improvements in order to happen at every iteration of the commutator, but only
 at the first step.
\end{rem}

\begin{proof}
  It follows from Lemma \ref{lem:A'series} and equation \eqref{eq:modifiedcomm} that
 \begin{equation}
 \begin{aligned}
  R &= \sum_{k\ge 1} \frac{1}{k!} \ad^k(A^{\CO\CD};\Psi) + \sum_{k \ge 2} \frac{1}{k!} \ad^{k}(A^{\CD},\Psi) \\
  &= \sum_{k\ge 1} \frac{1}{k!} \ad^{k-1}(\ad(A^{\CO\CD};\Psi);\Psi) + \sum_{k \ge 2} \frac{1}{k!} \ad^{k-1}(\ad(A^{\CNR};\Psi),\Psi).
  \end{aligned}
 \end{equation}
As in the proof of Lemma \ref{lem:A'series}, both series converge absolutely in $\BS^{\gamma}$.
\end{proof}

\begin{rem}
 The hypotheses of the previous lemma can be achieved in many ways. The most common one does not depend on the operator $A$, but only on the algebraic structure
 of $\BS^\infty$: it is when commutators naturally improve order. The principal
 example is pseudo-differential operators in $\mathrm{L}^2(\R^d)$ that are
 almost periodic with respect to the translation group $\mathbb{R}^d$. To obtain
 the commutator estimates in this case one requires some limited smoothness in
 $\xi$. We refer to \cite[Lemma 3.4]{Sobolev2005} for a proof, and
 \cite{Sobolev2006,ParSob2010,ParSht2012,MorParSht2014,ParSht2016} for examples
 of further
 applications. In all of these cases, the smooth structure of functions on $\R^d$ was
 used, and the resonance cut-off functions were taken to be smooth
 approximations
 to indicator functions of the non-resonant regions rather the indicators
 themselves.

 It is also possible that one cannot reach $\Psi \in \BS^0$ through only
 non-resonant cut-offs. In such a case, in order to achieve convergence of the
 series for $\exp(i\Psi)$ and $[A]_\Psi$ achieved in \eqref{eq:exppsi} and Lemma
 \ref{lem:A'series} we will need energy cut-offs, i.e. cutting
 off large $\xi$. See \cite{ParSob2010,MorParSht2014} where this idea is being used. 
\end{rem}

\section{Systems of Almost Periodic Operators} \label{sec:mapo}

In this section, we provide a construction suitable to describe almost periodic
operators with matrix-valued symbols within the framework of Sections \ref{sec:operators} -- \ref{sec:gt}.
\subsection{Symbol formalism for systems of almost periodic operators} 
Let the index set $\Xi$ and the group $G$ be as in Section \ref{sec:operators}. Let $m\in\mathbb{N}$ and $\bbb:G\times\Xi\to\CL(\C^m)$ be a function such that there exists a countable frequency set $\Theta=\Theta^{-1}\subseteq G$ with $\bbb_\theta(\xi)=0$ for all $\theta\in G\setminus \Theta$ and $\xi\in\Xi$. Furthermore, assume that
\begin{equation}
\sum\limits_{\theta\in\Theta}\|\bbb_\theta(\xi)\|^2 <\infty,\ \text{for all } \xi\in\Xi,
\end{equation}
where $\|\,\cdot\,\|$ is the operator norm on $\CL(\C^m)$. 
For every $\xi\in\Xi$, let $\set{v_j(\xi):\ j\in \Z/m\Z}$ be an orthonormal basis for $\C^m$ so that $\big\lbrace\be_\xi\otimes v_j(\xi)\big\rbrace_{\xi\in\Xi,j\in \Z/m\Z}$ is an orthonormal basis for $\ell^2(\Xi;\C^m)=\ell^2(\Xi)\otimes\C^m$. In analogy to \eqref{eq:defopB}, an almost periodic operator $\BB$ in $\ell^2(\Xi;\C^m)$ with symbol $\bb$ is defined by
\begin{equation}\label{eq:Basmatrixop}
\BB\big(\be_\xi \otimes v_j(\xi)\big):=\sum_{\theta\in\Theta}\be_{\theta\act\xi} \otimes \big[\bbb_{\theta}(\xi)v_j(\xi)\big].
\end{equation}
We introduce the index set $\Xi_m := \Xi \times \Z/m\Z$ equipped with the
weight function 
\begin{equation}\label{eq:angxim}
 \ang{(\xi,j)}_m:= \ang \xi
\end{equation}
and define the group 
$G_m:= G \times \Z/m\Z$, and its (free) action on $\Xi_m$ by
\begin{equation}
  (g,k)\act(\xi,j):= (g\act \xi,k+j).
  \label{eq:actiongprime}
\end{equation}
Applying the unitary map $T_m:\ell^2(\Xi_m)\to\ell^2(\Xi;\C^m)$ defined by 
\begin{equation}\label{eq:defTm}
T_m\big(\be_{(\xi,j)}\big):= \be_\xi\otimes v_j(\xi),\ (\xi,j)\in \Xi_m,
\end{equation}
we can relate the operator $\BB$ to the operator
\begin{equation}\label{eq:defBfromBB}
B:=T_m^\ast \BB T_m
\end{equation}
in $\ell^2(\Xi_m)$. For every $g\in G$ and $\xi\in\Xi$, let $\left[\bbb_g(\xi)\right]$ be the matrix representation of $\bbb_g(\xi):\C^m\to\C^m$ with
respect to the pair of bases $\set{v_j(\xi)}_j$ and
$\set{v_j(g\act \xi)}_j$ on the domain and codomain, respectively, i.e.
\begin{equation}\label{eq:bbbonbasis}
\bbb_g(\xi)v_j(\xi)=\sum\limits_{k\in\Z/m\Z}\big[\bb_g(\xi)\big]_{kj}v_k(g\act\xi), \quad \text{for all} \ j\in \Z/m\Z.
\end{equation}
Define the scalar symbol $b:G_m\times\Xi_m\to\C$ by
\begin{equation}
b_{(g,k)}(\xi,j):=\left[\bbb_g(\xi)\right]_{k+j\; j} 
  \label{eq:symbolgprime}
\end{equation}
for $(g,k)\in G_m$ and $(\xi,j)\in\Xi_m$, and note that $\Theta\times \Z/m\Z$ is a frequency set for $b$. In view of \eqref{eq:defopB}, \eqref{eq:actiongprime}, and \eqref{eq:symbolgprime}, we have
\begin{equation}
\Op(b) \be_{(\xi,j)} = \sum_{(\theta,k) \in \Theta \times \Z/m\Z}  \left[\bbb_\theta(\xi)\right]_{kj} \be_{(\theta\act\xi,k)}, \ \text{for} \ (\xi,j)\in\Xi_m.
  \label{eq:actionmapo}
\end{equation}
Hence, \eqref{eq:Basmatrixop}, \eqref{eq:defTm}, \eqref{eq:defBfromBB},
\eqref{eq:bbbonbasis}, and \eqref{eq:actionmapo} yield $T_m^* \BB T_m = B = \Op(b)$.
This justifies calling the operators from $\BS^\infty(G_m,\Xi_m)$ \emph{systems
of almost periodic operators}, also known as matrix-valued operators. We shall use the notation
\begin{equation}\label{eq:defsystemclasses}
\BT_m^\gamma:=\BT^\gamma(G_m,\Xi_m),\ \BT\in\{\BS, \BD\BS, \BD\BE\BS, \BS\BE\BS, \BE\BS \}, \ \gamma\in\R\cup\set{\pm\infty},
\end{equation}
and 
\begin{equation}
 \mathsf H_m^\gamma:=\mathsf H^\gamma(\Xi_m),\ \gamma\in\R\cup\set{\pm\infty}.
\end{equation}

Since the map of symbols $\bbb\mapsto b$ is one-to-one, for those $\bbb$ which
are mapped to $b\in\BS_m^\infty$ we write $\Op(\bbb):=\Op(b)$. We use this identification to apply the results of Sections
\ref{sec:operators}--\ref{sec:gt} without always making explicit the conjugation
by the operators $T_m$. 
Note that $\ang{(g,k)}_m=\ang{g}$, for all $(g,k)\in G_m$, see \eqref{eq:modulusg} and \eqref{eq:angxim}. Hence, for $b\in\BS_m^\beta$, we have the equivalence of norms,
\begin{equation}
  \mathfrak{c}_m^{-1}\snorm{b}{\beta}{l}\leq\sum_{\theta \in \Theta} \ang \theta^l \sup_{\xi \in \Xi} \ang\xi^{-\beta}\norm{\bbb_{\theta}(\xi)}\leq \mathfrak{c}_m\snorm{b}{\beta}{l},
  \label{eq:matrixseminorm}
\end{equation}
where
the constant $\mathfrak{c}_m>0$ only depends on $m$.
Two sub-algebras of $\BS_m^\infty$ will be of particular interest in the sequel:
uncoupled operators and diagonal operators.
\begin{defi}\label{def:uncoupled}
The \emph{uncoupled operators} in $\BS_m^\beta$, $\beta\in\R\cup\set{\pm\infty}$, are defined by 
\begin{equation}
  \begin{aligned}
\BU\BS_m^\beta:=\big\{\BB\in\BS_m^\beta: \ &\text{the matrix } \big[\bbb_g(\xi)\big], \text{ see \eqref{eq:bbbonbasis},} \text{is diagonal for all }g\in G,\ \xi\in\Xi\big\}\\ =\big\{\BB\in\BS_m^\beta: \ &\text{the frequency set for }b\text{ can be}\text{chosen as a subset of }G\times\{0\}\big\}.
\end{aligned}
\end{equation}
For any operator $\BA =\Op(\ba)\in \BS_m^\alpha$,
$\alpha\in\R\cup\set{\pm\infty}$, we define the symbol
\begin{equation}
  \label{eq:uncoupled}
\big[\ba_g^\CU(\xi)\big]_{k j}:= \begin{cases}
                                           \big[\ba_g(\xi)\big]_{k j} & \text{if
                                           }k =j\\ 0 & \text{if } k\neq j
                                          \end{cases},\quad\text{ for all }g\in
                                          G, \xi\in \Xi, k,j\in\Z/m\Z.
\end{equation}
We write $\BA^\CU := \Op(\ba^\CU)$ for the projection of $\BA$ onto
$\BU\BS_m^\alpha$ which we call the uncoupled part. We also denote $\BA^{\CC} := \BA - \BA^{\CU}$ its
coupled part. It can easily be seen that if $\BA \in
\BS_m^\gamma$, $\gamma \in \R$ then for all $l \ge 0$
\begin{equation}
  \label{eq:estpartuncbyfull}
  \snorm{\BA^\CU}{\gamma}{l} \le \snorm{\BA}{\gamma}{l}, \quad
  \snorm{\BA^\CC}{\gamma}{l} \le \snorm{\BA}{\gamma}{l}
\end{equation}
and that if $\BA$ is symmetric so are $\BA^\CU$ and $\BA^\CC$.
\end{defi}
The second sub-algebra is $\BD\BS_m^\infty$, see \eqref{eq:defsystemclasses} and Definition \ref{def:dapo}. Noting that $\id_{G_m}=(\id_G,0)$, we infer from \eqref{eq:symbolgprime} that 
\begin{equation}
\BD\BS_m^\infty=\big\{B\in\BU\BS_m^\infty:\ \bbb_g(\xi)=0\ \text{for all} \ g\in G\setminus \set{\id_G}\big\},
\end{equation}
so that $\BD\BS_m^\infty \subseteq \BU\BS_m^\infty \subseteq \BS_m^\infty$. As
in the scalar case,
for any operator $\BA =\Op(\ba)\in \BS_m^\alpha$, $\alpha\in\R\cup\set{\pm\infty}$, we denote by
\begin{equation}
  \label{eq:diagonalmatrix}
  \begin{split}
\BA^\CD &:=\Op(\ba^\CD)\quad\\  \big[\ba_g^\CU(\xi)\big]_{k j}&:= \begin{cases}
                                           \big[\ba_g(\xi)\big]_{k j} & \text{if
                                           }k =j
                                           \text{ and } g = \id\\ 0 &
                                           \text{otherwise}.
                                          \end{cases}
\end{split}\end{equation}
and $\BA^{\CO\CD} = \BA - \BA^\CD$. This definition makes it so that $\BA^\CD
= T_m A^\CD T_m^*$. Similarly, if resonant and non-resonant regions are defined
in terms of $\Xi_m$, we set $\BA^\CR$ = $T_mA^\CR T_m^*$ and $\BA^{\CNR} = T_m A^{\CNR}
T_m^*$. We can also combine notions of coupling and resonance; we set for
instance $\BA^{\CR,\CU} = (\BA^{\CR})^\CU$, and proceed similarly for other
combinations of the indices.

The following lemma is useful when changing the reference orthonormal
basis of $\C^m$. Nevertheless, for the rest of this section the reference basis of
$\ell^2(\Xi;\C^m)$ will remain fixed as $\{\be_\xi\otimes v_j(\xi)\}_{(\xi,j)\in\Xi_m}$.
\begin{lem}
Assume that, for every $\xi\in\Xi$, the set $\set{u_j(\xi):\ j\in \Z/m\Z}$ is an orthonormal basis for $\C^m$. Then the unitary operator
\begin{equation}\label{eq:defBU}
\BU:\ell^2(\Xi;\C^m) \to \ell^2(\Xi;\C^m), \quad \be_\xi\otimes v_j(\xi)\mapsto \be_\xi\otimes u_j(\xi)
\end{equation}
satisfies $T_m^*\BU T_m\in\BS_m^0$.
\end{lem}
\begin{proof}
It is clear from \eqref{eq:defBU} that $T_m^*\BU T_m=\Op(\bu)$ where
$\bu_g(\xi)\in \operatorname{U}(m)$ is unitary for all $g\in G$, $\xi\in\Xi$,
and $\set{\id_G}$ is a frequency set for $\bu$. Thus the equivalence of norms
\eqref{eq:matrixseminorm} implies that $T_m^*\BU T_m\in\BS_m^0$. 
\end{proof}
The previous lemma has the following corollary, justifying our terminology of
uncoupled operators.
\begin{cor}
  Let $\set{v_j:\ j\in\Z/m\Z}$ be a fixed basis for $\C^m$. Then any operator $A \in \BU\BS_m^\infty$ is unitarily equivalent to an orthogonal sum $\bigoplus\limits_{j\in \Z/m\Z}A_j$
 where for every $j\in\Z/m\Z$, $A_j$ acts in $\ell^2(\Xi) \otimes \spann\set{v_j}$.
\end{cor}

\subsection{Gauge transform in \texorpdfstring{$\BS_m^\infty$}{Sm}: the reduction to
uncoupled operators} \label{sec:reducscalar}

We would like to apply a weak gauge transform to an operator in the class
$\BS\BE\BS^\infty_m$ --- cf. \eqref{eq:defsystemclasses} ---
in order to obtain an operator of the same ordrer with an
uncoupled principal symbol. In this section, we give two sufficient conditions
that allow us to do this.
The first one is more restrictive on the
off-diagonal part and gives a non-trivial remainder, but allows for a principal
symbol with multiple eigenvalue.
The second one requires the principal symbol to have only simple eigenvalues,
in which case the procedure is more efficient and the restrictions on the
off-diagonal symbol are much milder.

\begin{thm}
  \label{thm:systemonestep}
  Let $\BA = \Op(\ba) \in \BS\BE\BS_m^\alpha$ be symmetric and let $\beta <
  \alpha$ be such that $\BA^{\CC}:= \Op(\ba^{\CC}) \in \BS^\beta_m$. Assume that
  \begin{equation}
    \big[\ba_{\id}\big]\left( \xi \right) = \ang \xi^\alpha
    \diag(a_1(\xi),\dotsc,a_m(\xi))
  \end{equation}
  and that $\Theta$ is a frequency set for $\ba^\COD$. Here, for $j \in \Z/m\Z$,
  $a_j : \Xi \to \R$ are bounded functions such that for all $\theta \in
  Z(\Theta) = \bigcup_{k=1}^\infty \Theta^k$,
  \begin{equation}
    \label{eq:shftinfty}
    \lim_{\ang \xi \to \infty} \frac{a_j(\theta \act \xi)}{a_j(\xi)} = 1.
  \end{equation}
  Suppose finally that
  there exists $C, c> 0$ such that for every $j \in \Z/m\Z$ and $k \in \Z / m \Z
  \setminus \set{0}$, either
  \begin{equation}
  \label{eq:ajboundfrombelow}
  \inf_{\ang \xi > C} \abs{a_j(\xi) - a_{j+k}(\xi)} \ge c > 0,
  \end{equation}
  or
  \begin{equation}
    \label{eq:hypresonant}
    \big[\ba^{\COD}]_{j,j+k} \in \BS^{2\beta - \alpha}.
  \end{equation}
  Then, for all $\eps > 0$ and $N \in \N$ there exists a symmetric operator
  $\BPsi \in \BS_m^{\beta - \alpha}$ such that
  \begin{equation}
    \label{eq:desiredafteronestep}
    [\BA]_{\BPsi} = \exp(-i\BPsi) \BA \exp(i\BPsi) = \BA^\CD + \BY + \BR_1 + \BR_2
  \end{equation}
  where $\BY \in \BU\BS_m^\beta$, $\BR_1 \in \BS^{2 \beta - \alpha}$,
  $\norm{\BR_2}_{\mathsf H^\beta_m \to \mathsf H_m^0} < \eps$ and $\BY, \BR_1,
  \BR_2$
  are symmetric. If $\BA$ is quasi-periodic, one can choose $\BR_2 = 0$.
\end{thm}

\begin{rem}
  The conditions \eqref{eq:ajboundfrombelow} and \eqref{eq:shftinfty} are
  satisfied in the simple case of constant functions $a_j(\xi)=a_j\in\R\setminus
  \{0\}$.
\end{rem}
\begin{proof}
Fix $\epsilon' > 0$. We first eliminate the long-range coupling. Since
$\snorm{\BA^{\COD}}{\beta}{0} < \infty$, there exists a
finite subset $\tilde{\Theta} \subseteq \Theta$, closed under inversion and containing the identity, such that
\begin{equation}
  \sum_{\theta \in \Theta \setminus \tilde \Theta}
  \sup_{\xi \in \Xi} \ang \xi^{-\beta}\norm{\ba_\theta^{\COD}(\xi)} < \epsilon'.
  \label{eq:smallnorm}
\end{equation}
Let $\BB:= \Op(\bbb)$ with the symbol
\begin{equation}
  \bbb_\theta(\xi) := \begin{cases}
    \ba_\theta^{\COD}(\xi) & \text{if } \theta \in \tilde \Theta, \\
    0 & \text{otherwise.}
  \end{cases}
  \label{eq:btilde}
\end{equation}
For $\tilde \BR: = \BA^{\COD} - \BB$, \eqref{eq:matrixseminorm} implies
\begin{equation}
  \label{eq:estRtilde}
\big\|\widetilde{\BR}\big\|_{0}^{(\beta)} < \mathfrak{c}_m\epsilon',
\end{equation}
and we write $\widetilde{\BA}:= \BA^\CD + \BB$ so that
$\BA=\widetilde{\BA}+\widetilde{\BR}$ and $\tilde \BA^\CD = \BA^\CD$. For every $j \in \Z/m\Z$, define the set
\begin{equation}
  \label{eq:specificconstants}
  I_j := \set{k \in \Z/m\Z : \eqref{eq:ajboundfrombelow} \text{ holds}}.
\end{equation}
Finiteness of $\tilde \Theta$ and bounded range of action imply that
\begin{equation}
  \lim_{\ang \xi \to \infty} \sup_{\theta \in \tilde \Theta} \abs{\frac{\ang{\theta
  \act \xi}}{\ang \xi} -1} = 0.
\end{equation}
Combining this with \eqref{eq:shftinfty} and
\eqref{eq:ajboundfrombelow}, as well as boundedness of the functions $a_j$
implies the existence of $s'$ depending on $\eps'$ and the constants $c, C$ in
\eqref{eq:ajboundfrombelow} such that
\begin{equation}
  \inf_{j \in \Z / m \Z} \, \, \inf_{k  \in I_j} \, \,  \inf_{\theta \in
  \tilde \Theta} \, \,
  \inf_{\ang \xi > s'} \Big|a_{j+k}(\theta \act \xi) \ang{\theta \act \xi}^\alpha
  - a_j(\xi) \ang \xi^\alpha\Big| > \frac c 2 \ang \xi^\alpha.
\end{equation}
Thus, for $(g,k) \in G_m$, the sets
\begin{equation}
  \label{eq:resregionsystem}
  \Lambda_{(g,k)}^{\alpha,c/2} := \begin{cases}
    \set{(\xi,j) \in \Xi_m : \min\set{\ang \xi, \ang{g \act \xi}} \le s'}, &
    \text{if } g \in \tilde \Theta \text{ and } k \in I_j, \\
  \Xi_m, & \text{otherwise}
  \end{cases}
\end{equation}
are $\alpha$-resonant regions for the operator $\tilde A = T_m^* \BA T_m$, 
cf.~\eqref{eq:resonantregion}, and we choose the corresponding (scalar) resonance cut-off function
\begin{equation}
  \chi_{(g,k)}(\xi,j):=\bone_{\Xi_m\setminus\Lambda_{(g,k)}^{\alpha,c/2}(\xi,j)},
\end{equation}
see Remark \ref{rem:res}\eqref{item:remrescutoff}. Thus, taking $\Psi$ as in Lemma \ref{lem:psi}, we have that $\BPsi = T_m \Psi
T_m^* \in \BS_m^{\beta
- \alpha}$. In view of \eqref{eq:A'}, we deduce that
\begin{equation}
  [\tilde \BA]_{\BPsi} = \exp(-i \BPsi) \tilde \BA \exp(i \BPsi) =  \BA^\CD +
  \tilde \BA^\CR + \BR,
\end{equation}
where Corollary \ref{cor:onestepweakgauge} and conjugation by $T_m$ give $\BR \in \BS_m^{2 \beta -
\alpha}$.
We turn our attention to $\tilde \BA^\CR$. We decompose it as $\tilde \BA^\CR =
\tilde \BA^{\CR,\CU} + \tilde \BA^{\CR,\CC}$. By definition of the resonant region
$\Lambda^{\alpha,c/2}_{(g,k)}$, we have that
\begin{equation}
  [\tilde \ba^{\CR,\CC}]_{j,j+k} =
    [\tilde \ba^{\CO\CD,\CC}]_{j,j+k} \quad \text{if } k \not \in I_j 
\end{equation}
and
\begin{equation}
\supp([\ba^{\CR,\CC}]_{j,j+k}) \subset \set{\xi \in \Xi : \min_{\theta \in
\tilde \Theta} \ang{\theta \act \xi} \le s'} \quad \text{if } k \in I_j.
\end{equation}
By \eqref{eq:hypresonant}, for every $k \not \in I_j$ we have $[\tilde
\ba^{\CR,\CC}]_{j,j+k} \in \BS^{2\beta - \alpha}$. 
Finiteness of $\tilde
\Theta$ and bounded range of action imply that 
the support
of $[\ba^{\CR,\CC}]_{j,j+k}$ is bounded for $k \in I_j$.
Together, along with Proposition
\ref{prop:normorder}, this gives $\tilde \BA^{\CR,\CC} \in \BS_m^{2 \beta - \alpha}$.

All of this implies
\begin{equation}
  \label{eq:decompo}
  [\BA]_{\BPsi} = \BA^\CD + \tilde \BA^{\CR,\CU} + \tilde \BA^{\CR,\CC} + \BR +
  \exp(-i\BPsi) \tilde
  \BR \exp(i \BPsi).
\end{equation}
We claim that this has the desired form \eqref{eq:desiredafteronestep} with 
\begin{equation}
  \label{eq:decompodetails}
\BY =
\tilde \BA^{\CR, \CU},  \qquad \BR_1 = \tilde \BA^{\CR,\CC} + \BR, \qquad
\text{and} \qquad \BR_2 =
\exp(-i \BPsi)
\tilde \BR \exp(i\BPsi).
\end{equation}
Indeed, it follows from \eqref{eq:estpartbyfull} and
\eqref{eq:estpartuncbyfull} that $\tilde \BA^{\CR,\CU} \in \BU\BS^\beta_m$. We
have that $\tilde \BA^{\CR,\CC}, \BR \in \BS_m^{2\beta - \alpha}$ so that their sum
$\BR_1$ also is. 

Finally, recall that $\exp(i\BPsi)\in\BS^0_m$ by Corollary \ref{cor:complete}. In particular, we have
\begin{equation}
\snorm{\exp(i\BPsi)}{0}{|\beta|}\leq \sum\limits_{k=0}^\infty
\frac{\big(\snorm{\BPsi}{0}{|\beta|}\big)^k}{k!},
\end{equation}
where by Corollary \ref{cor:onestepweakgauge} and conjugation with $T_m$ we have
\begin{equation}
\begin{aligned}
  \snorm{\BPsi}{0}{|\beta|}\leq \snorm{\BPsi}{\beta-\alpha}{|\beta|}\le
  \frac{4}{c}
\big\|\BB\big\|_{|\beta|}^{(\beta)} \leq \frac 4 c
\snorm{\BA^{\COD}}{\beta}{|\beta|}
\end{aligned},
\end{equation}
with $c$ the constant in \eqref{eq:ajboundfrombelow}. Consequently,
$\snorm{\exp(i\BPsi)}{0}{|\beta|}$ is bounded uniformly in $\epsilon'\searrow
0$. Hence, for any $\epsilon>0$, choosing
\begin{equation}
  0< \eps' < \frac{\eps}{\left(\snorm{\exp(\mathrm i\BPsi)}{0}{|\beta|}\right)^2
  \mathfrak c_m},
\end{equation}
we obtain by Lemma \ref{lem:domain} and \eqref{eq:estRtilde} that
\begin{equation}\label{eq:normremainder}
\begin{aligned}
\norm{\BR_2}_{\RH_m^{\beta} \to \RH_m^0}&\leq \norm{\exp(
i\BPsi)}_{\RH_m^0\to\RH_m^0}\big\|\widetilde{\BR}\big\|_{\RH_m^{\beta} \to
\RH_m^0}\norm{\exp(i\BPsi)}_{\RH_m^{\beta}
\to \RH_m^{\beta}}\\
&\leq \big\|\widetilde{\BR}\big\|_{0}^{(\beta)}\left(\snorm{\exp(
i\BPsi)}{0}{|\beta|}\right)^2<\epsilon.
\end{aligned}
\end{equation}
This finishes the proof.
\end{proof}

\begin{thm}
  \label{thm:gtsystem}
  Let $\BA = \Op(\ba) \in \BS\BE\BS_{m}^\alpha$ be symmetric and let
  $\beta<\alpha$ such that $\BA^{\CO\CD}:=\Op(\ba^{\COD}) \in
  \BS^\beta_m$ with frequency set $\Theta \subset G$. Assume that
\begin{equation}
  \big[\ba_{\id}(\xi)\big]=\ang{\xi}^\alpha\diag(a_1(\xi),\dotsc,a_m(\xi))
\end{equation}
for some bounded functions $a_j:\Xi\to\mathbb{R}$. Moreover, suppose that there exist $C, c>0$ such that
\begin{equation}
 \label{eq:ajboundfrombelowbis}
\inf\limits_{\ang \xi >C}\min\limits_{j\neq k} |a_j(\xi)-a_k(\xi)|\geq c>0,
\end{equation} 
and that, for all $j=1,2,\dots,m$, and $\theta\in Z(\Theta)=\bigcup\limits_{k=1}^\infty \Theta^k$,
\begin{equation}\label{eq:shiftsatinfinity}
\lim\limits_{\ang{\xi}\to\infty} \frac{a_j(\theta\act\xi)}{a_j(\xi)}=1.
\end{equation}
Then, for all $\epsilon>0$ and $N \in \N$ there exists a symmetric operator
$\BPsi\in\BS^{\beta-\alpha}_m$ such that 
\begin{equation}
  [\BA]_{\BPsi} = \exp( i \BPsi) \BA \exp(i\BPsi) = \BA^{\CD} + \BY^{\CU} +
  \BR_1 + \BR_2
  \label{eq:matgt}
\end{equation}
with $\BY\in\BS^\beta_m$, $\BR_1 \in \BS_m^{-N}$, $\norm{\BR_2}_{\mathsf H_m^{\beta}\to \mathsf H_m^{0}} < \epsilon$
, and $\BY$, $\BR_1$, $\BR_2$ symmetric. If $\BA^{\COD}$ is quasi-periodic, then
one can choose $\BR_2=0$.
\end{thm}
\begin{proof}
  The proof essentially follows the scheme of the proof of Theorem
  \ref{thm:systemonestep}. We first eliminate
  long-range coupling and find $\BB \in \BS_m^\beta$ and $\widetilde \BR$ such that $\BA = \BA^\CD  +
  \BB +
  \widetilde \BR$ and $$\snorm{\widetilde \BR}{\beta}{0} < \eps'.$$ 
  Assumption \eqref{eq:ajboundfrombelowbis} leads this time to $\alpha$-resonant regions
  \begin{equation}
    \label{eq:resmultstep}
    \Lambda^{\alpha,c/2}_{(g,k)} = \begin{cases}
      \set{(\xi,j) \in \Xi_m : \min\set{\ang \xi, \ang{g \act \xi} } \le s'}, &
      \text{if } k \ne 0, \\
      \Xi_m, & \text{if } k = 0
    \end{cases}
  \end{equation}
  for some $s'$ depending on $\eps'$. 
  Put
  \begin{equation}
    K := \frac{N + \beta}{\alpha - \beta}.
  \end{equation}
  We apply a parallel weak gauge
  transform according to \eqref{eq:afterparallelgauge}. We have from Proposition
\ref{prop:symbolestweakgauge} and conjugating by $T_m$ that there exists symmetric operators $\BPsi \in \BS_m^{\beta - \alpha}$, $\BY \in \BS_m^\beta$ and $\BR
\in \BS_m^{-N}$ such that
\begin{equation}
  [\tilde \BA]_{\BPsi} = \exp(-i \BPsi) \tilde A \exp(i \BPsi) =  \BA^\CD + \BY^\CD + \BY^\CR + \BR
\end{equation}
and
\begin{equation}
  \label{eq:boundsnormpsi}
  \snorm{\BPsi}{\beta - \alpha}{\abs{\beta}} \ll \left(1 +
  \snorm{\BA^{\COD}}{\beta}{n_K}\right)^K.
\end{equation}
where the implicit constant depends only on $c$ and $K$.
Inequality \eqref{eq:boundsnormpsi} implies that $\|\exp(i\BPsi)\|_{\abs \beta}^{(0)}$ is
uniformly bounded as $\eps' \searrow0$.

With the resonant region as in \eqref{eq:resmultstep}, for every $\BY
  \in \BS_m^\gamma$, $\gamma < \alpha$ we have that $
\BY^\CD+   \BY^\CR = \BY^\CU + \BR_\BY$, where the symbol of $\BR_\BY$ has bounded support,
  implying $\BR_\BY \in \BS_m^{-\infty}$. We put $\BR_1 = \BR + \BR_\BY \in
  \BS^{-N}_m$ and
  $\BR_2 = \exp(-i\BPsi) \tilde \BR \exp(i \BPsi)$ gives us
  \begin{equation}
    \norm{\BR_2}_{\RH_m^\beta \to \RH_m^0} \le \eps' 
    \left( \snorm{\exp(i\BPsi)}{\abs
    \beta}{0}\right)^2.
  \end{equation}
  Therefore for any $\eps > 0$ by choosing $ 0 < \eps' < \eps
  \left( \snorm{\exp(i\BPsi)}{\abs
  \beta}{0}\right)^{-2}$
  we obtain
  \begin{equation}
  [\BA]_{\BPsi} = \BA^\CD + \BY^\CU + \BR_1 + \BR_2
  \end{equation}
  with the claimed properties.

\end{proof}

\section*{\textsc{Part II : Applications to asymptotic properties of systems}}

In this second part, we consider some specific examples where the methods and
results developed in the first half are applicable. As was mentioned earlier,
these methods work very well for operators
$$
H = H_0 + B
$$ 
of Schrödinger type acting on $\RL^2(\R^d)$. Here, $H_0 = (-\Delta)^{\alpha/2}$ and
$B$ is a pseudo-differential perturbation of order $\beta < \alpha$. In
particulars, the gauge transform method allows us to solve the following two
types of problems, see
\cite{barbpar,MorParSht2014,parnovski,ParSht2016,ParSht2012,ParSob2010,Sobolev2006} :
\begin{itemize}
  \item obtain a complete asymptotic expansion
    for the integrated density of states of almost periodic operators, as the
    spectral parameter goes to infinity;
  \item Prove that some
    elliptic periodic operators have the Bethe--Sommerfeld property, which
 asserts that the spectrum of such operators
contains a half-line $[\lambda;\infty)$ for some $\lambda \in \R$. 
\end{itemize}

We now consider these questions in the setting of elliptic systems of operators. We establish answers to both of these
problems in the case where symbols are
periodic, for the Bethe-Sommerfeld property, and almost periodic, for the
integrated density of states. We will do so by using the tools developed in Part I
of this paper to reduce these operators to uncoupled operators. We will show
that such a reduction cannot change the integrated density of states too much,
and we will show that it cannot open infinitely many gaps in the spectrum.
Since elliptic systems of operators do not have to be semi-bounded,
we will obtain these results as the spectral parameter goes to $\pm
\infty$. In order to do this, we will establish quantitative estimates based
upon the
results of Sections \ref{sec:ids}, \ref{sec:gt} and \ref{sec:mapo} under generic
assumptions about the perturbations. 

In Section
\ref{sec:besicovitch}, we describe the Besicovitch space of almost periodic
functions, and the operators acting on it. We also describe the structure of
operators that are periodic rather than simply almost periodic, interpreting the Floquet-Bloch decomposition through the lens of
almost periodic functions. 

In Section \ref{sec:prelim}, we describe the approach and the conditions
required to prove the existence of complete asymptotics for the integrated
density of states (IDS), and we state Theorem \ref{thm:aexpconcrete}, which describes
the asymptotic behaviour of the IDS for elliptic systems of operators. We
reduce the problem to obtaining asymptotics for the IDS in families of
intervals, which is the statement of Theorem \ref{thm:ainint}.

In Section \ref{sec:asyexp}, we prove Theorem \ref{thm:ainint} by showing
that it holds for uncoupled operators. We then use the gauge transform to show
that it is sufficient to obtain a complete asymptotic expansion for uncoupled
operators to get one for general systems.

In Section \ref{sec:bs}, we change perspective and we study periodic operators.
In Theorem \ref{thm:bs}, we give conditions under which elliptic systems of
periodic operators enjoy the Bethe--Sommerfeld property. We then
use the reduction to uncoupled operators and bounds for the density of states
obtained in Section \ref{sec:asyexp} to show that it is sufficient to prove
that the spectral overlap function is sufficiently bounded away from $0$ for uncoupled
operators. This will be done by reusing the results of Section \ref{sec:ids},
but interpreting fibrewise eigenvalue counting functions as instances of the
IDS.

We prove those lower bounds in Section \ref{sec:cg} by refining arguments based
on combinatorial geometry that were previously used in proving the
Bethe-Sommerfeld conjecture for Schrödinger-type operators.

Finally in Section \ref{sec:dirac}, we spend a few words to show that periodic
and almost periodic perturbations of the Dirac operator fit in the framework that
we have described in this part.

\section{Besicovitch space and systems of operators} \label{sec:besicovitch}

In this section, we turn back to the space $\RB^2(\R^d;\C^m)$ of almost periodic
vector-valued functions, corresponding to the case where $G = \Xi = \R^d$  and 
\begin{equation}
\bg_1\bg_2:=\bg_1+\bg_2, \quad \bg\act\bxi:=\bg+\bxi,
\end{equation}
for all $\bg_1,\bg_2,\bxi\in\mathbb{R}^d$. The weight function is $\ang \bxi = 1
+ \abs \bxi$.  From \eqref{eq:modulusg} we also get that $\ang \bg = 1 + \abs
\bg$ and that $G$ has bounded range of action. The case $m = 1$ corresponds to the usual
Besicovitch space. We now offer a concrete description of this space, along with
a few results relating the properties of operators acting on $\RL^2$ and
$\RB^2$. These results can be found in
\cite{CoburnMoyerSinger,Shubin1978,Shubin1979}.

Let $\set{v_1,\dotsc,v_m}$ be an orthonormal basis for $\C^m$ and for $1 \le j
  \le m$ let 
\begin{equation}
  \be_{\bxi,j}(\bx):= \exp(i \bxi \cdot \bx) \otimes v_j,
\end{equation}
  The space $\RB^2(\R^d;\C^m)$ is the closure of
\begin{equation}
  \spann \set{\be_{\bxi,j}:\ \bxi\in\R^d,\ j=1,\dots ,m},
\end{equation}
taken with respect to the
inner product 
\begin{equation}
  (f,g)_{\RB} = \lim_{L \to \infty} \frac{1}{(2L)^d}\int_{[-L,L]^d} f \cdot \bar g \de \bx.
\end{equation}

For the remainder of this article, we will use $\BS_m^\infty$,
$\BD\BS_m^\infty$, etc. to refer to the
spaces of almost periodic operators acting on $\RB^2(\R^d;\C^m)$. 
Let $A$ be an operator in $\BS_m^\alpha$ with symbol $\ba(\bx,\bxi)$.  
The action of $\BA$ in $\RL^2(\R^d;\C^m$ as an operator in the H\"ormander class
  $\Psi^\alpha(\R^d;\C^m)$ with almost periodic symbol is defined through the
  usual Fourier integral representation of pseudo-differential operators (see
  e.g. \cite{HormanderIII}) as
$$
\widetilde {\operatorname{Op}}(\ba)f(\bx) = \frac{1}{(2\pi)^d}\iint_{\R^d \times \R^d} \exp(i \bxi \cdot( \bx-\by))
\ba(\bx,\bxi) f(\by) \de \by \de \bxi.
$$
The following
proposition 
links its properties as an operator in $\RL^2$ and $\RB^2$, respectively.

\begin{prop} \label{prop:normspec}
 If $\BA \in \BS_m^\infty$ is bounded or elliptic, then
 \begin{equation}
  \spec_{\RB^2}(\BA) = \spec_{\RL^2}(\BA)
\end{equation}
as a set. In particular, if $\BA$ is bounded, its norm in $\RL^2$ and $\RB^2$
coincide.
\end{prop}

The proof of this proposition is exactly the same as the one in
\cite{Shubin1978} for the case $m = 1$. Indeed, it relies on
some facts about function approximation proven in \cite[Lemmata 4.1 and 4.2]{Shubin1978}
which remain true as $m > 1$ since they
apply coordinatewise. 
Boundedness or ellipticity then implies Proposition
\ref{prop:normspec}. When we refer to the norm of an operator, we will not distinguish whether that
operator is acting in $\RL^2(\R^d;\C^m)$ or $\RB^2(\R^d;\C^m)$ since those norms
are the same.

As mentioned in Remark \ref{rem:Shubin}, there is a faithful,
norm-preserving
$*$-representation $\BA \mapsto \BA^\sharp$ of almost periodic operators  given by 
$\BA^\sharp := \ba(\bx + \by,D_\by)$
acting in $$\FH_m := \RB^2(\R^d) \otimes \RL^2(\R^d) \otimes \C^m.$$ 
Here $\bx$ is a variable of functions in $\RB^2(\R^d;\C^m)$ and $\by$
is a variable of functions in $\RL^2(\R^d;\C^m)$. The operator $\BA^\sharp$ is interpreted as a direct integral over $\bx$
of operators acting in $\RL^2(\R^d;\C^m)$. We denote by $e_J(x,y)$
the Schwartz kernel of the spectral projection $E_J(\BA)$. Note that in view of
Proposition~\ref{prop:normspec} and \cite{CoburnMoyerSinger}, if $\BA \le \BB$ as operators,
then $\BA^\sharp \le \BB^\sharp$ and
$\norm \BA = \norm{\BA^\sharp}$.

Finally, the operator $\BA^\sharp$ is affiliated to the $\text{II}_\infty$ factor $\FA$ generated by the two families of operators
\begin{equation}\nonumber
 \set{\be_{\bxi} \otimes \be_{\bxi}\otimes M: \bxi \in \R^d,\ M\in \CM_m} \quad \text{and} \quad \set{I \otimes T_{\bxi}\otimes M: \bxi \in \R^d,\ M\in \CM_m},
\end{equation}
where $\be_{\bxi}$ is the operator of multiplication by $e^{i \bxi \cdot \bx}$, $T_{{\bxi}}$ is the operator of translation $T_{\bxi} f(\bx) = f(\bx - \bxi)$
and $\CM_m$ is the algebra of $m\times m$ matrices with complex entries. This means that the results
of Section \ref{sec:ids}--\ref{sec:mapo} on the density of states measure (DSM)
also called the integrated density of states
(IDS) apply to this algebra of operators and this representation.

In the classical setting, the IDS is defined for differential operators using
the large box limit and for pseudo-differential operators as the trace of the
Schwartz kernel
\begin{equation}
  N(J;\BA) = M_\bx(\operatorname{tr} e_J(\bx,\bx)),
\end{equation}
where $M$ is the almost periodic mean. Note that this kernel is actually a
smooth integral kernel
whenever $J$ is a bounded interval, see \cite{ParSht2016}. 

Our terminology for the IDS is justified in \cite[Remark 3.1]{Shubin1979},
where it is shown that the IDS as defined in Section
\ref{sec:ids} is the same as the one obtained from the classical definition for
either differential or pseudo-differential operators.

\subsection{Concrete systems of operators}

From now on, we turn our attention to almost periodic pseudo-differential
operators whose principal symbol is diagonal and nondegenerate.
\begin{defi}
  A \emph{\pd} operator is an operator $\BA \in \BE\BS_m^\alpha$ for which there
  exists an unitary operator $\BU \in \BS_m^0$ so that $\BU^* \BA \BU = \BA_0 + \BB \in
  \BS\BE\BS_m^\alpha$ has the following properties.
  \begin{itemize}
    \item The \emph{principal part} $\BA_0 \in \BD \BE \BS_m^\alpha$, with symbol
\begin{equation}
  \ba_0(\bxi) = \diag\left(a_1 \abs \bxi^\alpha, \dotsc, a_m \abs
  \bxi^\alpha\right),
  \label{eq:defh0}
\end{equation}
with $a_j \ne 0$ and without loss of generality $a_1 \ge \dotso \ge a_m$. We set
$m^+ = \max{j : a_j > 0}$, where by convention $m_+ = 0$ if $a_1 < 0$. 
\item The \emph{subprincipal part} $\BB \in \BS_m^\beta$ for $\beta < \alpha$ and has
  frequency set $\Theta$. We also suppose that $\BB$ is formally self-adjoint,
  i.e. that its symbol satisfies
\begin{equation}
  \bbb_{\btheta}(\bxi) = \bbb_{-\btheta}(\bxi + \btheta)^*,
  \label{eq:selfadjb}
\end{equation}
for all $\bxi \in \RD$ and $\btheta \in \Theta$, where for a matrix $\ba$,
$\ba^*$ is its conjugate transpose.
\end{itemize}
If $a_j \ne a_k$ for $j \ne k$, we say that $\BA$ is a \emph{\mpu} operator.
\end{defi}
\begin{rem}
Since we are interested only in spectral properties of elliptic operators, for
the remainder of this paper we can always assume that the operators are already in $\BS\BE\BS_m^\alpha$.
\end{rem}

Without loss of generality, we assume that the frequency set $\Theta$ spans $\R^d$, contains
$\mathbf 0$, and is symmetric about $\mathbf 0$. Recall from
\eqref{eq:ZTheta} that, using sum rather than product notations for the
group of shifts in $\R^d$, that we also put
\begin{equation}
  \Theta^k = \Theta + \dotso + \Theta,
  \label{eq:thetasum}
\end{equation}
where the sum is taken $k$ times, and 
\begin{equation}
  \label{eq:ZTh}
  Z(\Theta) = \bigcup_{k\in \N}
\Theta^k
\end{equation}
The set $Z(\Theta)$ is countable and non-discrete, unless $\Theta$ generates a
lattice.

\subsection{Conditions on the perturbation and its frequency set}
\label{sec:prelim}

In this section, we state the exact conditions under which we can obtain
asymptotics for the integrated density of states for a system of operators acting
in $\RB^2(\R^d;\C^m)$. We also
show how we can reduce the problem to computing the IDS 
solely on some intervals contained in a large enough range of energies.

We are interested in the asymptotics for the positive energy and negative energy integrated
densities of states for a \pd{} operator $\BA = \BA_0 + \BB$, defined as
\begin{equation}
 N^+(\lambda):=  N^+(\lambda;\BA) := N([0,\lambda);\BA),
  \label{eq:npositive}
\end{equation}
and
\begin{equation}
  N^-(\lambda):= N^-(\lambda;\BA) := N( (-\lambda,0];\BA),
  \label{eq:nnegative}
\end{equation}
as $\lambda \to \infty$. For this, we will need some conditions on the
frequency set of the perturbation $\BB$.
  In Section \ref{sec:asyexp}, we reduce the
operator $\BA$ to a direct sum of operators of the type appearing in
\cite{MorParSht2014}. In that paper, the perturbations are required to satisfy
some conditions, which we describe for completeness. Conditions \ref{condI} and
\ref{condIV} correspond to Conditions \textbf A and \textbf C in
\cite{MorParSht2014} and we do not use them explicitly. Condition \ref{condII}
addresses \cite[equation 2.4]{MorParSht2014}, while Condition \ref{condIII} addresses
\cite[Condition \textbf B]{MorParSht2014}. We
refer the reader to \cite{MorParSht2014}, as well as
\cite{ParSht2012} for a more detailed discussion around these conditions and
their genericity.

We first need the following generic condition on the set $Z(\Theta)$ defined in
\eqref{eq:ZTh}.
\begin{cond}
  \label{condI}
  Suppose that $\btheta_1,\dotsc,\btheta_d \in Z(\Theta)$. Then,
  $Z(\set{\btheta_1,\dotsc,\btheta_d})$ is discrete.
\end{cond}
This condition is clearly satisfied for periodic $\BB$, but for quasi-periodic or
almost periodic $\BB$ it is meaningful. The next two conditions describe how well
$\BB$ is approximated by finite sums of homogeneous functions of $\bxi$, and by
quasi-periodic operators. 

\begin{cond}
  \label{condII}
  There exists a constant $C_0 > 1$ and 
  a discrete
  subset $J \subset (-\infty,\beta]$ such that for all $\btheta \in \R^d$ and $\abs \bxi \ge C_0$,
  \begin{equation}
    (1 - \bone_{C_0}(\bxi)) \bbb_{\btheta}(\bxi) = \sum_{\iota \in J}
    \abs{\bxi}^\iota
    \bbb_{\btheta}^{(\iota)}\left( \frac{\bxi}{\abs \bxi} \right),
  \end{equation}
  where $\bbb_{\btheta}^{(\iota)} \in \BS_m^0$ is positively homogeneous of
  degree $0$. We also suppose that for $\boeta \in \mathbb S^{d-1}$,
  $\bbb_{\btheta}^{(\iota)}(\boeta)$ has a series representation (written in
  multi-index notation)
  \begin{equation}
    \bbb_{\btheta}^{(\iota)}(\boeta) = \sum_{\bn \in \N_0^d}
    \bbb_{\btheta}^{(\iota,\bn)} \boeta^\bn,
  \end{equation}
  which converges absolutely in a ball of radius greater than one of $\R^d$. 
\end{cond}
If $\BB$ is quasi-periodic and $J_0$ is finite, these are the only conditions that
we need. Otherwise, we need to find a quasi-periodic approximation of $\BB$.
In view of \eqref{eq:smallnorm}, such an approximation will always exist, but we
need a quantitative version of it.
\begin{cond}
  \label{condIII}
  For every $k \in \N$, there exists $C_k > C_0$ such that for each $\rho > C_k$,
  there exists a finite symmetric
  $\tilde \Theta \subset (\Theta \cap \CB(\rho^{1/k}))$ and a finite subset $J_k
  \subset (-\infty,\beta]$ with 
  \begin{equation}
    \# J_k \le \rho^{1/k}
    \label{eq:cardtildej}
  \end{equation}
  such that the symbol
  \begin{equation}
    \label{eq:symbolrb2}
    \br_{\btheta}^{(k)}(\bxi) := \begin{cases}
      \bbb_{\btheta(\bxi)} & \text{if } \theta \not \in \tilde \Theta, \\
      \bbb_{\btheta(\bxi)} - \sum_{\iota \in J_k} \abs{\bxi}^{\iota}
      \bbb_{\btheta}^{(\iota)}\left( \frac{\bxi}{\abs \bxi} \right) & \text{if
      } \theta \in \tilde \Theta,
    \end{cases}
  \end{equation}
  satisfies, for all $\ell \in \N$,
  \begin{equation}
    \snorm{\br^{(k)}}{\beta}{\ell} \le c_{\ell,k} \rho^{-k},
    \label{eq:smallseminorm}
  \end{equation}
  for some $c_{\ell,k} > 0$.
\end{cond}
Finally, we need a Diophantine condition on the frequencies of $B$, for which we need some definitions. Fix $\tilde k \in \N$ (which will depend on the
order of the remainder in the asymptotic expansion, but not on $k$ as in
Condition \ref{condIII}). We say that $\FV$ is a quasi-lattice subspace of
dimension $q$ if there are linearly independent
$\btheta_1,\dotsc,\btheta_q \in \tilde \Theta^{\tilde k}$
such that $\FV = \spann(\btheta_1,\dotsc,\btheta_q)$. We denote by $\CV$ the collection
of all quasi-lattice subspaces. For $\FU,\FV \in \CV$, we write
$\phi(\FU,\FV)\in [0,\pi/2]$ to denote the angle between them, that is the angle
between $\FU \ominus (\FU \cap \FV)$ and $\FV \ominus (\FU \cap \FV)$, where for
a subspace $\FW \subset \FV$, $\FV \ominus \FW$ is the orthogonal complement of
$\FW$ in $\FV$. This angle is non-zero if and only
if $\FU$ and $\FV$ are strongly distinct, \emph{i.e.} if neither of them is a subspace of the other. 
Recalling that for any $k$ the choice of $\tilde \Theta$ depends on $\rho$, we put
\begin{equation}
  R(\rho) = \sup_{\btheta \in \tilde \Theta^{\tilde k}} \abs \btheta, \qquad
r(\rho) = \inf_{\btheta \in \tilde (\Theta^{\tilde k})'} \abs \btheta,
  \label{eq:defnrR}
\end{equation}
as well as 
\begin{equation}
s := s(\rho) = s(\tilde \Theta^{\tilde k})
:= \inf \sin(\phi(\FU,\FV)),
\end{equation}  
  where the infimum is over all strongly distinct pairs of subspaces in $\CV$. It is clear that
  \begin{equation}
    R(\rho) = \bigo{\rho^{1/k}},
    \label{eq:boundonR}
  \end{equation}
  where the implicit constant might depend on $k$ and $\tilde k$; however, we need the following condition for $r$ and $s$.
  \begin{cond}
    \label{condIV}
    For each fixed $k$ and $\tilde k$, the sets $\tilde \Theta$ can
    be chosen in such  way that for sufficiently large $\rho$, depending on $k$
    and $\tilde k$, the
    number of elements of $\tilde \Theta^{\tilde k}$ satisfies $\#\tilde \Theta^{\tilde k} \le \rho^{1/k}$ and we have that
    \begin{equation}
      s(\rho) \ge \rho^{-1/k}
      \label{eq:boundons}
    \end{equation}
    and
    \begin{equation}
      r(\rho) \ge \rho^{-1/k}.
      \label{eq:boundonr}
    \end{equation}
  \end{cond}
  \begin{rem}
    Condition \ref{condIV} is automatically satisfied for quasi-periodic and
    smooth periodic $\BB$. See \cite{ParSht2012} for further discussion of this condition.
  \end{rem}

\section{Asymptotic expansions for the IDS} \label{sec:asyexp}
  We now suppose that the perturbation $\BB$ satisfies Conditions
  \ref{condI}--\ref{condIV} and we set $\rho = \lambda^{1/\alpha}$, where $\alpha$
  is the order of $\BA_0$. We prove the two following theorems, depending on whether all the $a_j$
  in \eqref{eq:defh0} are distinct or not. Recall that $m_+ = \max \set{j : a_j
  > 0}$, with $m_+ = 0$ if $a_j < 0$ for all $j$. 
  \begin{thm} \label{thm:aexpconcretecut}
    Let $\BA$ be a \pd{} operator with subprincipal part $\BB \in \BS^{\beta}_m$,
    $\beta \le \alpha/2$ satisfying
    Conditions \ref{condI}--\ref{condIV}. Suppose that there exists $\gamma \le
    0$ such that whenever $a_j =
    a_k$ for some $1 \le j \ne k \le m$, then $[\BB]_{j,k} \in \BS^{\gamma}$ and
    put $\gamma^* = \max(2 \beta - \alpha, \gamma)$. Then, there exists a discrete set
    $L \subset (0,1 - \gamma^*)$ and constants $C_0^\pm$ and $C_{q,j}^\pm$, $0 \le q \le d-1$, $j \in
    L$ such that
    \begin{equation}
      \begin{aligned}
        N^\pm\left(\BA; \rho^\alpha \right) = C_0^\pm \rho^d  
    +  \sum_{j\in L}\sum_{q =0}^{d-1} 
        C_{j,q}^\pm \rho^{d-j} \log^q(\rho) + \bigo{\rho^{d - 1 + \gamma^*}},
      \end{aligned}
      \label{eq:asymptoticsconcretecut}
    \end{equation}
    as $\rho \to \infty$. If $m^+ = m$ (resp. if $m^+ = 0$), then $C_0^- = C^-_{j,q} =
    0$ (resp. $C_0^+ = C^+_{j,q}$ = 0) except for $(j,q) = (d,0)$. 
  \end{thm}
  \begin{thm} \label{thm:aexpconcrete}
    Let $\BA$ be a \mpu{} operator satisfying
    Conditions \ref{condI}--\ref{condIV}. 
    Then, for every $K \in \R$ there exists a discrete set
    $L \subset (0,d+K)$ and constants $C_0^\pm, C_{q,j}^\pm$, $0 \le q \le d-1$, $j \in
    L$, such that
    \begin{equation}
      \begin{aligned}
        N^\pm\left(\BA; \rho^\alpha \right) = C_0^\pm \rho^d + 
        \sum_{j\in L}\sum_{q =0}^{d-1} 
        C_{j,q}^\pm \rho^{d-j} \log^q(\rho) + \bigo{\rho^{-K}},
      \end{aligned}
      \label{eq:asymptoticsconcrete}
    \end{equation}
    as $\rho \to \infty$. If $m^+ = m$ (resp. if $m^+ = 0$), then $C_0^- = C^-_{j,q} =
    0$ (resp. $C_0^+ = C^+_{j,q}$ = 0) except for $(j,q) = (d,0)$. 
  \end{thm}
  \begin{rem}
  Note that the statement for $m^+ \in \set{0,m}$ follows from the operator being semi-bounded either above or below,
  respectively.

  Note as well that if $J \subset \Z$, i.e. if the symbol of $\BA$ is a
  classical symbol, see \cite[Chapter 7]{TaylorII}, then $L=
  \set{0,\dotsc,K + d - 1}$. 

    The set $L$ of allowable exponents can be made explicit, depending on $J$
    and $K$, see \cite[Remark 2.7]{MorParSht2014}.
\end{rem}

The proof of Theorems \ref{thm:aexpconcretecut} and \ref{thm:aexpconcrete} are
obtained after many reductions to simpler cases. Recall that they are the
general versions of
Theorems \ref{thm:dirac2dids} and \ref{thm:dirac3dids} in the introduction.

\subsection{IDS for uncoupled operators} \label{sec:uncoupled}

In this subsection, we prove that the conclusion of Theorem \ref{thm:aexpconcrete}
holds in the special case where $\BA \in \BU\BS_m^\alpha$, regardless of whether
an operator is \pd{} or \mpu{}. 
This means that in addition of satisfying the
conditions of Section \ref{sec:prelim}, its symbol is given by
\begin{equation}
  \ba(\bx,\bxi) = \ba_0(\bxi) + \bbb(\bx,\bxi),
  \label{eq:modeldmapo}
\end{equation}
where $\bbb(\bx,\bxi)$ is a diagonal matrix.
  \begin{prop}\label{thm:aexpconcreteu}
    Let $\BA \in \BU\BS^{\alpha}_m$ be an \pd{} operator satisfying conditions
    \ref{condI}--\ref{condIV}.  
    Then, for every $K \in \R$ there exists a discrete set
    $L \subset (0,d+K)$ and constants $C_0^\pm, C_{q,j}^\pm$, $0 \le q \le d-1$, $j \in
    L$, such that
    \begin{equation}
      \begin{aligned}
        N^\pm\left(\BA; \rho^\alpha \right) = C_0^\pm \rho^d +
        \sum_{j\in L}\sum_{q =0}^{d-1}
        C_{j,q}^\pm \rho^{d-j} \log^q(\rho) + \bigo{\rho^{-K}},
      \end{aligned}
      \label{eq:asymptoticsconcreteunc}
    \end{equation}
    as $\rho \to \infty$. If $m^+ = m$ (resp. if $m^+ = 0$), then $C^-_0 = C^-_{j,q} =
    0$ (resp. $C_0^+ = C^+_{j,q} = 0$) except for $(j,q) = (d,0)$. 
  \end{prop}
\begin{proof}
  Since $\BA \in \BU\BS_m^\alpha$, it can be split as a direct sum of operators 
  $A_1 \oplus \dotso \oplus A_m$ acting in the mutually orthogonal subspaces $\RB^2(\R^d) \otimes v_j$. As such, we have that on any interval $J$,
  \begin{equation}
    N(J;\BA) = \sum_{j=1}^m N(J;A_j).
    \label{eq:sumIDS}
  \end{equation}
  This means that, for $j\le m^+$, $A_j$ is semi-bounded below and acts
   invariantly on $\RB^2(\R^d) \otimes v_j$ 
  as the operator considered in \cite{MorParSht2014}.  For $j > m^+$, it is the
  operator $-\BA_j$ that acts in such a way. 
  From \cite[Theorem 2.5]{MorParSht2014}, this means that $N( (-
  \infty,\lambda);\BA_j)$ (resp. $N( (\lambda,\infty); \BA_j)$)
  enjoys an asymptotic expansion of the form \eqref{eq:asymptoticsconcrete} for
  $1 \le j \le m^+$ (resp. $m^+ < j \le m$). Observe that we have
\begin{equation}
  \begin{aligned}
  N^+\left(\rho^\alpha;\BA\right) &= \sum_{j=1}^{m^+} N(
  (-\infty,\rho^\alpha);A_j) - \sum_{j=1}^{m^+} N( (-\infty,0];A_j) \\ 
  &\qquad + \sum_{j=m^++1}^m N( (0,\rho^\alpha);A_j).
\end{aligned}
  \label{eq:sumIDSsplit}
\end{equation}
The terms in the first sum have the required asymptotic expansion. The terms in the second sum do not depend on $\rho$,
hence they might only change the constant term in \eqref{eq:asymptoticsconcrete}. Finally, the operators in the third sum are semi-bounded above, hence for
$\rho$ large enough the terms are constant and once again only affect the constant term. 
This proves the existence of the asymptotic expansion \eqref{eq:asymptoticsconcrete} for $N^+$. The proof for $N^-$ is the same,
interchanging the role of the semi-bounded below and above operators.
\end{proof}

  \subsection{Reduction to a finite interval}
 
  The strategy in this subsection is an adaptation of the one found in
  \cite{MorParSht2014,ParSht2012}. It consists in showing that an 
  asymptotic expansion holds in overlapping dyadic intervals $I_n$.
  
  For $K > -d$, we choose $\rho_0$ sufficiently large, to be fixed later. For
  every $n \in \N$, we put $\rho_n := 2 \rho_{n-1} = 2^n \rho_0$. We also define
  the intervals $I_n := \left[\rho_{n-1},\rho_{n+1}\right]$. We prove the
  following theorem, which implies Theorems \ref{thm:aexpconcretecut} and \ref{thm:aexpconcrete} as a
  corollary.
  \begin{thm} \label{thm:ainint}
    Let $\BA$ be an operator satisfying the conditions of either Theorem
    \ref{thm:aexpconcretecut} or \ref{thm:aexpconcrete}. Then, for either $K =
    -d + 1 - \gamma^*$ in the former case or any $K \in \R$ in the
    latter, there exists $\rho_0$
    large enough, a discrete
    set $L \subset (0,d+K)$ and constants $C_0^\pm, C^\pm_{j,q}$ for every $j \in L$ and $0
    \le q \le d-1$ such that for every $n \in \N$ and every $0 < \mu < \nu$ with $\mu, \nu \in I_n$, 
    \begin{equation}
      \begin{aligned}
      N((\mu^\alpha,\nu^\alpha);\BA) &= C_0^+(\nu^d - \mu^d) + \\ & \quad + 
      \sum_{j\in L}\sum_{q =0}^{d-1} 
      C_{j,q}^+
      \left(\nu^{d-j} \log^q(\nu) - \mu^{d-j} \log^q(\mu)\right)  + \bigo{\rho_n^{-K}},
      \label{eq:ainint}
    \end{aligned}
    \end{equation}
    where the implicit constants might depend on $K$, but not on $n$. Similarly,
    \begin{equation}
      \begin{aligned}
      N((-\nu^\alpha,-\mu^\alpha);\BA) &=  C_0^- (\nu^d - \mu^d) + \\ & \quad + 
      \sum_{j\in L}\sum_{q =0}^{d-1} 
      C_{j,q}^- 
      \left(\nu^{d-j} \log^q(\nu) - \mu^{d-j} \log^q(\mu)\right)  +
      \bigo{\rho_n^{-K}}.
      \label{eq:ainintneg}
    \end{aligned}
    \end{equation}
  \end{thm}
  \begin{rem}
    The reader familiar with previous works on the integrated density of states
    for almost periodic operators can notice that the roles of the dyadic
    decomposition in intervals $I_n$ is slightly different here. In previous
    work, this decomposition was necessary because the resonant zones were
    significantly different for different spectral intervals were different.
    This yielded coefficients $C^\pm$ depending possibly on $n$. It was however
    shown that the asymptotics had to match if the coefficients didn't grow too
    fast. 

    In our case, we need this decomposition in order to apply Theorem
    \ref{thm:ctou} when the perturbation is unbounded. Indeed, it relies on
    Lemma \ref{lem:spectralperturb} which can only be applied for some interval
    with control on how far away the endpoints can be. We will therefore obtain
    asymptotics when both endpoints belong to a specific dyadic interval, then
    glue the intervals together. We end up comparing the
    density of states with the one obtained in \cite{MorParSht2014} for
    operators acting on scalar functions, i.e. the case $m = 1$. In such a case, the dependence on $n$
    of the coefficients has already been removed.
  \end{rem}
  \begin{proof}[Proof of Theorems \ref{thm:aexpconcretecut} and \ref{thm:aexpconcrete} assuming Theorem
    \ref{thm:ainint}]
    We prove the theorem for $N^+$, the proof for $N^-$ is the same.
    For $K \in \R$, suppose that $\rho_0$ is large enough for Theorem
    \ref{thm:ainint} to hold. Suppose without loss of generality that for all
    $n$, $\rho_n$ is a point of continuity of $N^+$. For $\rho \in I_n$, we have that
    \begin{equation}
      \begin{aligned}
      N^+(\rho^\alpha) &= N^+(\rho_0^\alpha) + \sum_{j=1}^{n-1} N(
      (\rho_{j-1}^\alpha, \rho_j^\alpha);\BA) + N(
      (\rho_{n-1}^\alpha,\rho^\alpha);\BA) \\
      &= N^+(\rho_0^\alpha) +  
      \sum_{j\in L} \sum_{q =0}^{d-1} 
      C_{j,q}^+ 
      \left(\rho^{d-j} \log^q(\rho) - \rho_0^{d-j} \log^q(\rho_0)\right)  +
      \sum_{j=1}^n S_j,
    \end{aligned}
      \label{eq:foldingsum}
    \end{equation}
    where $S_j = \bigo{\rho_j^{-K}}$. This implies that
    \begin{equation}
      \sum_{j=1}^n S_j \ll \rho_0^{-K} \sum_{j = 1}^n 2^{-Kj} \ll \rho_n^{-K} \ll \rho^{-K}
      \label{eq:glued}
    \end{equation}
    since $\rho \in I_n$. One can see that the term depending on $\rho_0$ is
    $\bigo{1}$, so that it can be included either in the error term
    $\bigo{\rho^{-K}}$ when $K \le 0$
    or in the constant term in
    \eqref{eq:asymptoticsconcrete} and \eqref{eq:asymptoticsconcretecut}
    otherwise.
  \end{proof}

  \subsection{Reduction to a quasiperiodic operator}

  We now show in the following lemma that it is sufficient to prove
  Theorem \ref{thm:ainint} for quasiperiodic operators.
\begin{lem}
  \label{lem:qpenough}
  Let $\BA \in \BS\BE\BS_m^\alpha$ be an \pd{} operator with subprincipal part
  $\BB
  \in \BS^\beta_m$ satisfying
  Condition \ref{condIII} and $k \ge 2$. There exists $\rho_0 > 0$ and $0 < c_0 <
  1$ so that for every $n \in \N$ there exists a quasi-periodic \pd{}
  operator $\BA' \in \BS\BE \BS_m^\alpha$ with frequency set $\tilde \Theta \subset
  \CB(\rho_n^{1/k})$ such that
  \begin{itemize}
      \setlength\itemsep{1em}
    \item  $\BA - \BA' \in \BS_m^\beta$;
    \item $\supp(\ba'{}^{\CO\CD}) \subset \set{\abs{\bxi} > c_0 \rho_n}$;
      \item there is $\eps \ll \rho_n^{\alpha-k}$ such that for all $J \subset I_n^\alpha$,
  \begin{equation}
    \label{eq:qpapprox}
    N((\pm J)_{-\eps};\BA') \le N(\pm J;\BA) \le N( (\pm J)_\eps;\BA').
  \end{equation}
  \end{itemize}
\end{lem}

\begin{proof}
  For $k \in \N$ let $\tilde \Theta \subset \Theta \cap \CB(\rho_n^{1/k})$ be the
  frequency set given by Condition \ref{condIII} with $\rho = \rho_n$, and $\BR \in \BS^\beta_m$ be the operator with
  symbol given in \eqref{eq:symbolrb2}, which by \eqref{eq:smallseminorm}
  satisfies $\snorm{\BR}{\beta}{0}
  \ll \rho_n^{-k}$.
  Setting $\BA'' = \BA - \BR$ we have that $\BA - \BA'' \in
  \BS_m^{\beta}$, and that
  $\tilde \Theta$ is a frequency set for $\BA''$ and that as long as $\rho_0$ is
  large enough,
  \begin{equation}
    \label{eq:snormcomp}
    \snorm{\BA''}{\gamma}{0} \le 2 \snorm{\BA}{\gamma}{0}
  \end{equation}
  for all $\beta \le \gamma \le \alpha$.

  Writing any interval $J \subset \pm I_n^\alpha$ in the form $[M - r,M+r]$, it is easy to
  see that $\abs M + r \le (2 \rho_n)^\alpha$. 
  Put $\beta_0 = \max\set{\beta,0}$. 
  By Lemma \ref{lem:monotonicity}, estimate
  \eqref{eq:qpapprox} holds with $\BA''$ instead of $\BA'$ and
  \begin{equation}
    \eps_1 = \frac{\snorm{\BR}{\beta_0}{0}}{2 + \snorm{\BR}{\beta_0}{0}}\left(\abs M + r + C(1 +
    \snorm{\BR}{\beta_0}{0})^\frac{\alpha}{\alpha - \beta_0}\right) \ll  \rho_n^{\alpha -
  k}
  \end{equation}
  instead of $\eps$.
  Let us now define 
  \begin{equation}
  \BA' = \BA'' - \BA''{}^{\CO\CD}P_{c_0\rho_n} = (\BA'')^\CD +
  \BA''{}^{\CO\CD}(1 - P_{c_0 \rho_n}),
\end{equation}
where $0 < c_0 < 1$ is to be determined later. By \eqref{eq:snormcomp} and
\eqref{eq:partition}
\begin{equation}
  \snorm{\BA'}{\gamma}{0} \le 4 \snorm{\BA}{\gamma}{0}
\end{equation}
for all $\beta \le \gamma \le \alpha$. We apply Lemma
  \ref{lem:spectralperturb} with
  \begin{equation}
  H_0 = (\BA'')^{\CD}, \qquad B = \BA''{}^{\CO\CD}(1 - P_{c_0\rho_n}), \qquad A =
  \BA''{}^{\CO\CD}P_{c_0\rho_n}, \qquad H = \BA'.
  \end{equation}
  By Proposition \ref{prop:normorder},
  \begin{equation}
    \norm{\BA''{}^{\CO\CD}P_{c_0\rho_n}} \le (c_0 \rho_n)^{\beta_0}
    \snorm{\BA''{}^{\CO\CD}}{\beta_0}{0}.
  \end{equation}
  Set $X = \lfloor (2 - \alpha + k + \beta_0) \log_3 \rho_n \rfloor$, and 
  let
  \begin{equation}
    \label{eq:increasingproj}
    Z_l := c_0 \rho_n +  l  \rho_n^{2/3}, \qquad 0 \le l \le X
    - 1,
  \end{equation}
  so that if $\rho_0$ is large enough, $Z_{X - 1} \le 2 c_0 \rho_n$.
  For $0 \le l \le X$ introduce the family of projections
  \begin{equation}
    P_l := \begin{cases}
      P_{Z_0} & \text{for } l = 0, \\
      P_{Z_l} - P_{Z_{l -1}} & \text{for } 0 < l < X, \\
      1 - P_{Z_{X - 1}} & \text{for } l = X.
    \end{cases}
  \end{equation}
  We now verify that the conditions of Lemma \ref{lem:spectralperturb} are
  satisfied. It is
  clear that $\BB P_{c_0\rho_n} P_{Z_0} = B P_{c_0\rho_n}$, and relations \eqref{eq:projections}
  follow from \eqref{eq:boundonR} and \eqref{eq:increasingproj} as long as $k
  \ge
  2$ and $\rho_0$ is large enough. By Proposition \ref{prop:normorder},
  for $0 \le l < X$,
  \begin{equation}
    \label{eq:dproj}
    \norm{P_l \BA' P_l} \le 
    Z_{X-1}^\alpha \snorm{\BA'}{\alpha}{0} 
    \le 4(2 c_0
    \rho_n)^\alpha \snorm{\BA}{\alpha}{0}.
  \end{equation}
  We also have that
  \begin{equation}
    \label{eq:odproj}
    \norm{P_l \BA''{}^{\CO\CD} P_{l-1}} + 
    \norm{P_l \BA''{}^{\CO\CD} P_{l+1}} \le 2 Z_{X-1}^{\beta}
    \snorm{\BA''}{\beta_0}{0}  \le 4 (2c_0\rho_n)^{\beta_0}
    \snorm{\BA}{\beta_0}{0}
  \end{equation}
For $0 \le l < X$, set
  \begin{equation}
    D_l = \dist(J,\spec(P_l \BA' P_l)^\sharp).
  \end{equation}
  By  \eqref{eq:dproj} and Lemma \ref{lem:sharpS0} for $l \le X - 1$
  \begin{equation}
    \spec((P_l \BA' P_l)^\sharp) \subset \left[- 4(2 c_0
      \rho_n)^\alpha\snorm{\BA}{\alpha}{0},
    4(2 c_0 \rho_n)^\alpha\snorm{A}{\alpha}{0}
  \right],
  \end{equation}
  so that setting 
  $  c_0^{-\alpha} = 2^{\alpha+ 3} \snorm{\BA}{\alpha}{0}$
  gives
  \begin{equation}
    D_l \ge  \frac{\rho_n^\alpha}{2},
  \end{equation}
    in particular
    \eqref{eq:distspectrum} holds. Combining with
  \eqref{eq:odproj} we have that
  \begin{equation}
    \max_{0 \le l < X} \set{\frac{
    \norm{P_l \BA''{}^{\CO\CD} P_{l-1}}+ 
\norm{P_l \BA''{}^{\CO\CD} P_{l+1}}}{D_l}} \le 2 (2c_0)^\beta \rho_n^{\beta -
\alpha} \snorm{\BA}{\beta_0}{0},
  \end{equation}
  so that for $\rho_0$ large enough, \eqref{eq:projectionnorm} is satisfied.
  
  Since the conditions of Lemma
  \ref{lem:spectralperturb} are satisfied,  for
  \begin{equation}
    \eps_2 = 3^{2 - X} \norm{\BA''{}^{\CO\CD}P_{c_0\rho_n}} \le  2\rho_n^{\alpha -
    k} \snorm{\BA}{\beta_0}{0},
  \end{equation}
  we have that
  \begin{equation}
    N\left( I_{- \eps_1 - \eps_2}; \BA' \right) \le N(I_{-\eps_1};\BA'') \le
    N(I;\BA);
  \end{equation}
  and 
  \begin{equation}
    N(I_{+ \eps_1 + \eps_2};\BA') \ge N(I_{+\eps_1};\BA'') \ge N(I;\BA).
  \end{equation}
  Our claim therefore holds with $\eps = \eps_1 + \eps_2$.
\end{proof}

\subsection{Reduction to uncoupled operators} \label{sec:redtouncoupled}
In this subsection, we show that it is sufficient to prove Theorem
\ref{thm:ainint} for quasi-periodic operators, it is sufficient to do so for
uncoupled operators. At the end of the section, we finally prove Theorem
\ref{thm:ainint} after all those reductions, which completes the proof of
Theorems \ref{thm:aexpconcretecut} and \ref{thm:aexpconcrete}.

\begin{thm}
  \label{thm:ctoucut}
  Let $\BA$ be an operator satisfying the conditions of Theorem
  \ref{thm:aexpconcretecut}. Then, for every $n \in \N$ there is an operator $A' \in \BU\BS^\alpha_m$ and $\eps \ll
  \rho_n^{\gamma^*}$ such that for all $\mu,\nu \in I_n$ and $I = (\mu^\alpha,
  \nu^\alpha)$,
  \begin{equation}
    \label{eq:errorestinterval}
    N(\pm I_{- \eps}; \BA') \le N(\pm I;\BA) \le N(\pm I_\eps;\BA').
  \end{equation}
  In particular,
  \begin{equation}
    \label{eq:errorest}
    N(\pm I;\BA) = N(\pm I;\BA') + \bigo{\rho_n^{d - 1 + \gamma^*}}.
  \end{equation}
\end{thm}

\begin{proof}
  We only prove this theorem for $I$, the proof for $-I$ follows from the same
  argument, making the relevant sign changes. By Proposition
  \ref{thm:aexpconcreteu}, we see that
  \eqref{eq:errorest} follows from \eqref{eq:errorestinterval}. We therefore
  only prove the latter. 

  We keep a quantitative track of the estimates found in Section
  \ref{sec:reducscalar}. By Lemma \ref{lem:qpenough}, we can without loss of
  generality assume for some fixed $C_0 > 0$ that $\BA$ is a quasi-periodic operator whose frequency set
  $\Theta$
  lies in the ball $B(\rho_n^{1/k})$ for some $k \in \N$ and such that the support of
$\ba^{\CO\CD}$ lies in $\set{\abs \xi > C_0 \rho_n}$. In particular, we can
  assume that there is $s > 0$ such that for all $\xi \in \supp(\ba^{\CO\CD})$,
  all $\theta \in \Theta$ and all $j,k$ such that $a_j \ne a_k$,
  \begin{equation}
    \Big|a_j \abs{\btheta + \bxi}^\alpha - a_k \abs{\bxi}^\alpha\Big| > s \abs
    \bxi^\alpha.
  \end{equation}
  By Theorem \ref{thm:systemonestep}, since $\BA$ is quasi-periodic there are
  symmetric operators $\BY \in
\BU\BS^{\beta}_m$, $\BR \in \BS_m^{\gamma^*}$ and $\BPsi \in \BS_m^{\beta
  - \alpha}$ such that $\BA$ is unitarily equivalent through conjugation with
  $\exp(i\BPsi)$ to
  \begin{equation}
    \BA' + \BR = \BA^{\CD} + \BY + \BR.
  \end{equation}
  Here, the symbol of $\BPsi$ is given by
  \begin{equation}
    [\boldsymbol \psi_{\btheta}(\bxi)]_{j,k} =
    \frac{i[{\bbb^\CC}_{\btheta}(\bxi)]_{j,k} \chi_{j,k}}{a_j \abs{\bxi +
      \btheta}^\alpha -
    a_k \abs \bxi^\alpha},
  \end{equation}
  where $\chi_{j,k} = 1$ if $a_j \ne a_k$ and $0$ otherwise.
  Using the fact that $$\ad(\BA;\BPsi) = \ad(\BA^{\CO\CD};\BPsi) +
  \ad(\BA^\CD;\BPsi) = \ad(\BA^{\CO\CD};\BPsi) - \BA^{\CNR},$$
  the operator $\BR$ is
  obtained from equations \eqref{eq:decompo}, \eqref{eq:decompodetails} with
  $\widetilde \BR = 0$, and \eqref{eq:defRconsecgauge} by
  \begin{equation}
    \label{eq:defRonestepconcrete}
    \BR = \BB^{\CR,\CC} + \ad(\BA^{\CO\CD};\BPsi) + \sum_{k=2}^\infty
    \frac{1}{k!}\ad^k(\BA^{\CO\CD};\BPsi) - \sum_{k=2}^\infty \frac{1}{k!}
    \ad^{k-1}(\BA^{\CNR};\BPsi).
  \end{equation}
  By Lemma \ref{lem:monotonicity},
  \begin{equation}
    N(I_{-\eps};\BA') \le N(I;\BA' + \BR) \le N(I_\eps;\BA')
  \end{equation}
  for $\eps = \norm{\BR}$. Since $\BPsi$ has order $\beta - \alpha$ and is
  supported on $\set{\abs \xi > c \rho}$, by Corollary \ref{cor:domain} and
  Lemma \ref{lem:product} we have as in  Proposition
  \ref{prop:normorder} that
  \begin{equation}
    \norm{\ad(\BA^{\CO\CD};\BPsi)} \ll \rho_n^{2 \beta - \alpha} \snorm{\BA^{\CO\CD}}{\beta}{0}
      \snorm{\BPsi}{-\beta}{\abs{\beta}},
  \end{equation}
  so this gives the contribution from the second term in
  \eqref{eq:defRonestepconcrete}. The third and fourth terms uses the same estimate and the
  fact that this sum is absolutely convergent. Finally, for the first term we
  supposed that $\BB^{\CR,\CC} \in \BS_m^{\gamma}$, and it is also
  supported on $\set{\abs \xi > c \rho}$ so that by Proposition \ref{prop:normorder},
  \begin{equation}
    \norm{\BB^{\CR,\CC}} \ll \rho_n^{\gamma}.
  \end{equation}
  Together, this completes the proof.
\end{proof}
When $a_j \ne a_k$ whenever $j \ne k$, we get the following stronger statement.
\begin{thm} \label{thm:ctou}
  Let $\BA$ be an operator satisfying the hypotheses of Theorem
  \ref{thm:aexpconcrete}. There is a decreasing sequence $\set{\gamma_K}_{K \in
  \N}$, $\gamma_K \to - \infty$ such that for all $K, n \in \N$
  there is an operator $\BA_K \in \BU\BS_m^\alpha$ and some $\eps \ll \rho_n^{-\alpha - d - K}$ such that
  for all $\mu,\nu \in I_n$, and $I = (\mu^\alpha,\nu^\alpha)$, 
   \begin{equation}
  N(I_{-\eps};\BA_K) \le N(I;\BA) \le N(I_\eps;\BA_K)
    \label{eq:desiredeps}
  \end{equation}
    and such that if $K_1 < K_2$, then
  \begin{equation}
    \BA_{K_1} \equiv \BA_{K_2} \mod \BS^{\gamma_{K_1}}_m.
  \end{equation}
  In particular,
  \begin{equation}
  N(\pm I;\BA) = N(\pm I; \BA_K) + \bigo{\rho_n^{-K}}.
    \label{eq:approxbydiagonal}
  \end{equation}

\end{thm}

\begin{proof}
  This statement is proven in the same way as the previous one, replacing the
  use of Theorem \ref{thm:systemonestep} with the parallel gauge transform
  Theorem \ref{thm:gtsystem}, with a number of steps depending on $K$. This
  is possible because the condition on the terms coupling $a_j=a_k$ for $j
  \ne k$ is vacuously verified, so that it is assuredly preserved after each step of gauge
  transform. This yields a remainder $\BR \in \BS_m^{-N}$ for any $N$, allowing
  for the arbitrary precision in the approximation for the density of states.
\end{proof}
\begin{rem}
  Note that after this reduction, Conditions \ref{condII} and \ref{condIII},
  corresponding to \cite[Equation 2.4 and Condition
  B]{MorParSht2014}, do not hold anymore. However, the reason why these
  conditions are needed is to have a specific form for the functions
  $\bbb_{\btheta_1}(\bxi + \btheta_2)$ $\btheta_1,\btheta_2 \in \tilde
  \Theta$, see \cite[Equation 10.5]{MorParSht2014}. This expansion still holds
  if these conditions are imposed on the symbol prior to reduction to uncoupled
  operators.
\end{rem}

\begin{proof}[Proof of Theorem \ref{thm:ainint}]
  By Theorems \ref{thm:ctoucut} and \ref{thm:ctou}, there is $\rho_0$ large enough so that
for any $n$ there is an operator $\BA_K \in \BU\BS_m^\alpha$ such that for $\mu,\nu \in I_n$,
\begin{equation}
 N((\mu^\alpha,\nu^\alpha);\BA) = N((\mu^\alpha,\nu^\alpha);\BA_K) +
 \bigo{\rho_n^{-K}},
\end{equation}
where $K = 1- d - \gamma^*$ if $\BA$ satisfies the hypotheses of Theorem
\ref{thm:aexpconcretecut} and $K \in \R$ if $\BA$ satisfies the hypotheses of
Theorem \ref{thm:aexpconcrete}.
Equation \eqref{eq:ainint} (with coefficients $C^\pm_{j,q}$ depending on $n$) then follows by Proposition \ref{thm:aexpconcreteu}
for $\BA_K$ and the fact that
\begin{equation}
  \label{eq:munu}
N((\mu^\alpha,\nu^\alpha);\BA_K) = N^+(\nu^\alpha;\BA_K) - N^+(\mu^\alpha;\BA_K).
\end{equation}

In order to remove the dependence on $n$ of the coefficients, it is sufficient
for every $K \in \N$  to prove that they must agree for all $n$ large enough.
Since the coefficients obtained in Proposition \ref{thm:aexpconcrete} do not
depend on $n$, this means that as soon as $\mu, \nu \in I_{n} \cap I_{n+1}$,
\eqref{eq:munu} gives the same coefficients for the asymptotic expansion up to
terms of order $\rho_n^{-K}$, which means that the coefficients need to agree
for all $n$ large enough. 
\end{proof}

\section{The structure of periodic operators}
\label{sec:periodic}

We now turn our attention to periodic operators. In this section, we describe the structure of operators that are periodic with respect to
some lattice $\Lambda$, and we give a quantitative
approach to the study of the Bethe--Sommerfeld property. We realise the usual
Bloch--Floquet decomposition through Besicovitch spaces. 

\subsection{Description of periodic operators}
For  periodic operators we assume that $G$
is not $\R^d$ but rather the dual lattice $\Theta := \Lambda^\dagger \subset \R^d$. We note that in this
case $Z(\Theta) = \Theta$. 

Invariance of $\BA$ under the action of $\Lambda$ means that for all $\bk \in \R^d$, the subspace
\begin{equation}
  \label{eq:invsubspace}
  \ell^2_\bk(\Theta;\C^m) :=  \overline{\spann \set{\be_{\bxi,j} : 1 \le j \le m, \bxi \in
  \Theta + \bk} } \subset \RB^2(\R^d;\C^m)
\end{equation}
is an invariant subspace for $\BA$, and we denote by $\BA(\bk)$ the restriction
of $\BA$
to this subspace. It is clear from the definition that we can restrict ourselves
to $\bk \in \CO = \R^d/\Theta$, and we call $\bk$ a
\emph{quasimomentum}. For any $\bxi \in \R^d$, its fractional part $\set \bxi \in
\CO$ is the
image of $\bxi$ under the quotient map.  The spectrum
of $\BA$ can be
obtained as
\begin{equation}
  \spec(\BA) = \bigcup_{\bk \in \CO} \spec(\BA(\bk)),
\end{equation}
see \cite[Theorem 4.5.1]{
Kuchment}. For every $\bk \in \CO$ the spectrum of $\BA(\bk)$ in $\ell^2_\bk(\Theta;\C^m)$ is discrete.

The usual approach to studying the $\RL^2(\R^d;\C^m)$ theory of periodic operators is through the Floquet-Bloch
decomposition, see \emph{e.g.} \cite{Kuchment}, where we represent $\BA$ as a
direct integral over $\CO$ of the fibre operators $\BA(\bk)$. This would
require us to introduce a considerable amount of machinery. However, the 
Bethe-Sommerfeld property is strictly about the spectrum as a set,
and to every elliptic periodic operator acting in $\RL^2(\R^d;\C^m)$ there
corresponds an elliptic operator acting in $\RB^2(\R^d;\C^m)$ with the same
spectrum. We therefore consider periodic operators as operators on Besicovitch
space where we decompose them according to the invariant subspaces
\eqref{eq:invsubspace}. This makes our statements and proofs more 
direct.
\begin{rem}\label{rem:Iinfinitybs}
The subspaces $\ell^2_\bk(\Theta;\C^m)$ can be realised as
$$\RH^0(\Theta + \bk;\C^m) = \overline{\spann\left(\set{\be_{\btheta + \bk,j},
\btheta \in \Theta, 1 \le j \le m}\right)}.$$ The group $G$ is $\Theta$ acting on
$\Theta + \bk$ by translation. In this case,
consider the $\text{I}_\infty$ factor $\FA$ generated by $
\set{\be_{\btheta} \otimes M: \btheta \in \Theta, M \in \CM_m}.
$

It is clear that the restriction of the subalgebra of periodic operators to $\ell^2_\bk(\Theta;\C^m)$ is
affiliated to $\FA$, and that it respects the conditions described at the beginning
of Section \ref{sec:ids}. The associated trace of the spectral projection over an
interval $J$ is simply
$N(J;\BA(\bk)) := \#\set{j : \lambda_j(\BA(\bk)) \in J}$, the number 
of eigenvalues of $\BA(\bk)$ in that interval.

We also make the observation that for a bounded, periodic self-adjoint
operator $\BPsi$,
the restriction to $\ell^2_\bk(\Theta;\C^m)$
of the unitary operator $\exp(i \BPsi)$ is still unitary, since
$\ell^2_\bk(\Theta;\C^m)$ is an invariant subspace.
This means that we can simultaneously use the gauge transform on each of the
 fibre operators and that the estimates from Section \ref{sec:ids} hold uniformly for the counting
function of the fibre operators.
\end{rem}
Let us now describe the structure of the spectrum of $\BA$ in terms of the
spectra of the fibre operators $\BA(\bk)$. Since $\BA$ is self-adjoint, it has
Fredholm index $0$, this implies that the Bloch variety
\begin{equation}
  \set{(\bk,\lambda) \subset \CO \times \R : \lambda \in \spec(\BA(\bk))}
  \end{equation}
  is a principal analytic set \cite[Corollary 3.1.6 and Section
  3.4.C]{Kuchment}. As such, if $\BA$ is semi-bounded below, we can
naturally label the eigenvalues of $\BA(\bk)$ in non-decreasing order, counting
multiplicity. Then, the functions
$\lambda_j(\bk):= \lambda_j(\BA(\bk))$ are piecewise
analytic functions of $\bk$. If $\BA$ is not semi-bounded, we can label the
eigenvalues in non-decreasing order by $j \in \Z$ and it is possible to choose the
labelling so that the functions $\lambda_j(\bk)$ are piecewise
analytic. This requirement determines the labelling uniquely up to a uniform shift of the indices.
Note that continuity in $\bk$ of the functions $\lambda_j$ and discreteness of
the spectrum imply that labelling the eigenvalues at one quasimomentum $\bk$
induces a labelling everywhere in $\CO$. The interval
\begin{equation}
  \iota_j:= \iota_j(\BA) := \bigcup_{\bk \in \CO^\dagger} \lambda_j(\BA(\bk))
\end{equation}
is called the $j$th spectral band of $\BA$. 

\subsection{The overlap function}
In order to prove that an operator has the Bethe-Sommerfeld
property we study the band overlap, characterized by the \emph{overlap
function} $\zeta(\lambda;\BA)$, $\lambda \in \R$, introduced by M. Skriganov
\cite{Skriganov}. The overlap function is defined as the maximal
number $t$ such that the symmetric interval $[\lambda-t,\lambda+t]$ is entirely
contained in one band, \emph{i.e.}
\begin{equation}
  \label{eq:defnoverlap}
  \zeta(\lambda;\BA) := \begin{cases}
    \max_j \max\set{t \ge 0 : [\lambda-t,\lambda+t] \subset \iota_j} & \text{if } \lambda \in
 \spec(\BA)\\
 0 & \text{if }\lambda \not \in \spec(\BA).
  \end{cases}
\end{equation}
It is not hard to see that $\zeta$ is a continuous function of $\lambda$. 
In
order to use our machinery we will relate the overlap function to the
eigenvalue counting functions of the operators $\BA(\bk)$. This type of idea has been used
in the past but crucially relied on the fact that $\BA$
was semi-bounded below. In the following proposition we find an equivalent
formulation that is
robust under perturbations yet works for operators that are
not semi-bounded. Recall that for an interval $I = [s,t] \subset \R$ and $\eps
\in \R$, we define
\begin{equation}
  I_\eps := \begin{cases}
    \varnothing & \text{for } \eps < \frac{s-t}{2}, \\
    [s - \eps,t + \eps] & \text{otherwise.}
  \end{cases}
\end{equation}

\begin{lem} \label{lem:overlapidscont}
  Suppose that $\BA_1$, $\BA_2$ are self-adjoint periodic operators. Suppose that
  for all $p \in \set{1,2}$, $\bk \in \CO$, $\BA_p(\bk)$
  has discrete spectrum. For $\lambda \in \R$ and $t > 0$, let
  \begin{equation}
    \label{eq:smalldist}
    \delta = \min_{\bk\in \CO} \max\set{\dist(\mu;\spec(\BA_1(\bk))) :
      \mu \in [\lambda
    - t, \lambda + t]}.
  \end{equation}
  
  Suppose that there is $0 \le \eps \le \delta/4$ such that for all $\bk \in
  \CO$
  and any interval $I \subset [\lambda - t, \lambda + t]$
  \begin{equation}
  \label{eq:condids}
  N(I; \BA_2(\bk)) \le N(I_\eps;\BA_1(\bk))
  \quad \text{and} \quad N(I;\BA_1(\bk)) \le N(I_\eps;\BA_2(\bk)).
  \end{equation}
  
  Then, for $p \in \set{1,2}$ there exist sets of consecutive integers $J_p \subset \Z$ and
  surjective maps
  $$\lambda(\BA_p(\bk)) : J_p \to \spec(\BA_p(\bk))$$
      such that for all $j \in J_p$, $\lambda_j(\BA_p(\bk))$ are
    continuous in $\bk$ and such that
\begin{equation}
  \label{eq:evclose}
  \abs{\lambda_j(\BA_1(\bk)) - \lambda_j(\BA_2(\bk))} \le \eps
\end{equation}
for all $\bk$ and $j$ such that $\lambda_j(\BA_p(\bk)) \in [\lambda-t,\lambda+t]$.
\end{lem}

\begin{rem}
We do not ask in the previous lemma that both operators share the properties of
being either bounded, semi-bounded above or below, or unbounded in both directions.
\end{rem}

\begin{proof}
  For any $\bk \in \CO$, $0 < \eta \le \delta$, we say that $\mu \in [\lambda - t,
  \lambda + t]$ is $\eta$-distant (from the spectrum of $\BA_1(\bk)$) at $\bk$ if 
\begin{equation}
  \label{eq:choicemu}
  \dist(\mu,\spec(\BA_1(\bk))) \ge \eta.
\end{equation}
By \eqref{eq:smalldist}, for every $\bk \in \CO^\dagger$ there exists $\mu \in
[\lambda - t, \lambda + t]$ which is $\delta$-distant at $\bk$. 
 By the second inequality in \eqref{eq:condids}, if $\mu$ is
 $\eta$-distant at $\bk$ for some $\eta > \delta/4$, then for $p \in \set{1,2}$
  \begin{equation}
    \label{eq:muintempty}
    (\mu - \eps,\mu + \eps) \cap \spec(\BA_p(\bk)) = \varnothing.
\end{equation}
 Choose $\bk_0 \in \CO$ and $\mu_0$ a point $\delta$-distant at $\bk_0$.
 Maps $j \mapsto \lambda_j(\BA_p(\bk))$ can be uniquely defined from the
 properties that they are nondecreasing, mapping to continuous functions in
 $\bk$, and that $\lambda_0(\BA_p(\bk_0))$ is the
smallest eigenvalue larger than $\mu_0$. Note that the sets $J_1$ and $J_2$
are both defined uniquely from these properties, in particular if $\BA_p$ is
unbounded both above and below then $J_p = \Z$. 

We now prove that, for all $\bk \in \CO$, if $\mu$ is $\delta/2$-distant at $\bk$, then
for all $j \in J_1 \cap J_2$, then
\begin{equation}
  \label{eq:sameside}
  \big(\lambda_j(\BA_1(\bk)) - \mu\big) \big(\lambda_j(\BA_2(\bk)) - \mu\big) > 0,
  \end{equation}
  in other words, for $p \in \set{1,2}$, $\lambda_j(\BA_p(\bk))$ are both on the
  same side of $\mu$. The functions $\lambda_j(\BA_p(\bk))$ were constructed
  specifically so that \eqref{eq:sameside} holds at $\bk_0$ and $\mu_0$, our
  goal is to show that this property propagates to other $\mu$ and $\bk$. 
  
  We first prove that if \eqref{eq:sameside} holds for some
  $\mu$ $\delta/2$-distant at $\bk$, then it holds for all other
  $\nu$ $\delta/2$-distant at $\bk$. This is a direct consequence of
  \eqref{eq:condids} and \eqref{eq:muintempty}, which imply that
\begin{equation}
  N\big([\mu,\nu];\BA_1(\bk)\big) = N\big([\mu,\nu];\BA_2(\bk)\big).
\end{equation}

By continuity, for every $\bk \in \CO$ there is $s_\bk > 0$ so that
whenever $\mu$ is $\delta$-distant at $\bk$, $\mu$ is also $\delta/2$-distant at
every $\bk' \in \CB(\bk,s_\bk)$. 
This also implies that if
\eqref{eq:sameside} holds at $\bk$ for one of those $\mu$, it also holds for
that $\mu$ at every $\bk' \in \CB(\bk,s_\bk)$, and therefore at every $\nu$
$\delta/2$-distant at $\bk'$.

By compactness of $\CO$, there are $\bk_1,\dotsc,\bk_\ell$ such that
$\CO$ is covered by the balls
$U_j = \CB(\bk_j,s_{\bk_j}/2)$, with $0 \le j \le \ell$. If $U_{j} \cap U_{j'} \ne
\emptyset$, we have that $\bk_{j'} \in \CB(\bk_j,s_{\bk_j})$, so that if
\eqref{eq:sameside} holds for some $\mu$ $\delta$-distant at $\bk_j$ then it
also holds for all $\nu$ $\delta/2$-distant at $\bk_{j'}$, and therefore also at
any $\bk' \in U_{j'}$. By
connectedness of $\CO$, this means that \eqref{eq:sameside}
only needs to be verified for some $\bk_j$, $0 \le j \le \ell$ and one $\mu$
$\delta/2$-distant at $\bk_j$. Choosing $\bk_0$ and $\mu_0$, this means that
\eqref{eq:sameside} holds everywhere.

Suppose now that for some $p \in \set{1,2}$
there is some $j \in J_p$ and $\bk \in \CO$ such that
$\lambda_j(\BA_p(\bk)) \in [\lambda-t,\lambda+t]$ and
\begin{equation}
\lambda_j(\BA_1(\bk)) - \lambda_j(\BA_2(\bk)) > \eps.
\end{equation}
Let $\mu$ be $\delta$-distant at $\bk$. Without loss of generality
assume that $[\mu,\infty) \cap \spec(\BA_p) \ne \varnothing$ and let
\begin{equation}
  j' = \min \set{\ell : \lambda_\ell(\BA_p)(\bk) > \mu}.
  \end{equation}
 Supposing that $j \ge j'$, and using
\eqref{eq:muintempty} we obtain
\begin{equation}
  N\left(\left[\mu,\lambda_j(\BA_2(\bk))\right];\BA_2(\bk)\right) \ge j + 1 - j' >
  N\left(\left[\mu,\lambda_j(\BA_2(\bk))\right]_{\eps};\BA_1(\bk)\right),
\end{equation}
which contradicts the second inequality in \eqref{eq:condids}. Similarly,
supposing that $\lambda_j(\BA_2(\bk)) - \lambda_j(\BA_1(\bk)) > \eps$
contradicts the first inequality in \eqref{eq:condids}. The case $j < j'$ is
treated analogously. We can therefore
deduce that \eqref{eq:evclose} holds at every $\bk \in \CO^\dagger$.

\end{proof}

The previous lemma admits the following corollary in the situation
where the difference $\BA_1 - \BA_2$ is a bounded operator. The reader interested solely
in this case might notice that the proofs of the statement could have been more
direct on its own.

\begin{cor}
  Let $\BA = \BA_0 + \BB$ be a self-adjoint unbounded periodic operator such
  that $\BA(\bk)$ 
  has discrete spectrum for all $\bk \in \CO$ and such that $\BB$ is
  bounded. Then, there exist
  labelings $\lambda_j(\BA_0(\bk))$ and $\lambda_j(\BA(\bk))$ of the eigenvalues of the
fibre operators such that the functions $\lambda_j(\BA_0(\cdot))$ and
    $\lambda_j(\BA(\cdot))$ are both continuous on $\CO$ and such that for
    every $\bk \in \CO$
    \begin{equation}
      \abs{\lambda_j(\BA_0(\bk) - \lambda_j(\BA(\bk))} \le \norm{\BB}.
    \end{equation}
\end{cor}

\begin{proof}
  It suffices to observe that $\norm{\BB(\bk)} \le \norm{\BB}$ for all $\bk \in
  \CO$.
  Defining the continuous family of operators $\BA_t = \BA_0 + t \BB$, it is easy to
  see that $\BA_1 = \BA$ and $\norm{\BA_t - \BA_s} = \abs{t-s}\norm \BB$.  From Lemma
  \ref{lem:monotonicity}, we know that for all $I \subset \R$,
  \eqref{eq:condids} holds for $\BA_s,\BA_t$ with $\eps = \abs{t-s}\norm \BB$. It is
  also clear that $\delta$ defined in \eqref{eq:smalldist} is continuous in the parameter $t$. 
  Setting
  \begin{equation}
  N = \ceil{\frac{\norm \BB}{\min_{0 \le t \le 1} \eps_t}},
  \end{equation}
  and applying recursively Lemma \ref{lem:overlapidscont} to the operators
  $\BA_{j/N}$ and $\BA_{(j+1)/N}$, with $0 \le j < N$ yields the result we seek.
\end{proof}

The previous lemma and corollary provide us with an explicit way to compare the overlap
function. This is made precise in the following proposition.

\begin{prop}
  \label{prop:overlapcontev}
  Suppose that $\BA_1,\BA_2$ are self-adjoint, periodic operators such that for all
  $p \in \set{1,2}$, $\bk \in \CO$ $\BA_p(\bk)$ has discrete spectrum.
  Suppose that for $\eps  > 0$ there is a non-decreasing labelling of their eigenvalues so that whenever
  $\lambda_j(\BA_p(\bk)) \in [\lambda - 4\zeta(\lambda;\BA_1), \lambda +
  4\zeta(\lambda;\BA_1)]$ we have
  \begin{equation}
    \label{eq:evlabelclose}
    \abs{\lambda_j(\BA_1(\bk)) - \lambda_j(\BA_2(\bk))} \le \eps.
  \end{equation}
  Then, 
  \begin{equation}
    \zeta(\lambda;\BA_2) \ge \zeta(\lambda;\BA_1) - 2\eps.
  \end{equation}
\end{prop}
\begin{proof}
  If $2\eps > \zeta(\lambda;\BA_1)$, the result follows trivially from
  nonnegativity of the overlap function. Otherwise, choose $j \in \Z$ such that 
  $$[\lambda -
  \zeta(\lambda;\BA_1),\lambda + \zeta(\lambda;\BA_1)] \subset \iota_j(\BA_1).$$ Then
  \eqref{eq:evlabelclose} implies
  \begin{equation}
  [\lambda - \zeta(\lambda;\BA_1) + \eps,\lambda + \zeta(\lambda;\BA_1) - \eps]
  \subset \iota_j(\BA_1)_{-\eps} \subset \iota_j(\BA_2).
  \end{equation}
  The claim then follows from inspection of the definition \eqref{eq:defnoverlap} of the
  overlap function.
\end{proof}

\section{Systems of periodic operators -- The Bethe-Sommerfeld property}
\label{sec:bs}

In this section we prove that certain systems of periodic operators enjoy the
Bethe-Sommerfeld property in a quantitative way. This will imply that the spectrum of an elliptic periodic
operator $\BA$ in some classes contains a half-line.  Our proof is again based on a
reduction of the problem to uncoupled operators.

It is clear that if we show that the overlap function \eqref{eq:defnoverlap} is bounded away from $0$ at sufficiently
large $\lambda$ for some operator $\BA$, then $\BA$ has the Bethe--Sommerfeld
property.
This is  the strategy employed in \cite{ParSob2010}, where the
self-adjoint operators of the form
\begin{equation}
  A = (-\Delta)^\alpha + B,
\end{equation}
with $B\in \BS^\beta$, $\beta < 2\alpha$, and $B$ is $\Lambda$-periodic are
studied. 
 It is shown in \cite{ParSob2010} that there
are $S, c$ and $\lambda_0$, depending only on $\Theta$ and the symbol norms of
$B$, such that for all $\lambda \ge \lambda_0$, $\zeta(\lambda;H) \ge c
\lambda^S$. However, unlike in the situation of Proposition \ref{thm:aexpconcreteu},
it is no longer possible to simply extend the result to uncoupled operators,
i.e. direct sums of operators acting on scalar functions. While it is certainly
true that such operators enjoy the Bethe--Sommerfeld property, it could be possible \emph{a priori}
that the overlap function does not stay bounded away from $0$ when direct sums
are taken. This would imply
that the reduction to uncoupled operators would be able to open gaps. Our aim is to show that such a situation
is impossible for our class of operators.

\subsection{The Bethe-Sommerfeld property}
Our main theorem concerning systems of periodic operators is the following.
\begin{thm} \label{thm:bs}
  Suppose that $\BA \in \BE\BS^\alpha_m$, $\alpha > 0$ is periodic, self-adjoint
  \mpu{} operator
  with $\BA^\CD$ of the form \eqref{eq:defh0}
  and $a_j \ne a_k$ whenever $j \ne k$.
  Then, there exist positive $\tilde \lambda, S,c$ such
  that
  \begin{enumerate}
    \item  \label{it:bspos} if $\BA$ is unbounded above, $[\tilde \lambda,\infty) \subset
      \spec(\BA)$ and for every $\lambda \ge \tilde \lambda$, $\zeta(\lambda;\BA) \ge c
      \lambda^{-S}$;
    \item \label{it:bsneg} if $A$ is unbounded below, $(- \infty,-\tilde \lambda] \subset
      \spec(\BA)$ and for every $\lambda \ge \tilde \lambda$, $\zeta(-\lambda;\BA) \ge c
      \lambda^{-S}$.
  \end{enumerate}
  The overlap exponent $S$ depends only on $\alpha$ and the dimension $d$. The parameters
  $\tilde \lambda$ and $c$ can be chosen uniformly in the symbol norms of $\BA$ and
  $\BA^{\CO\CD}$.
\end{thm}

\begin{rem}
Saying that the parameters are chosen uniformly in the symbol norms means that
if $\BA$ and $\BA'$ are operators satisfying the conditions of Theorem \ref{thm:bs}
and for all $s,\ell$, $\snorm{\BA'}{s}{\ell} \le \snorm{\BA}{s}{\ell}$ and
$\snorm{\BA'^{\CO\CD}}{s}{\ell} \le \snorm{\BA^{\CO\CD}}{s}{\ell}$ then the
parameters obtained for $\BA$ also work for $\BA'$.
\end{rem}
Let us first observe that, as was the case in Sections \ref{sec:prelim} and
\ref{sec:asyexp},
it is sufficient to prove case (\ref{it:bspos}), case
(\ref{it:bsneg}) will then follow from the former applied to the operator $-A$. 

To prove Theorem \ref{thm:bs}, we proceed in two steps. The first one is the
following proposition, the proof of which is delayed until Section
\ref{sec:cg}.

\begin{prop}
  \label{thm:uncoupledbs}
  Let $\BA \in \BU\BS^\alpha_m \cap \BE\BS_m^\alpha$, $\alpha > 0$ be periodic,
  essentially self-adjoint, with $\BA^\CD$ of the form
  \eqref{eq:defh0}.
  Then, there exist $\tilde \lambda, S,c > 0$ such
  that
  \begin{enumerate}
    \item  \label{it:bsposuc} if $\BA$ is unbounded above, the interval $[\tilde \lambda,\infty) \subset
      \spec(\BA)$ and for every $\lambda \ge \tilde \lambda$, $\zeta(\lambda;\BA) \ge c
      \lambda^{-S}$;
    \item \label{it:bsneguc} if $\BA$ is unbounded below, the interval $(- \infty,-\tilde \lambda] \subset
      \spec(\BA)$ and for every $\lambda \ge \tilde \lambda$, $\zeta(-\lambda;\BA) \ge c
      \lambda^{-S}$.
  \end{enumerate}
  The overlap exponent $S$ depends only on $\alpha$ and $d$. The parameters
  $\tilde \lambda$ and $c$ can be chosen uniformly in the symbol norms of $\BA$ and
  $\BA^{\CO\CD}$.
\end{prop}

The second step on the way to proving Theorem \ref{thm:bs} is to show that small
perturbations of an operator (in the sense of perturbations that
do not change the integrated density of states much) cannot open a gap if the
overlap function is large enough.

\begin{lem}
  Suppose that $\BA \in \BE\BS^\alpha_m$ satisfies the hypotheses of Theorem
  \ref{thm:bs}. Then, for every $K \in \mathbb R$, there exists an
  operator $\BA_K \in \BU\BS^\alpha_m \cap \BS\BE\BS^\alpha_m$ satisfying the
  hypothesis of Proposition \ref{thm:uncoupledbs} such that
  for every $\lambda$ large enough, we have that
  \begin{equation}\label{eq:bsequiv}
    \zeta(\lambda;\BA) 
  \ge
\zeta(\lambda;\BA_K) + \bigo{\lambda^{-K}}.\end{equation}
\end{lem}
\begin{proof}
  From Remark \ref{rem:Iinfinitybs} and Theorem \ref{thm:ctou}, it is possible
  for any $K$ to find an operator $\BA_K \in \BU\BS_m^\alpha$ such that for $\eps = \lambda^{-\alpha - K}$ and
  any interval $I \subset [\lambda - 2 \zeta(\lambda;\BA),\lambda + 2
  \zeta(\lambda;\BA)]$ we
  have that for every $\bk \in \CO^\dagger$,
  \begin{equation} \label{eq:idsclose2}
    N(I;\BA(\bk)) \le N(I_\eps;\BA_K(\bk)) \quad \text{and} \quad N(I;\BA_K(\bk)) \le
    N(I_\eps;\BA(\bk)).
  \end{equation}
  As in \eqref{eq:smalldist}, put
  \begin{equation}
    \begin{aligned}
        \delta &= \min_{\bk \in \CO} \max\set{\dist(\mu;\spec(\BA(\bk)))
          : \mu \in [\lambda - 2\zeta(\lambda;\BA),\lambda +
        2\zeta(\lambda;\BA)]} \\
        & \ge C \frac{\zeta(\lambda;\BA)}{\max_\bk N([\lambda -
          2\zeta(\lambda;\BA),\lambda + 2 \zeta(\lambda;\BA)];\BA(\bk))}
    \end{aligned}
  \end{equation}
  for some $C> 0$. By Lemma \ref{lem:overlapidscont} and Proposition \ref{prop:overlapcontev},  if $\eps
  < \delta/4$, then \eqref{eq:idsclose2} implies \eqref{eq:bsequiv}. 
  It follows from Proposition \ref{thm:uncoupledbs} that $\zeta(\lambda;\BA) \ge
  \lambda^{-S}$ for some $S$. Weyl's law implies that
  \begin{equation}
  \max_\bk N([\lambda - 2
  \zeta(\lambda;\BA),\lambda + 2\zeta(\lambda;\BA)];\BA(\bk)) =
  \bigo{\lambda^{d/\alpha}}.
  \end{equation}
  It follows that by choosing $K > S + \frac d \alpha$, we have $\eps < \delta/4$ for
  $\lambda$ large enough, finishing the proof.
\end{proof}

Before making our constructions explicit in the next section, we prove the following lemma which
indicates that for uncoupled operators, we may suppose without loss of
generality that they are semi-bounded below.
\begin{lem}
  \label{lem:semibounded}
  Let $\BA \in \BU\BS^\alpha_m \cap \BS\BE\BS^\alpha_m$ is 
  self-adjoint and periodic, and suppose that
  \begin{equation}
    \BA = \BA_+ \oplus \BA_-
  \end{equation}
  with $\BA_+$ semi-bounded below and $\BA_-$ semi-bounded above. Then, Theorem \ref{thm:bs}, part \eqref{it:bspos} holds for
$\BA$ if and only if it holds for $\BA_+$.
  Similarly, Theorem \ref{thm:bs}, part \eqref{it:bsneg} holds if and only if it
  holds for $\BA_-$.
\end{lem}
\begin{proof}
  We prove it for the implication of Theorem \ref{thm:bs} part \eqref{it:bspos},
  the other one follows by replacing $\BA$ by $-\BA$. Since $\BA_-$ is semi-bounded
  above, there is some $M > 0$ such that $\spec(\BA_-) \cap (M,\infty) =
  \varnothing$. It therefore follows that for $\lambda > M$, 
  \begin{equation}
    \spec(\BA_+) \cap [\lambda,\infty) = \spec(\BA) \cap [\lambda,\infty)
  \end{equation}
  and our claim follows.
\end{proof}

\section{Bethe--Sommerfeld for uncoupled operators} \label{sec:cg}

In this section we prove Proposition \ref{thm:uncoupledbs}. 
Remark \ref{rem:Iinfinitybs}, Lemmas \ref{lem:overlapidscont} and
\ref{lem:semibounded}; and
Theorem \ref{thm:ctou} tell us that it is sufficient to prove Proposition
\ref{thm:uncoupledbs} after a few simplifying assumptions.
First, from now on we assume that $\BA \in \BU\BS^\alpha_m \cap
\BS\BE\BS_m^\alpha$ is semi-bounded below, in other words $a_j > 0$ for all $1
\le j \le m$. Then, we
also suppose that we have already applied the gauge transform in each component,
and that we have removed the remainder that changes the IDS by
$\bigo{\rho^{-K}}$ for some large $K$. In other words, we suppose that the
symbol of $\BA$ is given by
\begin{equation}
  \ba(\bx,\bxi) = \ba^\CD(\bxi) + \bbb^\CR(\bx,\bxi),
\end{equation}
where we have abused notation and used the same symbol for $\BA$ and the operator
after gauge transform. 

The study of the Bethe-Sommerfeld property in
high dimensions
\cite{barbpar,parnovski,ParSob2010} relies on variations of a certain combinatorial
geometric argument. The strategy goes as follows, using the notation from
\cite{ParSob2010}. First, find a function $g : \R^d \to
\spec(\BA)$ such that on fibres of the quotient map $\R^d \to
\R^d/\Theta$, $g$ is a bijection onto $\spec(\BA(\bk))$. Such a function
can be chosen in such a way that $$\abs{g(\bxi) - \abs \bxi^\alpha} =
\smallo{\abs \bxi^\alpha}.$$ 
The goal is then to show that for any large energy $\rho^\alpha$ the pre-image of some small
interval $[\rho^\alpha - \delta,\rho^\alpha + \delta]$ contains a curve connecting 
$g^{-1}(\rho^\alpha - \delta)$ and $g^{-1}(\rho^\alpha + \delta)$, and that $g$ is continuous along that curve.
Here, $\delta$ is a parameter that depends
on $\rho$ and corresponds to the overlap length at energy level $\rho^\alpha$. 

In order to show that such a curve exists we split the pre-image
$g^{-1}([\rho^\alpha - \delta, \rho^\alpha + \delta])$ into  bad (resonant) and good
(non-resonant) regions. The latter are defined in such a way that $g$ is
continuous and radially increasing within
them. The idea in \cite{barbpar,parnovski,ParSob2010} was to find a small radial interval in the
non-resonant region where the eigenvalues of $\BA(\bk)$ are all simple along that
small interval, which in turn gives a lower bound for the overlap function. 

This is done by showing that not only is the angular measure of the
resonant region small, but that the total measure of intersections of 
translates of the resonant regions by elements of $\Theta$ with the
non-resonant region is also small, in comparison to the volume of the
non-resonant region. This implies the existence of a radial interval along which
translations by $\btheta \in \Theta$ stay inside the non-resonant region,
which gives us the desired interval where the eigenvalues of $A(\bk)$ are
simple and increasing.

In the case of an operator $\BA \in \BU\BS^\alpha_m$ acting in $\RL^2(\R^d;\C^m)$, hoping for the
eigenvalues of $\BA(\bk)$ to be simple is unrealistic. However, such a requirement
is not necessary. Indeed, suppose that the eigenvalues of operator $H$ acting in
$\RL^2(\R^d;\C)$ are simple.
Then, since the spectrum of $H \oplus H$ is the same as for $H$, 
$H \oplus H$ enjoys the Bethe--Sommerfeld property if and only if it $H$ also does. However, the eigenvalues of $(H \oplus H) (\bk)$ are always at least double.
This means that we need a new way of reasoning. We obtain the following
conditions, which are sufficient to obtain lower bounds on the overlap function.

\begin{prop}
  \label{prop:reformulation}
  Suppose that there is an interval $[\bk_1,\bk_2]\subset \CO^\dagger$ and three
  families of real-valued continuous functions $\bmu$, $\bnu$
  and $\btau$ on $[\bk_1,\bk_2]$, satisfying the following properties.
  \begin{itemize}
    \item The equality of multisets (i.e. taking multiplicity into account)
      \begin{equation}
        \label{eq:firstbp}
        \begin{aligned}
          \set{\lambda_p(\BA(\bk)) ; p \in \N} &= \set{\mu(\bk) : \mu \in \bmu} \cup
        \set{\nu(\bk) : \nu \in \bnu} \cup \\ &\qquad \cup \set{\tau(\bk)
        : \tau \in \btau}
      \end{aligned}
    \end{equation}
      holds for all $\bk \in [\bk_1,\bk_2]$.
    \item For every $\mu \in \bmu$, $\mu$ is increasing as $\bk$ goes from
      $\bk_1$ to $\bk_2$. Furthermore, either
      \begin{equation}
        \label{eq:mucrossenough}
        \mu(\bk_1) < \rho^\alpha - \delta \qquad \text{or} \qquad \mu(\bk_2) >
        \rho^\alpha + \delta.
      \end{equation}
    \item We have that $\#\bnu < \infty$. Furthermore,
      \begin{equation} \label{eq:secondbp} \#\set{\mu \in \bmu : [\rho^\alpha - \delta, \rho^\alpha + \delta] \subset
      \mu([\bk_1,\bk_2])} > \#\bnu.\end{equation}
    \item For all $\tau \in \btau$, we have that
      \begin{equation} \label{eq:thirdbp} \tau([\bk_1,\bk_2]) \cap [\rho^\alpha - \delta, \rho^\alpha +
      \delta] = \varnothing.\end{equation}
  \end{itemize}
  Then, $\zeta(\rho^\alpha;\BA) \ge \delta$.
\end{prop}
\begin{proof}
  Since $\BA$ is semi-bounded below, we can use Skriganov's characterisation
  \cite{Skriganov} of the overlap function as
  \begin{equation}
    \label{eq:overlap}
    \zeta(\rho^\alpha,\BA) = \sup\set{t : \min_\bk N( (-\infty,\rho^\alpha +
    t];\BA(\bk)) < \max_\bk N( (-\infty,\rho^\alpha - t);\BA(\bk))}.
  \end{equation}
  It  then follows that if
  \begin{equation}
    \label{eq:goalovlp}
    N(\rho^\alpha - \delta;\BA(\bk_1)) - N(\rho^\alpha + \delta;\BA(\bk_2)) \ge 1,
  \end{equation}
  then $\zeta(\rho^\alpha;\BA) \ge \delta$. Equation \eqref{eq:firstbp} implies that we
  can separate the differences of each family
  $\bmu$, $\bnu$ and $\btau$ in \eqref{eq:goalovlp}. Equation
  \eqref{eq:thirdbp} ensures that 
  \begin{equation}
    \label{eq:diff1}
    \#\set{\tau \in \btau : \tau(\bk_1) \le \rho^\alpha - \delta} - 
    \#\set{\tau \in \btau : \tau(\bk_2) \le \rho^\alpha + \delta} = 0.
  \end{equation}
  We also have that
  \begin{equation} 
    \label{eq:diff2}
   \bigg| \#\set{\nu \in \bnu :  \nu(\bk_1) < \rho^\alpha - \delta} - 
   \#\set{\nu \in \bnu :  \nu(\bk_2) \le \rho^\alpha + \delta}\bigg| \le \#\bnu.
  \end{equation}
  Finally, from \eqref{eq:mucrossenough}, \eqref{eq:secondbp} and the fact that all functions
  $\mu \in \bmu$ are increasing, we have that
  \begin{equation}
    \label{eq:diff3}
  \#\set{\mu \in \bmu :  \mu(\bk_1) \le \rho^\alpha - \delta} - 
  \#\set{\mu \in \bmu :  \mu(\bk_2) \le \rho^\alpha + \delta} > \#\bnu.
  \end{equation}
  Indeed, assumption \eqref{eq:mucrossenough} ensures that every $\mu \in \bmu$
  contributes at least $0$ to the lefthand side of \eqref{eq:diff3}, whereas
  \eqref{eq:secondbp} ensures that there at least $\#\bnu + 1$ functions in
  $\bmu$ contributing $1$.  
  Combining \eqref{eq:diff1}--\eqref{eq:diff3} yields \eqref{eq:goalovlp}.
\end{proof}

In the remainder of this section, we therefore set out to prove the existence of
the interval $[\bk_1,\bk_2]$ as well as of the families of functions
$\bmu,$ $\bnu$ and $\btau$ satisfying the hypotheses
of Proposition \ref{prop:reformulation}. We will, as we go along, use relevant results from \cite{ParSob2010},
mostly about properties of the function $g$ and volume of various sets
associated with the preimage.

\subsection{A description of the spectrum}

In \cite{ParSob2010}, Parnovski and Sobolev have constructed functions $g_j : \R^d \to
\spec(A_j)$. When restricted to a coset $\set \bxi + \Theta$, $g_j$ is a
bijection (respecting multiplicity) with $\spec(A_j(\set\bxi))$. It is useful to
use $g_j$ to define a family of functions on the fundamental cell $\CO$
indexed by $\bp := (j,\btheta) \in \set{1,\dotsc,m} \times \Theta =:
\tilde \Theta$ as
\begin{equation}
  g_\bp(\bk) := g_j(\bk + \btheta).
\end{equation}
Since we assume in this section that 
\begin{equation}
  \BA = A_1 \oplus \dotso \oplus A_m \in \BU\BS_m^\alpha,
\end{equation}
and that is is a semi-boudned below operator, we have by construction that
$\set{\lambda_p(\BA(\bk)) : p \in \N}$ is a relabeling of
$\set{g_{(j,\btheta)}(\bk) : (j,\btheta) \in \tilde \Theta}$ in
increasing order. In particular, the latter set is discrete, bounded from below
and accumulates only at infinity. 

To prove Proposition \ref{thm:uncoupledbs}, it remains to show that there exists
$S \in \R$ such that the hypotheses of Proposition
\ref{prop:reformulation} hold for some $\delta \gg \rho^S$. The families
$\bmu$, $\bnu$ and $\btau$ will be realised as a partition of the set $\set{g_\bp
: \bp \in \tilde \Theta}$.
\begin{prop}
  \label{prop:counting}
    There exist $\rho_0 > 0$ and $S \in \R$ (depending on
    $\set{a_1,\dotsc,a_m}$, $\alpha$ and the implicit constants in
  \eqref{eq:condgsmall}) such that for all $\rho \ge \rho_0$, there exists an
    interval $[\bk_1,\bk_2]$ and three families of functions $\bmu, \bnu, \btau$ on
    $[\bk_1,\bk_2]$ satisfying the hypotheses
    \eqref{eq:firstbp}--\eqref{eq:thirdbp} of Proposition
    \ref{prop:reformulation}. 
\end{prop}
The remainder of this section is dedicated to the proof of this proposition. We give some useful properties of the functions $g_j$. For $1 \le j \le m$,
there are $a_j > 0$ and $G_j : \R^d \to \R$ such that
\begin{equation}
  \label{eq:gsmall}
  g_j(\bxi) = a_j \abs \bxi^\alpha + G_j(\bxi), \qquad
  a_j
  > 0.
\end{equation}
Furthermore, whenever 
\begin{equation}
  \label{eq:range}
  \abs \bxi \asymp \rho,
\end{equation}
we have 
\begin{equation}
  \abs{G_j(\bxi)} \ll \rho^\beta
  \label{eq:boundGj}
\end{equation}
for some $\beta < \alpha$, the implicit
constant depending on the ones in \eqref{eq:range}.

We need to control the
first and second derivatives of the functions $G_j$ on a large enough region. To
this end, for $\delta \in (0,\rho^\alpha/4)$, define the “annular” regions
\begin{equation}
  \CA_j := \CA_j(\rho;\delta) := g_j^{-1}\left(\left[ \rho^\alpha - \delta, \rho^\alpha +
  \delta\right] \right).
\end{equation}
We suppose that for every $0 < \delta < \rho^\alpha/4$, $\CA_j$ can be
decomposed further into a non-resonant set $$\CB_j := \CB_j(\rho,\delta) \subset
\CA_j$$ and a resonant set $$\CR_j := \CA_j \setminus \CB_j.$$ A precise
description of those sets is given in \cite{ParSob2010}; the non-resonant sets
$\CB_j$ correspond to simple eigenvalues of the operators $A_j(\set{\bxi})$, 
whereas the resonant sets $\CR_j$ correspond to clusters of eigenvalues. In
particular, when $\bxi \in \CB_j$, 
\begin{equation}
  \label{eq:nrinv}
\BA e_{\bxi} \otimes v_j  = g_j(\bxi) e_{\bxi}
\otimes v_j
.
\end{equation}An
important property of the non-resonant set is that the functions $g_j$ behave
well there. We suppose that the restriction of $G_j$ to $\CB_j$ is of class $\RC^2$
and that for all $\bxi \in \CB'_j$, the bounds 
\begin{equation}
  \begin{aligned}
  \label{eq:condgsmall}
  \abs{\nabla G_j(\bxi)} &\ll \rho^\gamma, \\
  \abs{\nabla^2 G_j(\bxi)} &\ll \rho^\sigma
\end{aligned}
\end{equation}
hold, where $\gamma < \alpha - 1$ and $\sigma < \alpha - 2$, and
the implicit constants are once again allowed to depend on the implicit
constants in
\eqref{eq:range}. In order to prove Proposition \ref{prop:counting}, it will be
sufficient to show that the sets $\CB_j$ are, in a sense, large enough. In subsection \ref{sec:constr} we constrain the radial projection on
the sphere of those resonant and non-resonant sets.

\subsection{Description of the resonant sets} \label{sec:constr}

The description of the resonant sets in this section expands on the
results in \cite[Sections 5,7,8,9]{ParSob2010}. We first introduce some notation.
Most of the sets we define depend on the parameters $\rho$ and
$\delta$; however, we will be quick to drop the (explicit) dependence on these
parameters to make notation lighter.
For every $\bxi \in \R^d \setminus \set 0$, let $\bu_{\bxi} := \abs \bxi^{-1}\bxi$ be the
unit vector in the direction of $\bxi$. For $\nu > 0$ and $\btheta \in
\Theta \setminus \set 0$,
we define spherical resonant regions as
\begin{equation}
  \CS(\btheta;T) := \set{\bzeta \in \BBS^{d-1} : \abs{\bzeta \cdot \bu_{\btheta}}
  < T}.
\end{equation}
For any subset $\CU$ of the sphere $\BBS^{d-1}$, we denote its radial extension
by
$$
\CU_{rd} := \set{\bxi \in \R^d : \bu_{\bxi} \in \CU}.
$$
The following lemma is proved in \cite[Section 5]{ParSob2010}.
\begin{lem}
  \label{lem:ParSob2010}
  There exists $c$ (depending on the coefficients $a_j$) such that for every $\rho$ large enough, and $\delta \in  (0,c\rho^\alpha)$, there are
  $\varkappa$, and $\nu$ such that $$0
  < \varkappa < d^{-2}, \qquad  \Theta_\varkappa :=  (\Theta\setminus
  \set 0) \cap
  B(\rho^\varkappa), \qquad  \varkappa d < \nu < 1$$ and such that for all $1 \le j \le
  m$,
  \begin{equation}
    \label{eq:crjsub}
    \CR_j(\rho;\delta) \subset \CA_j \cap \bigcup_{\btheta \in \Theta_\varkappa}
    \CS(\btheta;\rho^{-\nu})_{rd}.
  \end{equation}
  Denoting
  \begin{equation}
    \CT(\rho) := \BBS^{d-1} \setminus \bigcup_{\btheta \in \Theta_\varkappa}
    \CS(\btheta;\rho^{-\nu})
  \end{equation}
and
\begin{equation}
  \tilde \CB_j(\rho) := \CA_j \cap \CT_{rd},
\end{equation}
we also have that $\tilde \CB_j \subset \CB_j$. 
\end{lem}
\begin{rem}
 The inclusion \eqref{eq:crjsub} means that the radial projection to the sphere
 of the resonant region is included in
 the union of the sets $\CS(\btheta;\rho^{-\nu})$, justifying naming them
 spherical resonant regions. Consequently, we name $\CT$ the spherical
 non-resonant region.
\end{rem}

The object of the next lemma is that small enough neighborhoods of the spherical
resonant regions have small volume.

\begin{lem} \label{lem:enlargedresonant}
  Let $\varkappa$, $\Theta_\varkappa$ and $\nu$ be obtained in Lemma \ref{lem:ParSob2010} and define
\begin{equation}
  \tilde \CT(\rho) := \BBS^{d-1} \setminus \left(\bigcup_{\btheta \in
  \Theta_\varkappa} \CS(\btheta;2\rho^{-\nu})\right),
\end{equation}
\begin{equation}
  \label{eq:czjsub}
  \CZ_j(\rho;\delta) := \CA_j \cap \bigcup_{\btheta \in \Theta_\varkappa}
  \CS_{rd}(\btheta;2\rho^{-\nu}),
\end{equation}
and
\begin{equation}
  \label{eq:cgjsub}
  \CG_j(\rho;\delta) := \CA_j \cap \tilde \CT_{rd} = \CA_j \setminus \CZ_j.
\end{equation}

Then, $\CG_j \subset \CB_j$, and for all $\bxi \in \CG_j$, $\dist(\bxi,\CR_j)\gg
\rho^{1-\nu}$. Furthermore, for $\eps_0 = \nu - d \varkappa > 0$,
\begin{equation}
  \label{eq:volCZ}
  \vol(\CZ_j(\rho;\delta)) \ll \delta \rho^{d - \alpha - \eps_0}.
\end{equation}
and
\begin{equation}
  \label{eq:voltCG}
  \vol(\CG_j(\rho;\delta)) \asymp \delta \rho^{d - \alpha}.
\end{equation}
\end{lem}
\begin{proof}
  It is clear from the definition that $\CG_j \subset \tilde \CB_j \subset
  \CB_j$. 
  Define the sets
  \begin{equation}
    \CI_j(\bzeta) := \CA_j \cap \set{\bzeta}_{rd}.
  \end{equation}
  It follows from \cite[pp.518--519]{ParSob2010} that for all $\bzeta$,
  $\CI_j(\bzeta)$ is an interval of length $\abs{\CI_j} \ll \delta
  \rho^{1 - \alpha}$ (uniformly in $\bzeta$), and 
  \begin{equation}
    \CI_j \subset \set{\bxi : \abs \bxi \asymp \rho }.
  \end{equation}
  
  Furthermore, by definition
  $$
  \dist\left(\tilde \CT(\rho), \bigcup_{\btheta \in \Theta_\varkappa}
  \CS(\btheta;\rho^{-\nu})\right) > \rho^{-\nu}.
  $$
  It therefore follows from \eqref{eq:crjsub}, \eqref{eq:cgjsub} and basic
  trigonometry 
  that $\dist(\CG_j,\CR_j) \gg \rho^{1-\nu}$. For the volume
  estimate for $\CZ_j$, we compute
  \begin{equation}
    \label{eq:volczj}
    \begin{aligned}
      \vol(\CZ_j) &\le \sum_{\btheta \in \Theta_\varkappa} \int_{\tilde
      \CS(\btheta;\rho^{-\nu})} \int_{\CI_j(\bzeta)} t^{d-1} \de t \de \bzeta \\
    & \ll \#(\Theta_\varkappa)\max_{\btheta}\vol_{d-1}(
    \CS(\btheta;2\rho^{-\nu}))\delta\rho^{d - \alpha}.
  \end{aligned}
  \end{equation}
  Uniformly in $\btheta$ we have that, $\vol_{d-1}(
    \CS(\btheta;2\rho^{-\nu})  \ll \rho^{-\nu}$. We also have that $\#\Theta_\varkappa \ll
    \rho^{d\varkappa}$.
    Putting these two estimates in \eqref{eq:volczj}  yields \eqref{eq:volCZ}.
    For the estimate on $\vol(\CG_j(\rho;\delta))$, we observe that
    $\vol(\CG_j) = \vol(\CA_j) - \vol(\CZ_j)$ and that by \eqref{eq:gsmall},
    \begin{equation}
      \vol(\CA_j) \asymp \delta \rho^{d - \alpha}.
    \end{equation}
    Estimate \eqref{eq:voltCG} then follows from the fact that $\vol(\CZ_j) =
    \smallo{\vol(\CA_j)}$.
\end{proof}

\subsection{Volumes of intersections}
For $\bb_1,\bb_2 \in \R^d$ and $i,j \in \set{1,\dotsc,m}$ we define the
\emph{crossing sets}
\begin{equation}
  \label{eq:defCX}
\CX_{ij}(\rho,\delta,\bb_1,\bb_2) := (\CA_i(\rho;\delta)+\bb_1)\cap(\CA_j(\rho;\delta)+\bb_2).
\end{equation}
We are interested in volume estimates, and since
\begin{equation}
  \vol(\CX_{ij}(\rho,\delta,\bb_1,\bb_2)) =
  \vol(\CX_{ij}(\rho,\delta,\bzero,\bb_2-\bb_1)),
\end{equation}
we restrict ourselves to sets of the form
\begin{equation}
  \CX_{ij}(\bb) := \CX_{ij}(\rho,\delta,\bb) :=
  \CX_{ij}(\rho,\delta,\bzero,\bb).
\end{equation}
 Denote by $\phi(\ba,\bb)$ the (smaller) angle
between $\ba$ and $\bb$. For any angle $\omega \in [0,\pi]$, we define the set
\begin{equation}
  \label{eq:badvol}
  \CX_{ij,\omega}(\bb) := \set{\bxi \in \CX_{ij}(\bb) : \phi(\bxi,\bxi-\bb) >
  \omega}.
\end{equation}
From \cite[Section 9]{ParSob2010} since $g_i, g_j$ are defined as
in \eqref{eq:gsmall}, so that the conditions \eqref{eq:condgsmall} are
respected, the following holds: for any $\omega \in (0,\pi)$,  $\eps >
0$, if $\delta \rho^{2 - \alpha + 2 \eps} \to 0$
as $\rho \to \infty$, then
\begin{equation}
\vol(\CX_{ij,\omega}(\rho,\delta,\bb)) \ll \delta^2 \rho^{4 - 2\alpha + d + 6\eps}
+ \delta \rho^{1 - \alpha - \eps(d-1)},
  \label{eq:nokissmu}
\end{equation}
uniformly in $\bb$.
\begin{prop}
\label{rem:xemptyint}
 There are constants $c$ and $C$ depending on $\delta$, $\omega$, and the
 numbers $a_j$ such that $\CX_{ij,\omega}(\bb) \ne \varnothing$ implies that
 $c \rho \le \abs\bb \le C\rho$.
\end{prop}

\begin{proof}
  We first make the observation that there exists $C > 0$, depending on $\alpha$
  and the numbers $a_j$ such that if $\abs \bbb > C\rho$, then for $\rho$ large enough $\CA_i \cap (\CA_j + \bbb) = \varnothing$.
  On the other
    hand, it follows from basic planar trigonometry that for every $\omega$, there
    exists $c$, depending on the constants in $\abs\bxi \asymp \rho$, such
    that if $\abs \bbb < c \rho$, then $\phi(\bxi, \bxi-\bbb)\le \omega$.
\end{proof}

\begin{rem}
  The estimates in \eqref{eq:nokissmu} were established under stronger
  hypotheses in \cite[Theorem 9.1]{ParSob2010}. This was because of the fact
  that their methods required control of the volume of $\CX_{ij}(\bbb)$, for all
  $1 \ll \abs{\bb} \ll \rho$, while we only need to control the volume of
  $\CX_{ij,\omega}(\bbb)$ for $\abs{\bbb} \asymp \rho$. The stronger conditions
  were needed to treat the case of small angles, and small $\bb$ (that is, such
  that $\abs \bb = \smallo \rho$). Our approach does not require any such
  estimates.
\end{rem}
Before going on, let us make the following notational convention.
\begin{convention}
 For any family of subsets $\CE(\delta) \subset \R^d$ depending on the parameter
 $\delta > 0$, we denote
$$\CE'(\delta) := \CE(Z\delta),$$
where $Z$ is some large constant to be determined later and depending only on
the dimension $d$, the order
$\alpha$ and the numbers $\set{a_1,\dotsc,a_m}$.
\end{convention}
We now define crossing sets for the non-resonant sets $\CG_j$. For $\bbb \in
\R^d$, let
\begin{equation}
  \CY_{ij}(\bb) := \CG_{i} \cap (\CG_j +
  \bb)
\end{equation}
and for any angle $\omega \in (0,\pi)$, 
\begin{equation}
  \CY_{ij,\omega}(\bb) := \set{\bxi \in \CY_{ij}(\bbb) : \phi(\bxi,\bxi-\bb)
  > \omega} = \CX_{ij,\omega}(\bbb) \cap \CY_{ij}(\bbb).
  \end{equation}
  We need the following lemma.
  \begin{lem} \label{lem:totalintersection}
    For any $\omega \in (0,\pi)$ and $\eps > 0$, the condition $\delta \rho^{2 -
      \alpha + 2\eps} \to 0$ as $\rho \to \infty$ implies
    \begin{equation}
      \vol\left( \bigcup_{i,j = 1}^m \bigcup_{\btheta \in \Theta}
      \CY_{ij,\omega}'(\btheta) \right) \ll 
        \delta^2 \rho^{4 - 2\alpha + 2 d +6\eps} + \delta \rho^{1- \alpha + d - \eps(d-1)},
      \label{eq:sobranie}
    \end{equation}
    the implicit constants depending only on $\delta$, $\omega$, $Z$, and the coefficients
    $a_j$. 
  \end{lem}

  \begin{proof}
It is
sufficient to prove the result for a single pair $i,j$, then sum the estimates
over all $m^2$ of those pairs. From Proposition \ref{rem:xemptyint}, there are
constants $c$ and $C$ depending only on $\omega,\delta,T$ and the numbers $a_j$ such that
    \begin{equation}
      \begin{aligned}
        \vol\left(\bigcup_{\btheta \in \Theta}\CY_{ij,\omega}'(\btheta)\right)
        &\le \sum_{\substack{\btheta \in \Theta \\ c \rho \le \abs\btheta
        \le C\rho } }\vol\left( \CY_{ij,\omega}(\btheta) \right)\\
          &\ll\delta^2 \rho^{4 - 2\alpha + 2 d + 6 \eps} + \delta \rho^{1-
          \alpha + d - \eps(d-1)},
        \end{aligned}
    \end{equation}
  where the last line comes from $\CY_{ij,\omega}(\btheta) \subset
  \CX_{ij,\omega}(\btheta)$, estimate \eqref{eq:nokissmu} and the fact that
  \begin{equation}
    \#\set{\btheta \in \Theta^\dagger : c \rho \le \abs \btheta \le C \rho} \ll
    \rho^d.
  \end{equation}
  \end{proof}

  \subsection{Estimates on the overlap function}
 Define
  \begin{equation}
    \CZ := \bigcup_{1 \le j \le m} \CZ_j, \qquad \qquad \text{and} \qquad \qquad
    \CG := \bigcup_{1 \le j \le m} \CG_j.
  \end{equation}
  It follows directly from the definitions of $\CZ_j$ \eqref{eq:czjsub} and $\CG_j$ 
  \eqref{eq:cgjsub} that $\CG \cap \CZ = \varnothing$. Furthermore, for every
  $\bxi \in \CG$ and every $1 \le j \le m$, $\bxi \in \CG_j$ or $\bxi \not \in
  \CA_j$. 

  For every $\bxi \in \R^d$, and any subset $E \subset \R^d$, we define
  \begin{equation}
    n(\bxi,E) = \#\set{\btheta \in \Theta : \bxi + \btheta \in E}.
  \end{equation}
 
  Let us take a step back to see what is needed in order to prove Proposition
  \ref{prop:counting}. The estimates on the functions $G_j$ in the non-resonant
  regions $\CB_j$ ensure that $g_j$ are radially increasing in those regions. 

  We need three
  ingredients in order to prove the existence of a path (of the form, before
    projection down to $\CO$, $[t_1,t_2]\bzeta$ for some $\bzeta \in
  \BBS^{d-1}$) satisfying hypotheses \eqref{eq:firstbp} -- \eqref{eq:thirdbp} in
  Proposition \ref{prop:counting}.
  \begin{itemize}
    \item That for $\bxi \in \CG_j$, the functions $g_j$ are increasing not only
      in the radial direction but also along directions deviating from the
      radial one by angles smaller than $\pi/4$, at least for some controllable
      distance. 
    \item That there are points $\bxi \in \CG$ such that $n(\bxi;\CG) >
      n(\bxi;\CZ')$. 
     \item That there are some of those $\bxi$ such that for all $\btheta \in
       \Theta'=\Theta\setminus \set{0}$ satisfying
       $\bxi + \btheta \in \CG'$, we have that $\phi(\bxi,\bxi+\btheta) <
    \pi/4$. 
  \end{itemize}
  The first of those statements is proved directly. For the last two, we
  will show that the proportion of $\bxi \in \CG$ such that 
  either $n(\bxi;\CZ') \ge n(\bxi;\CG)$ or for which there exists $\btheta \in \Theta'$ such that
  $\bxi+\btheta \in \CG'$ and $\phi(\bxi,\bxi+\btheta) > \pi/4$ is smaller than
  $1$.
  Actually, for some appropriate choice of the free parameter $\delta$, we will in
  fact show that this proportion can be made arbitrarily small.

  The following lemma addresses the first bullet point.
  \begin{lem} \label{lem:increasing}
    Let $\nu$ be as obtained in Lemma \ref{lem:ParSob2010}. For all $\bxi \in
    \CG_j$, and all $\bbb$ such that $\abs{\bxi + \bbb} \asymp \rho$ and
    $\phi(\bxi,\bxi+\bbb)  \le \pi/4$ there is a $t_0 \gg \rho^{-\nu}$ such that for all $t \in
    [-t_0,t_0]$ the function $t \mapsto g_j(\bxi + t(\bxi + \bbb))$ is
    increasing and
    \begin{equation}
      \frac{\de}{\de t}g_j(\bxi + t(\bxi + \bbb)) \gg \rho^\alpha. 
    \end{equation}
    The implicit constant in $t_0 \gg \rho^{-\nu}$ depends only on 
  the functions $g_j$ and the implicit constants in $\abs{\bxi + \bbb} \asymp
  \rho$. 
  \end{lem}

  \begin{proof}
    Since $\bxi \in \CG_j$, we have that not only $\bxi \in \CB_j$ for some $j$,
    but also, by Lemma \ref{lem:enlargedresonant}, that there exists $r > 0$
    such that for all $j'$,
    $\dist(\bxi;\CR_{j'}) > r \rho^{1-\nu}$. Therefore, for $\abs t \le t_0 :=  r \rho^{-\nu}$, we have that
    $\bxi + t(\bxi + \bbb) \in \CB'_j$. By \eqref{eq:condgsmall}, we have that
    \begin{equation}
      \abs{\frac{\de}{\de t} G_j(\bxi + t(\bxi + \bbb))} \ll \rho^{\gamma+1} 
      = \smallo{ \rho^\alpha}.
    \end{equation}
     On the other hand,
    \begin{align}
      \frac{\de}{\de t} \abs{\bxi + t(\bxi + \bbb)}^\alpha &= \alpha
      \abs{\bxi + t (\bxi + \bbb)}^{\alpha - 2}\times \\ & \qquad \times \left(\abs \bxi \abs{\bxi +
  \bbb}\cos\left( \phi(\bxi,\bxi+\bbb)\right) + t \abs{\bxi+\bbb}^2 \right) \\
 & \gg \rho^\alpha,
    \end{align}
    where the last line holds from the fact that $\cos \phi(\bxi,\bxi + \bbb) > \sqrt{2}/2$.
  \end{proof}
  
  \begin{lem} \label{lem:numinter}
    Let
    \begin{equation}
      \CN := \set{\bxi \in \R^d : n(\bxi;\CG) \le m n(\bxi;\CZ')},
    \end{equation}
    and $\CN_\CG = \CN \cap \CG$.
    Then, we have that
    \begin{equation}
      \vol(\CN_\CG) \ll \vol(\CZ') \ll
      \delta\rho^{d - \alpha - \eps_0},
    \end{equation}
    where $\eps_0$ is as in Lemma
    \ref{lem:enlargedresonant}.
  \end{lem}
  \begin{proof}
    Observe first that for any $E \subset \R^d$, the function $n(\bxi;E)$ is constant on the
    fibres $\bxi \mod \Theta$,
    as such it is well defined on $\CO$ and $\CN$ is invariant under the
    action of $\Theta$. We therefore have that
    \begin{align}
      \vol(\CN_\CG) &= \int_{\CN/\Theta} n(\bxi;\CG) \de \bxi \\
      &\le m \int_{\CN/\Theta}  n(\bxi;\CZ') \de \bxi \\
      & = m \vol(\CN \cap \CZ' ) \\
      &\le m \vol(\CZ').
    \end{align}
    The claim now follows from Lemma \ref{lem:enlargedresonant}.
  \end{proof}

  \begin{lem} \label{lem:goodcrossings}
    Let
    \begin{equation}
      \begin{aligned}
      \CU &:=\big\{\bxi \in \CG \setminus \CN : \bxi + \btheta_1 \in
      \CY_{ij,\pi/4}'(\btheta_1 - \btheta_2) \\ &\qquad \quad  \text{ for some } 
        1 \le i,j \le m
    \text{ and } \btheta_1,\btheta_2 \in \Theta\big\}.
      \label{eq:CU}
    \end{aligned}
    \end{equation}
    Then, if $\eps$ and $\delta$ are such that $\delta \rho^{d - \alpha + 2 \eps} \to 0$, we
    have
    \begin{equation}
      \vol(\CU) \ll \delta^2 \rho^{4 - 2 \alpha + 2 d + 6\eps} + \delta \rho^{1
      - \alpha + d - \eps(d-1)}.
    \end{equation}
  \end{lem}
  \begin{proof}
    Suppose that $\bxi \in \CU \subset \CG$, so that $\bxi \in \CG_k$ for some
    $1 \le k \le m$. Consider the lattice elements $\btheta_1, \btheta_2 \in
    \Theta^\dagger$ such that $\bxi + \btheta_1 \in \CY_{ij,\pi/4}'(\btheta_1 -
    \btheta_2)$. By definition of $\CY_{ij,\pi/4}'$ and translation, this means that
    \begin{equation}
      \bxi \in (\CG_i' - \btheta_1) \cap (\CG_j' - \btheta_2),
    \end{equation}
and therefore that
\begin{equation}
  \bxi \in \CY_{ki}'(-\btheta_1) \cap \CY_{kj}'(- \btheta_2).
\end{equation}
Furthermore, $\phi(\bxi+\btheta_1,\bxi+\btheta_2) > \pi/4$. As such,
\begin{equation}
  \max \set{\phi(\bxi,\bxi+\btheta_1),\phi(\bxi,\bxi+\btheta_2)} > \pi/8.
\end{equation}
Combining the previous two displays yields
\begin{equation}
  \bxi \in \CY_{ki,\pi/8}'(-\btheta_1) \cup \CY_{kj,\pi/8}(-\btheta_2).
\end{equation}
Therefore,
    \begin{equation}
      \begin{aligned}
        \vol(\CU) &\le \vol\left( \bigcup_{i,j = 1}^m \bigcup_{\btheta \in
        \Theta} \CY'_{ij,\pi/8}(\btheta) \right)  \\
        &\ll \delta^2 \rho^{4 - 2 \alpha + 2d + 6\eps} + \delta \rho^{1 -
        \alpha + d - \eps(d - 1)},
    \end{aligned}
    \end{equation}
    the last line holding by virtue of Lemma \ref{lem:totalintersection}.
  \end{proof}
  \begin{prop}\label{prop:overlap}
   Let 
   \begin{equation}
     \label{eq:propoverlap}
     s := \min \set{\frac{\alpha d  - d^2 - 3d - \alpha - 2}{2(d+2)}, \alpha - d
     + \frac{\alpha - d - 2}{2(d+2)}}.
   \end{equation}
   For $\rho$ large enough and $\delta = \smallo{\rho^s}$,
   the set
   $$
   \CK:= \CG \setminus(\CN_\CG \cup \CU)   
   $$
   is non empty.
  \end{prop}
  \begin{proof}
    Recall from Lemma \ref{lem:enlargedresonant} that $\vol(\CG) \asymp \delta
    \rho^{d-\alpha}$. On the other hand we have from Lemma \ref{lem:numinter}
    that there is $\eps_0 > 0$ such that
    \begin{equation}
      \vol(\CN_\CG) \ll \delta\rho^{d- \alpha - \eps_0},
    \end{equation}
    and from Lemma \ref{lem:goodcrossings} that as soon as $\delta \rho^{d -
    \alpha + 2 \eps} \to 0$ we have that
    \begin{equation}
      \vol(\CU) \ll \delta^2 \rho^{4 - 2 \alpha + 2 d + 6 \eps} + \delta \rho^{1
      - \alpha + d - \eps(d-1)}.
    \end{equation}
 Take 
    \begin{equation}
    \label{eq:freefixed}
    \eps = \frac{\alpha - d - 2}{2(d+2)}.
  \end{equation}
  Observe that indeed when $\delta = \smallo{\rho^s}$ we have 
  $\delta \rho^{2 - \alpha + 2 \eps} \to 0$ as $\rho \to \infty$. We also
  observe that with that choice of parameters $\vol(\CU) + \vol(\CN_\CG) =
  \smallo{\vol(\CG)}$ and hence, for large enough $\rho$, $\CK$ is not empty.
\end{proof}  

We now have all the necessary ingredients to prove Proposition
\ref{prop:counting}.

\begin{proof}[Proof of Proposition \ref{prop:counting}]
  Let $s$ be defined as in \eqref{eq:propoverlap}. For any $\eps > 0$, set $s' = \min\set{s - \eps,\alpha
  -\nu}$ where $\nu$ is obtained in Lemma \ref{lem:ParSob2010}, and put $\delta
  = \rho^{s'}$. Then, by Proposition \ref{prop:overlap} $\CK$ is nonempty for
  $\rho$ large enough; choose $\bxi_0 \in \CK$. 
  For $1 \le j \le m$, let $\Gamma_j, \Gamma_j' \subset \Theta$ be defined as
      \begin{equation}
        \Gamma_j := \set{\btheta \in \Theta : \bxi_0 + \btheta
        \in \CG_j}, \qquad \Gamma_j' := \set{\btheta \in \Theta : \bxi_0
        + \btheta \in \CG_j'}.
      \end{equation}
      It follows from the definition of $\CK$ that 
      $$
      Q:= \sum_{j=1}^m \# \Gamma_j 
      \ge
      n(\bxi_0;\CG).
      $$
      Since $\bxi_0 \not \in \CU$, we have that 
      $\phi(\bxi_0,\bxi_0 + \btheta) \le \pi/4$ for all $\btheta \in
      \Gamma_j$. 
       From equations
       \eqref{eq:gsmall}--\eqref{eq:condgsmall}, Lemma
       \ref{lem:increasing}, since
       $\delta \ll \rho^{\alpha -\nu}$ there are $t \ll
      \rho^{-\nu}$ and $Z_0$ independent of $\rho$ such that for all $1 \le j \le
      m$ and $\btheta \in \Gamma_j$,
      \begin{equation}
        \label{eq:overlaponeside}
     \rho^\alpha - Z_0 \delta \le g_{j}\left(\left( 1 - t\right)\bxi_0
     + \btheta \right) \le \rho^\alpha - \delta,
      \end{equation}
      and
      \begin{equation}
        \label{eq:overlapotherside}
      \rho^\alpha + \delta \le   g_{j}\left(\left( 1 + t \right)\bxi_0
      + \btheta \right) \le \rho^\alpha +  Z_0 \delta.
      \end{equation}
      It is clear that if \eqref{eq:overlaponeside} and
      \eqref{eq:overlapotherside} hold, they also hold replacing $Z_0$ with any
      $Z > Z_0$. The precise value we assign to $Z$ might change as the proof goes along, but it will remain
      independent of $\rho$ and $\delta$. 
  We denote by $\CJ$ the radial interval
  $$
  \CJ:= [(1-t) \bxi_0,(1+t)\bxi_0].
  $$We now restrict ourselves to consider
      in \eqref{eq:overlap} only those $\bk \in \CJ \mod \Theta$. This
      corresponds, at the level of Besicovitch space, to the study of the
      operators
      \begin{equation}
        \BA^\CJ := \BA P_{\CJ + \Theta},
      \end{equation}
      where the projection $P_{\CJ+\Theta}$ is defined in Proposition
      \ref{prop:normorder}.
      We also define the projection $P_{\text{good}}$ acting on the basis
      elements $e_{\bxi} \otimes v_j$ as
      $$
    P_{\text{good}} e_{\bxi} \otimes v_j := (P_{\CJ + \Gamma_j} e_{\bxi}) \otimes
    v_j.
      $$
      From Lemma \ref{lem:increasing} and its proof, for all $\bxi \in \CJ$, and
      all $\btheta \in \Gamma_j$, we have that $\bxi
      + \btheta \in \CB_j'$. It therefore follows from \eqref{eq:nrinv} and that
      $P_{\text{good}}$ commutes with $\BA$ and $\BA^{\CJ}$. We decompose
      \begin{equation}
        \label{eq:decomp}
        \BA^\CJ = \BA_{\text{good}} + \BA_{\text{bad}},
      \end{equation}
      where
      \begin{equation}
        \BA_{\text{good}}   := \BA^\CJ P_{\text{good}} \qquad \text{and} \qquad
        \BA_{\text{bad}} = \BA^\CJ(\Id
        - P_{\text{good}}).
      \end{equation}
      For $\bxi \in \CJ$, denote by $\bmu := \set{\mu_p(\bxi):
      1 \le p \le Q}$ the eigenvalues of the
      operators $\BA_{\text{good}}(\set{\bxi})$, and
      $\hat \bnu := \set{\hat \nu_p (\bxi) : p \in \N}$ the eigenvalues of
      $\BA_{\text{bad}}(\set{\bxi})$, each ordered nondecreasingly pointwise. It follows that
      for every $\bxi \in \CJ$, we have the equality of multisets
      $$
      \set{\lambda_p(\BA(\set \bxi)) : p \in \N} = \bmu \cup \hat \bnu.
      $$
      Furthermore, for all $\nu \in \hat \bnu$ corresponds $(j,\btheta)$ such
      that either $\bxi_0 + \btheta \in \CZ_j'$ or $\bxi_0 + \btheta \not \in
      \CA_j'$. By Lemma \ref{lem:increasing}, for $1 \le j \le m$ and $\btheta \in
      \Gamma_j$, the functions $g_j(\bxi + \btheta)$ are continuous and
      increasing for $\bxi \in
      \CJ$. Since the functions $\mu_p$ are obtained simply by ordering the
      functions $g_j(\bxi + \btheta)$ for $1 \le j \le m$ and $\theta \in
      \Gamma_j$ in nondecreasing
      order at every point $\set{\bxi}$, they are
      therefore themselves continuous and increasing on $\CJ$.  
      
      Let us now consider the functions $\hat \nu_\bp$. Notice that they are formed by
      removing $\#\Gamma$ continuous branches from a family of continuous functions. They
      are therefore all continuous on $\CJ$, and \cite[Theorem 3.6]{ParSob2010}
      also applies to them. In particular, one can choose $Z \ge Z_0$ large enough
      so that all functions $\nu_\bp$ such that
      $\nu_\bp(\bxi_0) \le \rho^\alpha - Z\delta/2$, respect the bound $\nu_\bp(\bxi) <
      \rho^\alpha - \delta$ for all $\bxi \in \CJ$. Similarly, all functions
      $\nu_\bp$ such that
      $\nu_\bp(\bxi_0) > \rho^\alpha + Z\delta/2$ respect $\nu_\bp(\bxi) >
      \rho^\alpha + \delta$ for all $\bxi \in \CJ$. We therefore define the sets
      \begin{equation}
        \btau := \set{\nu \in \hat \bnu : \abs{\nu(\bxi_0) - \rho^\alpha}
        \ge \frac{Z\delta}{2}}
      \end{equation}
      and $\bnu = \hat \bnu \setminus \btau$. In particular, for all $\nu \in
      \bnu$ there corresponds $(j,\btheta)$ such that in $\bxi_0 + \btheta \in
      \CA_j'$ and therefore in $\CZ_j'$.
      
      By construction, we have
      indeed that at every $\bk \in [\bk_1,\bk_2]$ the images of the families
      $\bmu, \bnu, \btau$ are the eigenvalues of $\BA(\bk)$ counted with
      multiplicity, satisfying therefore \eqref{eq:firstbp}. Hypothesis
      \eqref{eq:mucrossenough} is also satisfied since all $\mu \in \bmu$ are
      increasing, and for all $\mu \in \bmu$, either
      \begin{equation}
        \label{eq:satisfied}
        \mu((1 - t)\bxi_0) < \rho^\alpha
      - \delta,\qquad \text{or} \qquad \mu( (1+t)\bxi_0) > \rho^\alpha + \delta.
    \end{equation}
      It follows from our choice of $t$ in
      \eqref{eq:overlaponeside}--\eqref{eq:overlapotherside} that there are at
      least $Q
      \ge n\left(\bxi_0;\CG \right)$ of them so that both inequalities in
      \eqref{eq:satisfied}. On the other hand, all
      $\nu \in \bnu$ correspond to some $(j,\btheta) \in \tilde \Theta$
      such that $\bxi_0 + \btheta \in \CZ'$, hence there's at most 
      $m n(\bxi_0;\CZ') < n(\bxi_0;\CG)$ of them, hence hypothesis
      \eqref{eq:secondbp} holds. Finally, we constructed
      $\btau$ explicitly so that for all $\tau \in \btau$,
      hypothesis \eqref{eq:thirdbp} holds.

      Proposition \ref{prop:counting} is therefore true, and we conclude that we
      have at least $\zeta(\rho^\alpha;\BA)>\delta$. In view of our choice of
      $\delta$ this also gives us an overlap exponent at least
      \begin{equation}
        S \ge \min\set{\alpha - \nu, s - \eps}
      \end{equation}
      where $\eps > 0$ is arbitrary and $s$ is given in \eqref{eq:propoverlap}.
\end{proof}
\section{The Dirac Operator} \label{sec:dirac}

In this section, we aim to get conditions on perturbations of the Dirac operator
so that the gauge transform and, more importantly, all the theorems from Part II
can be applied. Basic facts and theorem on the Dirac operator
are found in \cite{gilbertmurray,thaller}. We 
consider Dirac operators built through Clifford algebras, of which the usual
two- and three-dimensional cases are examples. We are then able to
explicitly describe perturbations to which we can apply the gauge transform
method and recover the results of Sections \ref{sec:besicovitch}--\ref{sec:cg}.

\subsection{Clifford algebras} \label{sec:clifford}

We give here basic facts about Clifford algebras used to construct the Dirac
operator in the flat setting. They can be found in \cite[Section
7]{gilbertmurray}. Let
$\R^{p,q}$ be the euclidean space of dimension $p+q$ equipped with the canonical
quadratic form $\eta$ of signature $(p,q)$. In our applications, we consider only the cases
$\R^{0,d}$ (Euclidean) and $\R^{1,d}$ (Minkowski). We denote their orthonormal
bases respectively
$\set{\bv_1,\dotsc,\bv_d}$ and $\set{\bv_0,\bv_1,\dotsc,\bv_d}$.
Consider the algebra $\FA_{p,q}$ generated by
$\set{1,\bv_1,\dotsc,\bv_d}$ or $\set{1,\bv_0,\dotsc,\bv_d}$ with the
relations 
\begin{align}
  \label{eq:anticommute}
  \bv_j \bv_k + \bv_k \bv_j = - 2\eta_{jk}.
\end{align}
It is easy to see that $\FA_{p,q}$ has dimension $2^{p+q}$.
 For
 any subset $S := \set{s_1,\dotsc,s_k} \subset \set{0,\dotsc,d}$ (or of
 $\set{1,\dotsc,d}$ in the euclidean setting), we denote
by $\bv_S$ the element $\bv_{s_1} \cdots \bv_{s_k} \in \FA_{p,q}$, where by
convention
$\bv_\varnothing = 1$. The Clifford algebra on $\R^{p,q}$ is isomorphic to the exterior
algebra $\Lambda^*(\R^d)$. 

From the anticommutation relation \eqref{eq:anticommute}, we deduce
that each pair of the $2^{p+q}$ generators of $\FA_{p,q}$ either commutes or
anticommutes, according to the rule
\begin{equation}
  \begin{cases}
    \bv_j \bv_S = (-1)^{\abs S} \bv_S \bv_j & \text{if } j \not \in S, \\
    \bv_j \bv_S = (-1)^{\abs S - 1} \bv_S \bv_j & \text{if } j  \in S. \\
\end{cases}
\end{equation}
When $p+q$ is even, there is a faithful representation of $\FA_{p,q}$ acting on the spinor
space $\C^{2^{(p+q)/2}}$. A specific representation by matrices constructed
recursively is given in
\cite{Upmeier} in the Euclidean and Minkowski cases. This representation
$\gamma$ has the
property that for all $1 \le j \le d$, the matrix $\gamma_j := \gamma(\bv_j)$ is
skew-hermitian and squares to $-\Id_{p+q}$, $\gamma_0 := \gamma(\bv_0)$ is
hermitian and squares to the identity, and there is some $\abs c = 1$ so that the grading operator $\Gamma := c \prod_j \gamma_j$ is a diagonal
matrix of the form
\begin{equation}
  \Gamma = \begin{pmatrix}
    \Id_{(p+q)/2} & 0 \\
    0 & - \Id_{(p+q)/2}
  \end{pmatrix}.
\end{equation}
We can observe that for all $j$, $\Gamma \gamma_j = -
\gamma_j \Gamma$. The operator $\Gamma$ is called “grading” because it induces a
$\Z_2$ grading on $\FA_{p,q}$. The even subalgebra of $\FA_{p,q}$ consists of
all the elements commuting with $\Gamma$, while the odd subspace consists of all
the anti-commuting elements. In particular, all the $\gamma_j$ are in the odd
subspace, which is characterised as a product of an odd number of generators,
while the even subalgebra is characterised as products of even number
of generators.

\begin{lem}
  \label{lem:anticom}
  Let $\gamma$ be an element of the odd subspace. Then, as a matrix it has the
  form
  \begin{equation}
    \gamma := \begin{pmatrix}
      \boldsymbol 0 & \BX \\
      \BY & \boldsymbol 0
    \end{pmatrix},
  \end{equation}
  where each of the blocks is a $m/2 \times m/2$ matrix.
\end{lem}
\begin{proof}
  This follows from a simple computation of the relation $\gamma \Gamma +
  \Gamma \gamma = 0$ on the matrix elements.
\end{proof}
The representation $\gamma$ also allows us to see that as a $C^*$-algebra, $\FA_{p,q}$ is naturally isomorphic to an algebra of operators on a
Hilbert space $\FS_{p+q}$, called the \emph{spinor space}. When $m$ is even, we have
that $\FS_{p+q} \cong \C^{2^{(p+q)/2}}$. Therefore, setting $m = 2^{(p+q)/2}$,
we can use this representation to obtain operators in $\BS^\infty_m$.

\subsection{Dirac operators}

We define (spatial) Dirac operators differently depending on whether the number of spatial
dimensions is even or odd.
\begin{defi}
  Let $d$ be odd. The $d$-dimensional free Dirac operator $\BA_d$ is the first order
  system acting on spinors in $\RL^2(\R^d;\C^m)$, for $m = 2^{\frac{d+1}{2}}$ given
  by
  \begin{equation}
    \BA_d = \sum_{j=1}^d \gamma_j \del_j,
  \end{equation}
  where the $\gamma_j$ are given by the representation of $\FA_{1,d-1}$ in
  $\CL(\C^{m})$.
\end{defi}
\begin{defi}
  Let $d$ be even. The $d$-dimensional free Dirac operator $\BA_d$ is the first
  order system acting on spinors $\RL^2(\R^d;\C^m)$, for $m = 2^{\frac d 2}$
  given by
  \begin{equation}
    \BA_d = \sum_{j=1}^d \gamma_j \del_j,
  \end{equation}
  where the $\gamma_j$ are given by the representation of $\FA_{0,d}$ in
  $\CL(\C^{m})$.
\end{defi}

It is easy to see in both cases that $\BA_d^2 = - \Delta \Id_{m}$. 
\begin{exa}
  The two-dimensional Dirac operator with mass $M$ is given in \cite[Equation 1.14]{thaller} 
   as
  \begin{equation}
    \BA_{2,M} = -i\left( \sigma_1 \del_{x_1} + \sigma_2 \del_{x_2} \right) + \sigma_3
    M,
  \end{equation}
where $\sigma_1, \sigma_2,
\sigma_3$ are the Pauli matrices
\begin{align}
\sigma_1=\begin{pmatrix}
0 & 1\\
1 & 0
\end{pmatrix},\ 
\sigma_2=\begin{pmatrix}
0 & -i\\
i & 0
\end{pmatrix},\
\sigma_3=\begin{pmatrix}
1 & 0\\
0 & -1
\end{pmatrix}.
\end{align}
It is a perturbation of order $0$ of the free Dirac operator. Indeed, the Pauli
matrices can be used for a representation of the Clifford algebra $\FA_{0,2}$,
and $\sigma_3$ corresponds to the grading operator $\Gamma$.
\end{exa}

\begin{exa}
  The three-dimensional Dirac operator with mass $M$ from \cite[Equation 1.11]{thaller} given
  by 
  \begin{equation}
    \BA_{3,M} = -i\left( \gamma_1 \del_{x_1} + \gamma_2 \del_{x_2}  + \gamma_3
    \del_{x_3}\right) + \Gamma M
  \end{equation}
  is also a perturbation of order $0$ of the free Dirac operator. Here, the
  matrices $\gamma_j$ are the Dirac $\gamma$-matrices used as a representation
  of $\FA_{1,3}$, and our notation generalises this notion, following
  \cite{Upmeier}.
\end{exa}

We now show
that the operators $\BA_d$ are elliptic in the sense of Section
\ref{sec:perturbation}.

\begin{prop}
  \label{prop:diagdirac}
  Let $m:= m(d)$ be the dimension of the spinor space on which $\BA_d$ acts. The
  operator $\BU \in \BS^0_m$ with symbol
  \begin{equation}
    \label{eq:explicitunitary}
    \bu(\bx,\bxi) := \frac{\bone_{\set{\abs \bxi \ge 1}}(\bxi)}{\sqrt 2}\left(\Id_m + \frac{i}{\abs \bxi} \Gamma
    \sum_{j=1}^d \bxi_j \gamma_j\right)+
    \bone_{\set{\abs{\bxi} < 1}}(\bxi)\Id_m
  \end{equation}
  is unitary. Furthermore, $\BU \BA_d \BU^* \in \BD\BE\BS_m^1$ and there is $\BR
  \in \BS_m^{-\infty}$ such that the symbol of $\BU \BA_d \BU^* - \BR$ is
  $\abs{\bxi} \Gamma$. 
\end{prop}
\begin{proof}
  The symbol of the adjoint of $\BU$ is given, following \eqref{eq:adjoint symbol}, by
  \begin{equation}
    \bu^\dagger(\bx,\bxi) =  \frac{\bone_{\set{\abs \bxi \ge 1}}(\bxi)}{\sqrt
    2}\left(\Id_m - \frac{i}{\abs \bxi} \Gamma
    \sum_{j=1}^d \bxi_j \gamma_j\right) + \bone_{\set{\abs{\bxi} < 1}}(\bxi)
    \Id_m
  \end{equation}
  and we can compute that 
  \begin{equation}
    \begin{aligned}
      \left[\bu\circ \bu^\dagger \right] (\bxi) &= \frac{\bone_{\set{\abs \bxi
      \ge 1}}(\bxi)}{2}\left(\Id_m -
        \frac{1}{\abs{\bxi}^2} \sum_{j,k = 1}^d \Gamma^2 \gamma_j \gamma_k \bxi_j
      \bxi_k\right) + \bone_{\set{\abs \bxi < 1}}(\bxi) \Id_m  \\
      &= \frac{\bone_{\set{ \abs \bxi \ge 1}}(\bxi)}{2} \left(\Id_m - \frac{1}{\abs{\bxi}^2} \sum_j \gamma_j^2
      \bxi_j^2\right) + \bone_{\set{\abs \bxi < 1}}(\bxi) \Id_m\\
      &= \Id_m.
  \end{aligned}
  \end{equation}
 In a very similar fashion, we see that the symbol of $\BU \BA_d \BU^*$ is given by
 \begin{equation}
   [\bu \circ \ba_d \circ \bu^\dagger](\bxi) = \bone_{\set{\abs \bxi \ge 1
   }}(\bxi) \abs \bxi\Gamma  + \bone_{\set{\abs \bxi < 1}}(\bxi) \ba_d(\bxi).
 \end{equation}
 This proves our claim where $\BR \in \BS_m^{-\infty}$ has symbol 
 \begin{equation}
   \br(\bxi)= \bone_{\set{\abs \bxi < 1}} (\bxi) (\ba_d(\bxi) - \abs \bxi
   \Gamma).
 \end{equation}
\end{proof}

We now see that for $d = m = 2$, the operators $\BA_2 + \BB$, $\BB \in \BS^\beta_m$,
$\beta < 1$ are unitarily
equivalent to an operator satisfying the hypotheses of Theorems
\ref{thm:aexpconcrete} and Theorem \ref{thm:bs}, which proves that we
generically have a
complete asymptotic expansion for the density of states, and that if $\BB$ is
periodic then $\BA$ has the Bethe--Sommerfeld property. In other words, the
following two theorems are proved, which are more precise reformulations of Theorems \ref{thm:dirac2dids}
and \ref{thm:dirac2dbs}. 

\begin{thm}
  Let $\beta < 1$ and $\BA = \BA_2 + \BB$, where $\BB \in \BS_2^\beta$ satisfies
  the generic conditions \textbf{A}, \textbf B and \textbf C. Then, for every $K
  > - 2$ there is a finite set $L \subset (0,2+K)$ so that for every $j \in L
  \cup \set{0}$ there are constants $C_{j}^\pm$, $C_{j,\log}^\pm$ such that
  \begin{equation}
    N^\pm(\BA; \lambda) = C_0^\pm \lambda^2 + \sum_{j \in L} \left( C_j^\pm
    \lambda^{2 - j} + C_{j,\log}^\pm \lambda^{2-j} \log \lambda \right) +
    \bigo{\lambda^{-K}}
  \end{equation}
  as $\lambda \to \infty$. 
\end{thm}

\begin{thm}
  Let $\beta < 1 $ and $\BA = \BA_2 + \BB$, where $\BB \in \BS_2^\beta$ is
  periodic. Then, $\BA$ has the Bethe--Sommerfeld property, i.e. there exists
  $\lambda_0 > 0 $ such that the spectrum of $\BA$ contains intervals $(-\infty,
  -\lambda_0]$ and $[\lambda_0,\infty)$. 
\end{thm}

We now want to address the
question of the perturbations that are allowed whenever $d \ge 3$.

\begin{prop}
  \label{prop:permittedperturbations}
  For $\beta < 1$, and $0 \le j \le d$ (with $0$ omitted when $d$ is even), let
  $B_{\Id}, B_\Gamma, B_j \in \BS_1^\beta$ be scalar pseudo-differential
  operators of order
  $\beta$, and put
  \begin{equation}
    \BB = B_{\Id} \Id_m + B_\Gamma \Gamma + \sum_j B_j \gamma_j.
  \end{equation}
  Then, there are operators $\BB' \in \BU\BS^\beta_m$, $\BR \in \BS_m^{\beta -1
  }$ and $\tilde \BB \in
  \BS_m^\beta$ whose symbol has image in the odd subspace of $\FA_{p,q}$ such
  that
  \begin{equation}
    \label{eq:afterunit}
    \BU (\BA_d + \BB) \BU^* = \Op(\abs \bxi) \Gamma + \BB' + \tilde \BB + \BR
  \end{equation}
\end{prop}
\begin{proof}
  The unitary operator $\BU$ from \eqref{eq:explicitunitary} can be written as
  \begin{equation}
    \BU = \frac{1}{\sqrt 2} \left(\Id_m + \sum_{j=1}^d U_j \Gamma \gamma_j\right)
    \mod \BS_m^{-\infty}.
    \label{eq:Udecomposed}
  \end{equation}
  Here, $U_j \in \BS_1^0$ are scalar pseudo-differential operators given by
  \begin{equation}
    U_j = \Op\left(\frac{i \bxi_j \chi(\bxi)}{\abs \bxi}\right),
  \end{equation}
  where $\chi$ is a smooth function supported in $\set{\abs{\bxi} \ge 1/2}$ and
  $\chi(\bxi) \equiv 1$ for all $\abs \bxi \ge 3/4$. We now compute $\BU B_\gamma \gamma \BU^*$ for different values of $\gamma$.
  All the sums range from $1$ to $d$ with additional restrictions, we have only
  written the restrictions to make notation lighter. For $1 \le j \le
  d$
  we have
  \begin{equation}
    \label{eq:unittrans}
    \begin{aligned}
      \BU B_{j} \gamma_j \BU^* &= \frac{1}{2}\Bigg(B_{j} \gamma_j + \sum_{k
          \neq j}[U_k;B_{j}]\Gamma \gamma_k \gamma_j -
        (U_j B_{j} + B_{j} U_j) \Gamma + \sum_{k} (U_k
        B_{j} U_j + U_j B_j U_k) \gamma_k - \\
          &\qquad \quad - \sum_{k} U_k B_{j} U_k \gamma_j -
          \sum_{\substack{\ell \neq j \\ k \neq j \\ k < \ell}}\left( [U_\ell;B_{j}
          U_k] + [B_{j};U_k]U_\ell \right) \gamma_k \gamma_\ell
    \gamma_j
  \bigg) \mod \BS_m^{- \infty}.
\end{aligned}
  \end{equation}
  Let us have a careful look at each of the six terms in Equation \eqref{eq:unittrans}. The second
  and the last terms involve commutators of operators with scalar-valued symbols,
  they are in $\BS_m^{\beta - 1}$ and we put $\BR_{j}$ as their sum. The third term is
  in $\BU\BS^\beta_m$, and we denote it $\BB'_j$. Finally, the first,
  fourth and fifth term are readily seen to have symbols in the odd subspace, we
  put $\tilde \BB_{j}$ as their sum.

  The operator $\BU B_{0} \gamma_0 \BU^*$ is computed similarly as in
  \eqref{eq:unittrans} with some of the terms vanishing. It is given by
  \begin{equation}
    \begin{aligned}
    \BU B_0 \gamma_0 \BU^* = \frac 1 2 \Bigg( B_0 \gamma_0 &+ \sum_{k} [U_k;B_0]
    \Gamma \gamma_k \gamma_0 - \sum_k U_k B_0 U_k \gamma_0 -   \\ & \qquad - \sum_{k<\ell}
  ([U_\ell;B_0 U_k] + [B_0;U_k]U_\ell) \gamma_k \gamma_\ell \gamma_0\Bigg) \mod
  \BS_m^{-\infty}.
\end{aligned}
  \end{equation}
  The first and third term have image in the odd subspace, we put $\tilde
  \BB_0$ as their sum. The second and
  last terms involve commutators of operators with scalar-valued symbols, as
  such they are in $\BS_m^{\beta - 1}$ and we put $\BR_0$ as their sum. We note
  that there are no uncoupled terms.

  The operator $\BU B_\Gamma \Gamma \BU^*$ is given by
  \begin{equation}
    \label{eq:Gammatrans}
    \begin{aligned}
      \BU B_\Gamma\Gamma \BU^* &= \frac{1}{2}\Bigg(B_\Gamma \Gamma - \sum_{k}(U_k B_\Gamma +
          B_\Gamma
        U_k) \gamma_k - \\
    &\qquad \quad - \sum_{k} U_k B_\Gamma U_k  \Gamma +
  \sum_{k < \ell}\left( [U_\ell;B_\Gamma U_k] + [B_\Gamma;U_k]U_\ell \right)\Gamma
    \gamma_\ell \gamma_k
  \bigg) \mod \BS_m^{-\infty}.
\end{aligned}
\end{equation}
This time, the first and third terms are seen to be in $\BU\BS^\beta_m$ and we
put their sum as $\BB'_\Gamma$. The
second term has symbol in the odd subspace and we denote it by $\tilde
\BB_\Gamma$. The last term can be seen to be in $\BS_m^{\beta - 1}$ and we
denote it by $\BR_{\Gamma}$.

Finally, the operator $\BU B_{\Id} \Id_m \BU^*$ is given by
  \begin{equation}
    \label{eq:Idtrans}
    \begin{aligned}
      \BU B_{\Id} \Id_m \BU^* &= \frac{1}{2}\Bigg(B_{\Id} \Id_m +
        \sum_{k}[U_k;B_{\Id}]\Gamma \gamma_k 
        + \sum_{k} U_k B_{\Id} U_k  \Id_m \\
    &\qquad \quad +
    \sum_{k < \ell}\left( [U_\ell;B_{\Id}U_k] + [B_{\Id};U_k]U_\ell \right)
    \gamma_\ell \gamma_k
  \Bigg) \mod \BS_m^{-\infty}.
\end{aligned}
\end{equation}
This time, we see that the first and third terms are in $\BU\BS_m^\beta$, we put
their sum as $\BB'_{\Id}$, while
the second and last terms are in $\BS_m^{\beta - 1}$ and we put their sum as
$\BR_{\Id}$. 

Finally, put $\tilde \BR \in \BS_m^{-\infty}$ as the sum of the remainders
$\mod \BS_m^{-\infty}$ obtained at every step. Combining all our computations
and Proposition \ref{prop:diagdirac} gives us that \eqref{eq:afterunit}
holds with
\begin{equation}
  \begin{aligned}
    \BB' &= \BB'_{\Id} + \BB'_{\Gamma} + \sum_{j=1}^d \BB'_j \\
    \tilde \BB &= \tilde \BB_{\Id} + \tilde \BB_{\Gamma} + \sum_{j=0}^d \tilde
    \BB_j \\
    \BR &= \tilde \BR + \BR_{\Id} + \BR_\Gamma + \sum_{j=0}^d \BR_j.
  \end{aligned}
\end{equation}
\end{proof}

The next theorem follows and includes Theorem \ref{thm:dirac3dids} as a
special case when $d = 3$.
\begin{thm}
  Let $m(d)$ be the dimension of the spinor space on which $\BA_d$ acts. For
  $\beta \le 1/2$ and $0 \le j \le d$ (with $0$ omitted when $d$ is even) let
  $B_\Gamma, B_j, B_{\Id} \in \BS^\beta$ be scalar
  pseudo-differential operators satisfying Conditions \ref{condI}--\ref{condIV}, and put
  \begin{equation}
    \BB = B_{\Id} \Id_m + B_\Gamma \Gamma + \sum_{j=0}^d B_j \gamma_j,
  \end{equation}
  and $\BA = \BA_d + \BB$. Then, putting $\gamma^* = \max \set{\beta - 1, 2
  \beta - 1}$, there exists a finite set $L \subset (0, 1 - \gamma^*)$ and constants $C_0^\pm$ and $C_{j,q}^\pm$, $0 \le q \le d-1$, $j \in
  L$ such that
  \begin{equation}
    N^\pm(\BA;\lambda) = C_0^\pm \lambda^d + \sum_{j \in L} \sum_{q = 0}^{d-1}
    C_{j,q}^\pm \lambda^{d-j} \log^q \lambda + \bigo{\lambda^{d-1 + \gamma^*}}
  \end{equation}
  as $\lambda \to \infty$. 
\end{thm}

\begin{proof}
  It follows from Proposition \ref{prop:permittedperturbations} that $\BU \BA \BU^*$ satisfies
  the hypotheses of Theorem \ref{thm:aexpconcretecut} with $\gamma^* =
  \max\set{\beta - 1, 2 \beta - 1}$. In particular, the
  restricted asymptotics of the IDS given in that theorem are true for such
  operators with $\alpha = 1$.
\end{proof}

Finally, in some highly non-generic cases we can get complete asymptotic
expansions and the Bethe--Sommerfeld property for $d$-dimensional Dirac
operators with $d \ge 3$. We state both results and observe that they follow
directly from the fact that after conjugation by $\BU$, these operators are
uncoupled.

\begin{thm}
  Let $m(d)$ be the dimension of the spinor space on which $\BA_d$ acts, $\beta
  < 1$ and $\BB \in \BU\BS_m^\beta$ satisfying Conditions
  \ref{condI}--\ref{condIV}. Put $\BA = \BA_d + \BU^* \BB \BU$. Then, $N^\pm(\BA;\lambda)$
  satisfies the complete asymptotic expansion \eqref{eq:asymptoticsconcrete}
  with $\alpha = 1$.
\end{thm}

\begin{thm}
  Let $m(d)$ be the dimension of the spinor space on which $\BA_d$ acts, $\beta
  < 1$ and $\BB \in \BU\BS_m^\beta$ be periodic. Put $\BA = \BA_d + \BU^* \BB
  \BU$. Then, $\BA$ has the Bethe--Sommerfeld property.
\end{thm}

\bibliographystyle{alpha}
\bibliography{bibliography}

\begin{thebibliography}{{Shu}79a}

\bibitem[BP09]{barbpar}
G.~Barbatis and L.~Parnovski.
\newblock Bethe-{S}ommerfeld conjecture for pseudodifferential perturbation.
\newblock {\em Comm. Partial Differential Equations}, 34(4-6):383--418, 2009.

\bibitem[BS87]{BirmanSolomyak}
M.~Birman and M.~Solomjak.
\newblock {\em Spectral Theory of Self-Adjoint Operators in Hilbert Space}.
\newblock Mathematics and its Applications. Springer Netherlands, 1987.

\bibitem[CMS73]{CoburnMoyerSinger}
L.~A. Coburn, R.~D. Moyer, and I.~M. Singer.
\newblock {$C\sp*$}-algebras of almost periodic pseudo-differential operators.
\newblock {\em Acta Math.}, 130:279--307, 1973.

\bibitem[Dix81]{Dixmier}
J.~Dixmier.
\newblock {\em Von {N}eumann {A}lgebras}, volume~27 of {\em North-Holland
  Mathematical Library}.
\newblock North-Holland Publishing Co., Amsterdam-New York, 1981.
\newblock With a preface by E. C. Lance, Translated from the second French
  edition by F. Jellett.

\bibitem[GM91]{gilbertmurray}
J.~E. Gilbert and M.~A.~M. Murray.
\newblock {\em Clifford algebras and {D}irac operators in harmonic analysis},
  volume~26 of {\em Cambridge Studies in Advanced Mathematics}.
\newblock Cambridge University Press, Cambridge, 1991.

\bibitem[H{\"{o}}r07]{HormanderIII}
L.~H{\"{o}}rmander.
\newblock {\em The analysis of linear partial differential operators. {III}}.
\newblock Classics in Mathematics. Springer, Berlin, 2007.
\newblock Pseudo-differential operators, Reprint of the 1994 edition.

\bibitem[Ivr19]{Ivrii2018}
V.~Ivrii.
\newblock {\em Complete Semiclassical Spectral Asymptotics for Periodic and
  Almost Periodic Perturbations of Constant Operators}, pages 583--606.
\newblock Springer, 2019.

\bibitem[Kuc93]{Kuchment}
P.~Kuchment.
\newblock {\em Floquet theory for partial differential equations}, volume~60 of
  {\em Operator Theory: Advances and Applications}.
\newblock Birkh\"{a}user Verlag, Basel, 1993.

\bibitem[MPS14]{MorParSht2014}
S.~Morozov, L.~Parnovski, and R.~Shterenberg.
\newblock Complete asymptotic expansion of the integrated density of states of
  multidimensional almost-periodic pseudo-differential operators.
\newblock {\em Ann. Henri Poincar\'e}, 15(2):263--312, 2014.

\bibitem[Na{\u{\i}}72]{Naimark}
M.~A. Na{\u{\i}}mark.
\newblock {\em Normed algebras}.
\newblock Wolters-Noordhoff Publishing, Groningen, third edition, 1972.
\newblock Translated from the second Russian edition by Leo F. Boron,
  Wolters-Noordhoff Series of Monographs and Textbooks on Pure and Applied
  Mathematics.

\bibitem[Par08]{parnovski}
L.~Parnovski.
\newblock Bethe-{S}ommerfeld conjecture.
\newblock {\em Ann. Henri Poincar\'{e}}, 9(3):457--508, 2008.

\bibitem[PS10]{ParSob2010}
L.~Parnovski and A.~V. Sobolev.
\newblock Bethe-{S}ommerfeld conjecture for periodic operators with strong
  perturbations.
\newblock {\em Invent. Math.}, 181(3):467--540, 2010.

\bibitem[PS12]{ParSht2012}
L.~Parnovski and R.~Shterenberg.
\newblock Complete asymptotic expansion of the integrated density of states of
  multidimensional almost-periodic {S}chr\"odinger operators.
\newblock {\em Ann. of Math. (2)}, 176(2):1039--1096, 2012.

\bibitem[PS16]{ParSht2016}
L.~Parnovski and R.~Shterenberg.
\newblock Complete asymptotic expansion of the spectral function of
  multidimensional almost-periodic {S}chr\"odinger operators.
\newblock {\em Duke Math. J.}, 165(3):509--561, 2016.

\bibitem[PS19]{ParSht2019}
L.~Parnovski and R.~Shterenberg.
\newblock Perturbation theory for almost-periodic potentials {I}:
  one-dimensional case.
\newblock {\em Comm. Math. Phys.}, 366(3):1229--1257, 2019.

\bibitem[Roz78]{Rozenbljum1978}
G.~V. Rozenbljum.
\newblock Near-similarity of operators and the spectral asymptotic behavior of
  pseudodifferential operators on the circle.
\newblock {\em Trudy Moskov. Mat. Obshch.}, 36:59--84, 294, 1978.

\bibitem[Shu78]{Shubin1978}
M.~A. Shubin.
\newblock Almost periodic functions and partial differential operators.
\newblock {\em Uspehi Mat. Nauk}, 33(2):3--47, 247, 1978.

\bibitem[{Shu}79a]{Shubin1979a}
M.~A. {Shubin}.
\newblock {Pseudodifferential almost-periodic operators and von Neumann
  algebras.}
\newblock {\em {Trans. Mosc. Math. Soc.}}, 35:103--166, 1979.

\bibitem[Shu79b]{Shubin1979}
M.~A. Shubin.
\newblock Spectral theory and the index of elliptic operators with
  almost-periodic coefficients.
\newblock {\em Uspekhi Mat. Nauk}, 34(2):95--135, 1979.

\bibitem[Skr85]{Skriganov}
M.~M. Skriganov.
\newblock Geometric and arithmetic methods in the spectral theory of
  multidimensional periodic operators.
\newblock {\em Trudy Mat. Inst. Steklov.}, 171:122, 1985.

\bibitem[Sob05]{Sobolev2005}
A.~V. Sobolev.
\newblock Integrated density of states for the periodic {S}chr\"{o}dinger
  operator in dimension two.
\newblock {\em Ann. Henri Poincar\'{e}}, 6(1):31--84, 2005.

\bibitem[Sob06]{Sobolev2006}
A.~V. Sobolev.
\newblock Asymptotics of the integrated density of states for periodic elliptic
  pseudo-differential operators in dimension one.
\newblock {\em Rev. Mat. Iberoam.}, 22(1):55--92, 2006.

\bibitem[Tay11]{TaylorII}
M.~E. Taylor.
\newblock {\em Partial differential equations {II}. {Q}ualitative studies of
  linear equations}, volume 116 of {\em Applied Mathematical Sciences}.
\newblock Springer, New York, second edition, 2011.

\bibitem[Tha92]{thaller}
B.~Thaller.
\newblock {\em The {D}irac equation}.
\newblock Texts and Monographs in Physics. Springer-Verlag, Berlin, 1992.

\bibitem[Upm02]{Upmeier}
H.~Upmeier.
\newblock Dirac operator and real structure on {E}uclidean and {M}inkowski
  spacetime.
\newblock In {\em Noncommutative geometry and the standard model of elementary
  particle physics ({H}esselberg, 1999)}, volume 596 of {\em Lecture Notes in
  Phys.}, pages 136--151. Springer, Berlin, 2002.

\bibitem[Wei77]{Weinstein1977}
A.~Weinstein.
\newblock Asymptotics of eigenvalue clusters for the {L}aplacian plus a
  potential.
\newblock {\em Duke Math. J.}, 44(4):883--892, 1977.

\end{thebibliography}

\end{document}